\documentclass[twocolumn,pre,aps,floatfix,superscriptaddress]{revtex4-1}
\usepackage{amsmath,amssymb,graphicx,times,enumitem,amsthm,mathtools}
\usepackage[dvipsnames]{xcolor}
\usepackage{bibunits}
\usepackage{chngcntr}
\theoremstyle{definition}

\newtheorem{proposition}{Proposition}

\usepackage[normalem]{ulem}
\usepackage{hyperref}

\newcommand{\eps}{\varepsilon}

\begin{document}

\title{Sensitive Dependence of Optimal Network Dynamics on Network Structure}

\author{Takashi Nishikawa}
\email{E-mail: t-nishikawa@northwestern.edu}
\affiliation{Department of Physics and Astronomy, Northwestern University, Evanston, IL 60208, USA}
\affiliation{Northwestern Institute on Complex Systems, Northwestern University, Evanston, IL 60208, USA}

\author{Jie Sun}
\affiliation{Department of Mathematics, Clarkson University, Potsdam, NY 13699, USA}
\affiliation{Department of Physics, Clarkson University, Potsdam, NY 13699, USA}
\affiliation{Department of Computer Science, Clarkson University, Potsdam, NY 13699, USA}
\affiliation{Clarkson Center for Complex Systems Science, Clarkson University, Potsdam, NY 13699, USA}

\author{Adilson E. Motter}
\affiliation{Department of Physics and Astronomy, Northwestern University, Evanston, IL 60208, USA}
\affiliation{Northwestern Institute on Complex Systems, Northwestern University, Evanston, IL 60208, USA}

\date{\today}

\begin{bibunit}

\begin{abstract}
The relation between network structure and dynamics is determinant for the behavior of complex systems in numerous domains.  An important long-standing problem concerns the properties of the networks that optimize the dynamics with respect to a given performance measure.  Here we show that such optimization can lead to {\it sensitive dependence} of the dynamics on the structure of the network. 
Specifically, using diffusively coupled systems as examples, we demonstrate that the stability of a dynamical state can exhibit sensitivity to unweighted structural perturbations (i.e., link removals and node additions) for undirected optimal networks and to weighted perturbations (i.e., small changes in link weights) for directed optimal networks.  As mechanisms underlying this sensitivity, we identify discontinuous transitions occurring in the complement of undirected optimal networks and the prevalence of eigenvector degeneracy in directed optimal networks.  These findings establish a unified characterization of networks optimized for dynamical stability, which we illustrate using Turing instability in activator-inhibitor systems, synchronization in power-grid networks, network diffusion, and several other network processes.  Our results suggest that the network structure of a complex system operating near an optimum can potentially be fine-tuned for a significantly enhanced stability compared to what one might expect from simple extrapolation.  On the other hand, they also suggest constraints on how close to the optimum the system can be in practice.  Finally, the results have potential implications for biophysical networks, which have evolved under the competing pressures of optimizing fitness while remaining robust against perturbations.
\end{abstract}

\maketitle

\section{Introduction}
Building on the classical fields of graph theory, statistical physics, and nonlinear dynamics, as well as on the increasing availability of large-scale network data, the field of network dynamics has flourished over the past 15 years~\cite{Newman2010,Chen:2015}.
Much of the current effort in this area is driven by the premise that understanding the structure, the function, and the relation between the two will help explain the workings of natural systems and facilitate the design of engineered systems with expanded capability, optimized performance, and enhanced 
robustness.
There have been extensive studies on this structure-dynamics relation~\cite{barrat2008dynamical,Porter:2016,Strogatz:2001il} in a wide range of contexts, such as synchronization~\cite{Arenas2008,Belykh:2006qr,PhysRevLett.93.254101,PhysRevE.74.066115,PhysRevLett.110.174102,Nishikawa:2003xr,Pecora:2014zr,Skardal2014,Wiley:2006fk,Restrepo2004,Restrepo2006}; reaction, diffusion, and/or advection dynamics~\cite{Colizza:2007uq,PhysRevLett.110.028701,PhysRevE.89.020801,Youssef:2013}; dynamical stability~\cite{Pomerance:2009fk,Bunimovich2012}; controllability~\cite{PhysRevE.75.046103,Whalen:2015}; and information flow~\cite{PhysRevE.75.036105,Sun2015}.
Many of these studies have led to systematic methods for enhancing the dynamics through network-structural modifications, with examples including network control~\cite{PhysRevE.87.032909,PhysRevE.81.036101,PhysRevLett.96.208701,RisauGusman200952} and synchronization 
enhancement~\cite{Hagberg:2008wd,Hart:2015,Nishikawa:2010fk,Watanabe2010lsg}, where the latter has been demonstrated in applications~\cite{PhysRevLett.110.064106,PhysRevLett.108.214101,PhysRevLett.107.034102}.

\begin{figure*}[htbp]
\begin{center}
\includegraphics[width=0.7\textwidth]{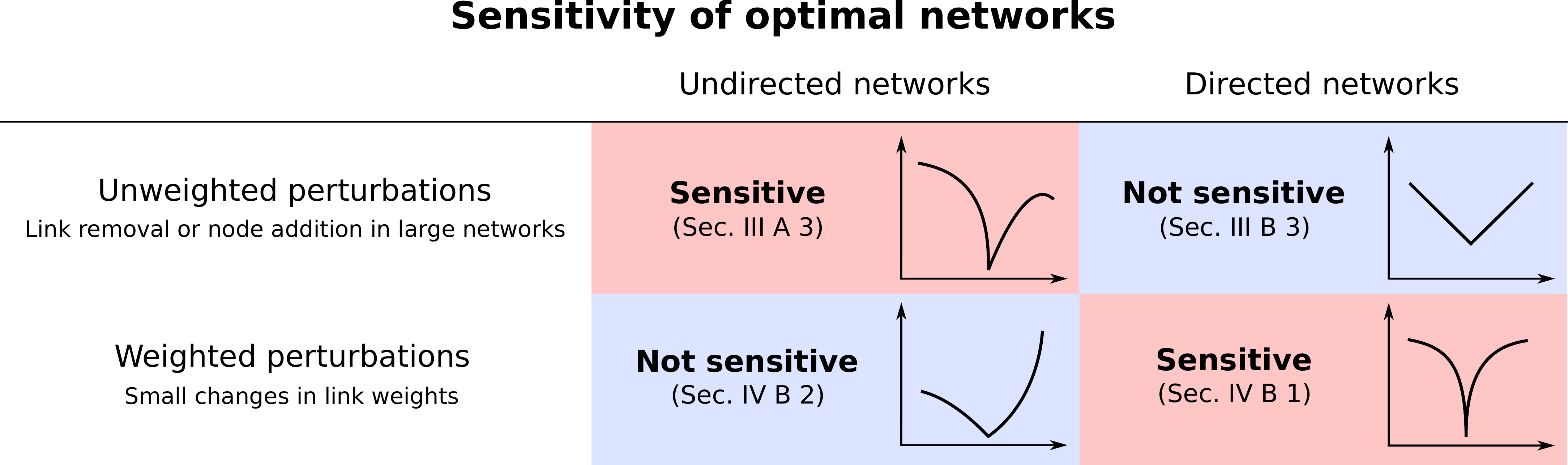}
\end{center}\vspace{-9pt}
\caption{\label{fig:summary}
Directed and undirected networks optimized for the stability of the network dynamics can be sensitive to weighted and unweighted structural perturbations, respectively.
The graphs schematically illustrate typical behavior for systems with sensitivity (in red boxes) and systems with no sensitivity (in blue boxes).
Under weighted perturbations, we actually show that all undirected networks (including non-optimal ones) are nonsensitive (lower blue box).
}
\end{figure*}

A fundamental question at the core of the structure-dynamics relation is that of optimization: which network structures optimize the dynamics of the system for a given function and what are the properties of such networks? The significance of addressing this question is twofold. 
First, 
knowledge of
the properties of optimized structures 
can inform
system architecture design.
For example, in power-grid networks, whose operation requires frequency synchronization among power generators, the structures that maximize 
synchronization stability could 
potentially be used to devise effective strategies for upgrading the system~\cite{Dobson:2007}.
Second, the identification of the network structures that guarantee the best fitness of natural complex systems can provide insights into the mechanisms underlying their evolution.
Examples of such systems include neuronal networks, whose (synaptic) connectivity structure is believed to have been optimized through evolution or learning for categorization tasks~\cite{Yamins10062014}, synchronization efficiency~\cite{Buzsaki:2004uq}, dynamical complexity~\cite{Bullmore:2009fk,Tononi24051994}, information transfer efficiency~\cite{10.1371/journal.pcbi.0030017,Bullmore:2009fk}, and/or wiring cost~\cite{Buzsaki:2004uq}. The question of 
optimizing the network structure 
can be conceptualized as the problem of maximizing or minimizing a measure of dynamical stability, robustness, or performance over all possible configurations of links connecting the dynamical units.

Here we demonstrate that optimized dynamics 
are
often highly sensitive to perturbations applied to the structure of the network. For concreteness, we focus on optimizing 
the linear stability of desired dynamical states over
all networks with a given number of nodes and links. We consider 
network
states in which the (possibly time-dependent) states of the individual nodes are identical across the network, such as 
in consensus dynamics, 
synchronized periodic or chaotic oscillations, and states of equilibrium in diffusion processes. 
We establish conditions under which the stability is sensitive or non-sensitive to structural perturbations, depending on the class of networks and the nature of the perturbations considered, as summarized in Fig.~\ref{fig:summary}.
In particular, we 
show that optimized stability can exhibit sensitivity under different types of perturbations for directed and undirected networks:
\begin{enumerate}[topsep=12pt, partopsep=0pt, leftmargin=12pt]
\item 
Sensitivity to link removals and node additions (unweighted perturbations) for undirected optimal networks in the limit of large network size 
(top left red box in Fig.~\ref{fig:summary}).\\[5pt]
We show that such sensitivity is observed for a class of optimal networks,
which we refer to as \underline{U}niform \underline{C}omplete \underline{M}ultipartite (UCM) 
networks. The UCM networks are composed of node groups of equal sizes that are fully connected to each other but have no internal links.
We prove that these networks
are the only networks that achieve the maximum stability possible for a given number of nodes and links.
The UCM networks are part of a larger class of networks,
characterized as having the \underline{M}inimum possible size of the largest \underline{C}omponents in their \underline{C}omplement (MCC) among all networks with a given number of nodes and links.
We provide a full analytical characterization of the MCC networks of arbitrary finite size and study their behavior as the network size approaches infinity.
\item 
Sensitivity 
to 
changes in
link weights (weighted perturbations) 
for finite-size 
directed optimal
networks
(bottom right red box in Fig.~\ref{fig:summary}).\\[5pt]
While specific examples can be found in the literature~\cite{Golub:2013,Eslami:1994,Chui:1997,Ottino-Loffler:2016}, no systematic study exists on general 
mechanisms and conditions for such sensitivity.
Here we provide such 
conditions
in terms of the spectral degeneracy of the 
network 
by
establishing the scaling relation between the stability and the perturbation size.
These conditions imply
that spectral degeneracy underlies such sensitivity to link-weight perturbations.
We 
expect this
sensitivity to be observed in many applications 
since
spectral degeneracy appears to be common in real networks~\cite{MacArthur:2009}.
Moreover, here we show that optimization tends to increase the incidence of spectral degeneracy, and we also show that the network exhibits approximately the same sensitivity even when the degeneracy (or the optimality) is only approximate.
\end{enumerate}
In addition to these two cases of sensitivity, we have results on the absence of sensitivity in the other two cases (blue boxes in Fig.~\ref{fig:summary}).
We illustrate the implications of our results using a general class of diffusively coupled systems for which the network spectrum is shown to determine the stability and other aspects of the
dynamics to be optimized. The specific cases we analyze include the rate of diffusion over networks, the critical threshold for Turing instability in networks of activator-inhibitor systems, and 
synchronization stability in power grids and in networks of chaotic oscillators.

The remainder of the article is organized as follows.
We first define the class of 
network dynamics
under consideration (Sec.~\ref{sec:class-systems}).
We then
present our results
on the two types of sensitivity anticipated above (Secs.~\ref{sec:opt} and \ref{sec:generic}), followed by 
examples of physical
systems exhibiting these types of sensitivity (Sec.~\ref{sec:phys-examples}).
We conclude 
with a discussion on further
implications of 
our
results (Sec.~\ref{sec:discussion}).

\section{Network dynamics considered}\label{sec:class-systems}

We aim to address a wide range of network dynamics in a unified way. For this purpose we consider the dynamics of 
networks
of coupled dynamical units governed by the following general equation with pairwise interactions:
\begin{equation}\label{eqn:coupled-syst}
\dot{\mathbf{x}}_i 
= \mathbf{F}\bigl(\mathbf{x}_i, \mathbf{H}_{i1}(\mathbf{x}_i,\mathbf{x}_1),\ldots,
 \mathbf{H}_{in}(\mathbf{x}_i,\mathbf{x}_n)\bigr)
\end{equation}
for $i = 1,\ldots,n$, where $n$ is the number of dynamical units (nodes), $\mathbf{x}_i = \mathbf{x}_i(t)$ is the column vector of state variables for the $i$th 
node
at time $t$, and $\dot{\mathbf{x}}_i$ denotes the time derivative of $\mathbf{x}_i$.
The function $\mathbf{F}(\mathbf{x},\mathbf{y}_1,\ldots,\mathbf{y}_n)$ is generally nonlinear and describes how the dynamics of 
node
$i$ 
are
influenced by the other 
nodes
through intermediate variables $\mathbf{y}_j = \mathbf{H}_{ij}(\mathbf{x}_i,\mathbf{x}_j)$, where $\mathbf{y}_j = \mathbf{0}$ indicates no interaction.
This means that the dynamics of an isolated 
node
are
described by $\dot{\mathbf{x}} = \mathbf{F}(\mathbf{x},\mathbf{0},\ldots, \mathbf{0})$.
We assume that the dependence of $\mathbf{F}$ on $\mathbf{y}_j$ is the same for all $j$ (or more precisely, that $\mathbf{F}$ is invariant under any permutation of $\mathbf{y}_1,\ldots,\mathbf{y}_n$).
Thus, the topology of the interaction network and the strength of individual pairwise coupling are not encoded in $\mathbf{F}$, but rather in the $(i,j)$-dependence of the coupling function $\mathbf{H}_{ij}$. This extends the framework introduced in Ref.~\cite{Pecora:1998zp} and can describe a wide range of dynamical processes on 
networks, including consensus protocol~\cite{4140748,4700861}, diffusion over networks~\cite{Newman2010}, emergence of Turing patterns in networked activator-inhibitor systems~\cite{Nakao:2010fk}, relaxation in certain fluid networks~\cite{maas1987transportation}, and synchronization of power generators~\cite{Motter:2013fk} as well as other coupled identical and non-identical 
oscillators~\cite{PhysRevE.61.5080,kuramoto1984chemical,Nishikawa:2006fk,Pecora:1998zp}.
Details on these examples can be found in Supplemental Material~\cite{sm}, Sec.~\ref{si:sec:example-systems}.

For the class of systems described by Eq.~\eqref{eqn:coupled-syst}, we consider {\it network-homogeneous states} given by 
\begin{equation}\label{eqn:sync-cond}
\mathbf{x}_1(t) = \cdots = \mathbf{x}_n(t) = \mathbf{x}^*(t),
\end{equation} 
where $\mathbf{x}^*$ satisfies the equation for an isolated 
node,
$\dot{\mathbf{x}}^* = \mathbf{F}(\mathbf{x}^*,\mathbf{0},\ldots, \mathbf{0})$.
Each of the example systems mentioned above exhibits such a state: uniform agreement in consensus protocols, 
synchronous dynamics
in oscillator networks, uniform occupancy in network diffusion, uniform concentration in coupled activator-inhibitor systems, and the equilibrium state in the fluid networks. Note that certain non-homogeneous states can also be represented 
using
such a solution by changing the frame of 
reference (demonstrated for specific examples of non-uniform phase-locked states in power grids and phase oscillator networks in Supplemental Material~\cite{sm}, Sec.~\ref{si:sec:example-systems}A).

To facilitate the stability analysis, we make two general assumptions on the nature of node-to-node interactions when the system is close to 
a network-homogeneous state.
Assumption~(A-1): The interactions are ``diffusive,'' in the sense that the coupling strength between two nodes, $\mathbf{H}_{ij}(\mathbf{u},\mathbf{v})$, is to first order proportional to the difference between their states, $\mathbf{v} - \mathbf{u}$.
In particular, we assume that the coupling strength vanishes as the node states become equal.
Assumption (A-2): There is a constant coupling matrix $A = (A_{ij})$ encoding the structure of the network of interactions, in the sense that the proportionality coefficient (the ``diffusion constant'') in assumption (A-1) can be written as $A_{ij} \cdot \mathbf{G}(t)$, where the scalar $A_{ij}$ represents the strength of coupling from node $j$ to node $i$, and the matrix-valued function $\mathbf{G}(t)$ is independent of $i$ and $j$.

Under these assumptions, we define a stability function $\Lambda(\alpha)$ for each complex-valued parameter $\alpha$ 
(derivation presented in Appendix~\ref{sec:stab_func}), 
which captures the factors determining the stability of the network-homogeneous state but is
independent of the network structure. This function, referred to as a master stability function in the literature, was originally
derived for a general class of systems that is different from the one we consider here~\cite{PhysRevE.61.5080,Pecora:1998zp}.
The influence of the network structure on the stability is only through the (possibly complex) eigenvalues of the Laplacian matrix $L$, defined by 
\begin{equation}\label{eqn:laplacian}
L_{ij} := d_i \delta_{ij} - A_{ij}, \quad d_i := \sum_{j=1}^n A_{ij}.
\end{equation}
Note that $L$ always has a null eigenvalue $\lambda_1 = 0$ associated with the eigenvector $(1,\ldots,1)^T$, which corresponds to the mode of instability that does not affect the condition 
in Eq.~\eqref{eqn:sync-cond}.
The maximum Lyapunov exponent measuring the stability of the network-homogeneous state is then given by
\begin{equation}\label{eqn:lambda-max}
\Lambda_{\max} := \max_{j\ge 2} \Lambda(\lambda_j),
\end{equation}
i.e., it is stable if $\Lambda_{\max} < 0$, and unstable if $\Lambda_{\max} > 0$.
In addition, $|\Lambda_{\max}|$ gives the asymptotic rate of exponential convergence or divergence.

As an example of stability optimization, we consider the following fundamental question:
\begin{quote}
\textit{For a given number of nodes representing dynamical units, and a given number of links with identical weights, what is the assignment of links that maximizes the rate of convergence to a network-homogeneous state?}
\end{quote}
In the context of this problem, we may assume $A_{ij}$ to be binary ($A_{ij} = 0$ or $1$) without loss of generality, 
since any link weight $\eps \neq 1$ can be factored out of $A_{ij}$ (making $A_{ij}$ binary) and absorbed into $\mathbf{G}(t)$, which is then accounted for by the stability function $\Lambda(\alpha)$.

\section{Sensitivity to unweighted perturbations\label{sec:opt}}

In this section, we demonstrate the sensitivity of the convergence rate to link removal and node addition in optimal undirected networks (Subsection A).
We then show that such sensitivity is not possible for optimal directed networks (Subsection B).

\subsection{Undirected networks\label{sec:sensitivity-undirected}}

\subsubsection{The optimization problem\label{sec:undirected-opt}}

For the class of networks with a fixed number of undirected links $m = \sum_i\sum_{j>i} A_{ij}$, we have the additional constraint that the matrix $A$ is symmetric.
This constraint can arise from the symmetry of the physical processes underlying the interaction represented by a link, such as the diffusion of chemicals through a  
channel connecting reactor cells in a chemical reaction network.
In this case, the maximization of the convergence rate can be succinctly formulated as the minimization of $\Lambda_{\max}$:
\begin{equation}\label{eqn:undirected-opt}
\begin{array}{ll}
\text{Minimize} & \Lambda_{\max}(A)\\[2mm]
\text{subject to} & A_{ij}\in\{0,1\}, \,\, A_{ii} = 0, \,\, A_{ij} = A_{ji}, \\[2mm]
& \displaystyle\sum_i\sum_{j > i} A_{ij} = m.
\end{array}
\end{equation}
If the stability function $\Lambda(\alpha)$ is strictly decreasing on the real line $\{\alpha\in\mathbb{C} \,\vert\, \text{Im}(\alpha) = 0 \}$ 
for $\text{Re}(\alpha) \le \bar{\lambda} := 2m/(n-1)$ (which is satisfied in most  
cases, as detailed in
Supplemental Material~\cite{sm}, Sec.~\ref{si:sec:example-systems}), {\it maximizing the convergence rate to the network-homogeneous state for undirected networks is equivalent to maximizing $\lambda_2$}, the smallest eigenvalue excluding the null eigenvalue that exists for any networks.
We note that the problem is also equivalent to minimizing a bound on the deviations from a network-homogeneous state in a class of networks of non-identical oscillators~\cite{Sun:2009hc}.
There have been a number of previous studies~\cite{PhysRevE.81.025202,PhysRevLett.95.188701,Nishikawa:2006fk,Nishikawa:2006kx,Wang:2007kx,6561538} on the related (but different) problem of maximizing the eigenratio $\lambda_2/\lambda_n$, which measures the synchronizability of the network structure for networks of coupled chaotic oscillators.

The maximization of $\lambda_2$ is generally a challenging task, except for the following particular cases.
For $m = n(n-1)/2$, the only network with $n$ nodes and $m$ links is the complete graph, resulting in the (maximum) value $\lambda_2 = n$.
For $m = n-1$ (implying that the network is a tree), the maximum possible value of $\lambda_2 = 1$ is achieved if and only if the network has the star configuration~\cite{maas1987transportation}. For other values of $m$ (assuming $m \ge n-1$ to ensure that the network is connected), it is challenging even numerically, mainly because each $A_{ij}$ is constrained to be either $0$ or $1$, which makes it a difficult non-convex combinatorial optimization.
The problem of maximizing $\lambda_2$
has been a subject of substantial interest in graph theory, with several notable results in the limit $n \to \infty$, assuming that each node in the network has the same degree and that this common degree is constant~\cite{Alon:1986uq,Friedman:1989:SER:73007.73063,Lubotzky:1988fk} or assuming a fixed maximum degree~\cite{Kolokolnikov:2015}.
In contrast to these bounded-degree results, below we address the maximization of $\lambda_2$ in a different limit, $n \to \infty$, keeping the link density $\phi := 2m/[n(n-1)]$ constant.

\subsubsection{Optimal networks: UCM and MCC\label{sec:opt-undirected}}

Here we define UCM and MCC networks, and then show that they provide analytical solutions of the optimization problem formulated in the previous section.
To define these networks, we first introduce two general quantities that characterize 
connected component sizes.
For a given $k$, let function $M(n,k)$ denote the maximum number of links allowed for any $n$-node network whose connected components have size $\le k$.
Given $m$, we define $k_{n,m}$ to be the smallest (necessarily positive) integer $k$ for which $m \le M(n,k)$, i.e., $k_{n,m}$ is the minimum size of the largest connected components of any network with $n$ nodes and $m$ links. 
We 
also use the notion of graph complements~\cite{MR2159259,Duan:2008fk,Nishikawa:2010fk}.  For a given network with adjacency matrix $A$, its complement
is defined as the network with the adjacency matrix $A^c$ given by 
\begin{equation}\label{eqn:def-comp}
A^c_{ij} = (1 - A_{ij})(1-\delta_{ij}).
\end{equation}
With these definitions and notations, we now
define an {\it MCC network} to be one whose largest connected component of the complement is of size $k_{n,m_c}$, where $m_c := n(n-1)/2 - m$ is the number of links in the complement.

To see how the definition of MCC networks relates to the maximization of $\lambda_2$, we note that the maximum Laplacian eigenvalue of any network is upper-bounded by its largest component size (stated and proved as Proposition~\ref{prop:max-lambda} in Supplemental Material~\cite{sm}, Sec.~\ref{si:sec:prop}B).
We also note that the nonzero Laplacian eigenvalues of a network and its complement are related through
\begin{equation}\label{eqn:comp-eig}
\lambda_{n-i+2}^c = n - \lambda_i, \quad i=2,\ldots,n,
\end{equation} 
where we denote the Laplacian eigenvalues of the network as $0 = \lambda_1 \le \lambda_2 \le\cdots\le \lambda_n$ (noting that the symmetry of $A$ constrains them to be real), and those of the complement as $0 = \lambda_1^c < \lambda_2^c \le \cdots \le \lambda_n^c$.
Thus, the smaller the largest component size in the complement, the smaller we expect the eigenvalue $\lambda_n^c$ to be, which would imply larger $\lambda_2$ according to Eq.~\eqref{eqn:comp-eig}.

For special combinations of $n$ and $m$, namely, $n = k\ell$ and $m = k^2 \ell(\ell - 1)/2$ with arbitrary positive integers $k$ and $\ell$, the complement of an MCC network necessarily consists of $\ell$ components, each fully connected and of size $k$ 
(stated and proved as Proposition~\ref{prop:UCM-unique} in Supplemental Material~\cite{sm}, Sec.~\ref{si:sec:prop}A).
We refer to this unique MCC network as the {\it UCM network} for the given $n$ and $m$.
Translating the structure of its complement to that of the network itself, the UCM network can be characterized as the one in which (i) the nodes are divided into $\ell$ groups of equal size $k$ (uniform), (ii) all pairs of nodes from different groups are connected (complete), and (iii) no pair of nodes within the same group are connected (multipartite).
Figure~\ref{fig:ucm} shows examples of UCM and MCC networks.

\begin{figure}[t]
\begin{center}
\includegraphics[width=1.0\columnwidth]{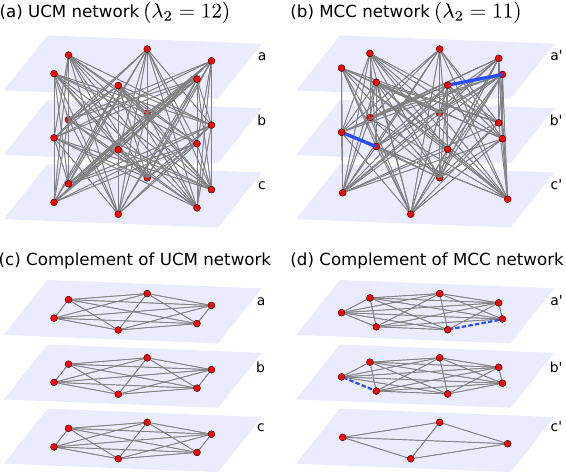}
\end{center}\vspace{-9pt}
\caption{\label{fig:ucm}
UCM and MCC networks with $n=18$ nodes.
(a) The UCM network with $\ell=3$ groups (labeled a, b, and c) of $k=6$ nodes each.
All pairs of nodes belonging to different groups are connected, while all pairs within the same group are not connected, leading to a total of $m=k^2 \ell(\ell - 1)/2=108$ links.
(b) An MCC network constructed with the same number of nodes but with one fewer link ($m=107$) and groups of unequal sizes (labeled $\text{a}'$, $\text{b}'$, and $\text{c}'$, and of sizes $7$, $7$, and $4$, respectively).
Note that in this case some nodes within the same group are connected (as indicated by solid blue lines).
(c) The complement of the UCM network in panel~(a), which has $\ell k(k-1)/2 = 45$ links.
In the complement, a node pair is connected if they are in the same group, and not connected if they are from different groups.
(d) The complement of the MCC network in panel~(b), which has $46$ links.
Since it has one more link than what can be accommodated by three isolated groups of size $6$ [as in panel~(c)], the minimum possible size of the largest component in the complement equals $7$ in this case.
Note that groups $\text{a}'$ and $\text{b}'$ have missing links (indicated by dashed blue lines), which correspond to the links within groups $\text{a}'$ and $\text{b}'$ in panel (b).
The required increase in the size of the largest component in the complement forces $\lambda_2$ to decrease by one.}
\end{figure}

To establish the optimality of UCM and MCC networks, we first prove the following general upper bound:
\begin{equation}\label{bound}
\lambda_2 \le \lfloor 2m/n \rfloor = \lfloor \phi(n-1) \rfloor
\end{equation}
for any $n$ and $m$ for which the link density $\phi = 2m/[n(n-1)] < 1$ (where $\lfloor x \rfloor$ denotes the largest integer not exceeding $x$).
We prove this bound using Proposition 3.9.3 of Ref.~\cite{brouwer2012spectra}, which states that 
\begin{equation}\label{prop393}
\lambda_n \ge d_\text{max} + 1
\end{equation}
holds true for any network (with at least one link), where $d_\text{max}$ denotes the maximum degree of the network.
Applying this proposition to the complement of the network (rather than the network itself)
gives
\begin{equation}\label{eqn:bound-proof}
\lambda_n^c \ge d_\text{max}^c + 1 \ge \lceil \bar{d}^c \rceil + 1 = \lceil 2m_c/n \rceil + 1,
\end{equation}
where $d_\text{max}^c$ and $\bar{d}^c$ denote the maximum and mean degree of the complement, 
respectively, and $\lceil x \rceil$ denotes the smallest integer larger than or equal to $x$.
Thus, we have
$\lambda_2 = n - \lambda_n^c 
\le n - (\lceil 2m_c/n \rceil + 1)
= n - (\lceil n - 1 - 2m/n \rceil + 1)
= \lfloor 2m/n \rfloor$, 
establishing
Eq.~\eqref{bound}.

The optimality of UCM networks can now be established for any combination of $n$ and $m$ for which the UCM network can be defined [i.e., $n = k\ell$ and $m = k^2 \ell(\ell - 1)/2$].
Indeed, since each connected component in the complement of such a UCM network is fully connected and of size $k$, it follows that the maximum Laplacian eigenvalue of the complement is $\lambda_n^c = k$.
(This is because the Laplacian spectrum of a network is the union of the Laplacian spectra of its connected components, which is a known fact presented as Proposition~\ref{prop:union} in Supplemental Material~\cite{sm}, Sec.~\ref{si:sec:prop}B.) 
We thus
conclude that $\lambda_2 = n-k = \lfloor 2m/n \rfloor$, implying that the UCM network attains the upper bound in Eq.~\eqref{bound} and has the maximum possible $\lambda_2$.
Moreover, the UCM network is actually the only optimizer among all networks with the same $n$ and $m$ 
(proved in Appendix~\ref{sec:ucm_optimality}).

For other MCC networks, we establish the formula
\begin{equation}\label{eqn:lambda2_pre}
\lambda_2 = n-k_{n,m_c}
\end{equation}
for any link density $\phi < 1$ and use it to show that MCC networks attain the upper bound in Eq.~\eqref{bound} and thus are optimal in several cases of lowest and highest link densities, as well as for a range of link density around each value corresponding to a UCM network.
We also show that each MCC network is locally optimal in the space of all networks with the same $n$ and $m$ in the sense that $\lambda_2 \le n-k_{n,m_c}$ holds true for any network obtained by rewiring a single link.
Proofs of these results 
can be found in Supplemental Material~\cite{sm}, Secs.~\ref{si:sec:prop}B and \ref{si:sec:prop}C.
The optimality of these networks, which have fully connected clusters in the complement, suggests potential significance of other, more general network motifs~\cite{Milo:2002}, whose statistics 
have
been studied in the context of network optimization~\cite{Kaluza:2007,Sporns:2004}.

These $\lambda_2$-maximizing networks can be explicitly constructed.
In fact, 
given any
$n$ and $m$, an MCC network with $n$ nodes and $m$ links can be constructed
by forming as many isolated, fully connected clusters of size $k_{n,m_c}$ as possible in the complement of the network.
Details on this procedure are described in Appendix~\ref{si:sec:construction_MCC}, and a MATLAB implementation is available for download~\cite{software}.
This procedure yields the (unique) UCM network if $n = k\ell$ and $m = k^2 \ell(\ell - 1)/2$.
Similar strategies that suppress the size of largest connected components, when incorporated into a network growth process, have been observed to cause discontinuous, or continuous but ``explosive'' percolation transitions~\cite{Achlioptas:2009ys,PhysRevLett.105.255701,DSouza:2015fk,Riordan:2011kx,Nagler:2011aa}.
The deterministic growth process defined in Ref.~\cite{Rozenfeld:2010uq} is particularly close to the definition of MCC networks because the process explicitly minimizes the product of the sizes of the components connected by the new link in each step.

\subsubsection{Sensitivity of optimal networks\label{sec:sens-opt-undirected}}

To demonstrate the sensitivity of UCM networks to link removals and node additions, we first study the dependence of $\lambda_2$ for MCC networks on the link density $\phi < 1$.
By deriving an explicit formula for $k_{n,m_c}$, we rewrite Eq.~\eqref{eqn:lambda2_pre} as
\begin{equation}\label{eqn:lambda2}
\lambda_2 = \lfloor C_{\ell,n}(\phi) \cdot n \rfloor,
\end{equation}
where 
\begin{equation}\label{eqn:C_ell_n}
C_{\ell,n}(\phi) = \frac{\ell^2 - \sqrt{\ell^2 - \phi\ell(\ell+1)\bigl(1-\frac{1}{n}\bigr)}}{\ell(\ell+1)},
\end{equation}
and $\ell$ depends on $\phi$ and is defined as the unique integer satisfying 
\begin{equation}\label{eqn:def-ell}
1 - \frac{1}{\ell} \le \Bigl(1 - \frac{1}{n}\Bigr) \phi < 1 - \frac{1}{\ell+1}
\end{equation}
(derivation presented in Supplemental Material~\cite{sm}, Sec.~\ref{si:sec:prop}D).
Equation~\eqref{eqn:lambda2} indicates that $\lambda_2$
experiences a series of sudden jumps as the link density increases from $\phi = 2/n$ (the minimum possible value for a connected network, corresponding to the star configuration) to $\phi = 1$ (corresponding to the fully connected network).
This behavior is better understood by considering the complement of the network as the number of links $m_c$ in the complement increases (corresponding to decreasing link density $\phi$), 
as illustrated for $n=20$ in Fig.~\ref{fig:construction}.
When the complement has exactly $M(n,k_{n,m_c})$ links, any additional link would force the maximum component size $k_{n,m_c}$ to increase by one, causing a jump in $\lambda_2 = n-k_{n,m_c}$.
In Fig.~\ref{fig:construction}, for example, when the network that 
has $m_c = M(20,4) = 30$, $k_{20,30} = 4$, and $\lambda_2 = 16$ gains one more link in its complement ($m_c = 31$), the component size jumps to $k_{20,31} = 5$ and $\lambda_2$ jumps down to $15$.
The 18-node UCM and MCC networks in Fig.~\ref{fig:ucm} also illustrate such a jump.
In the context of percolation problems, similar cascades of jumps in the maximum component size, called microtransition cascades, have been identified as precursors to global phase transitions~\cite{PhysRevLett.112.155701}.

\begin{figure}[t]
\begin{center}
\includegraphics[width=0.97\columnwidth]{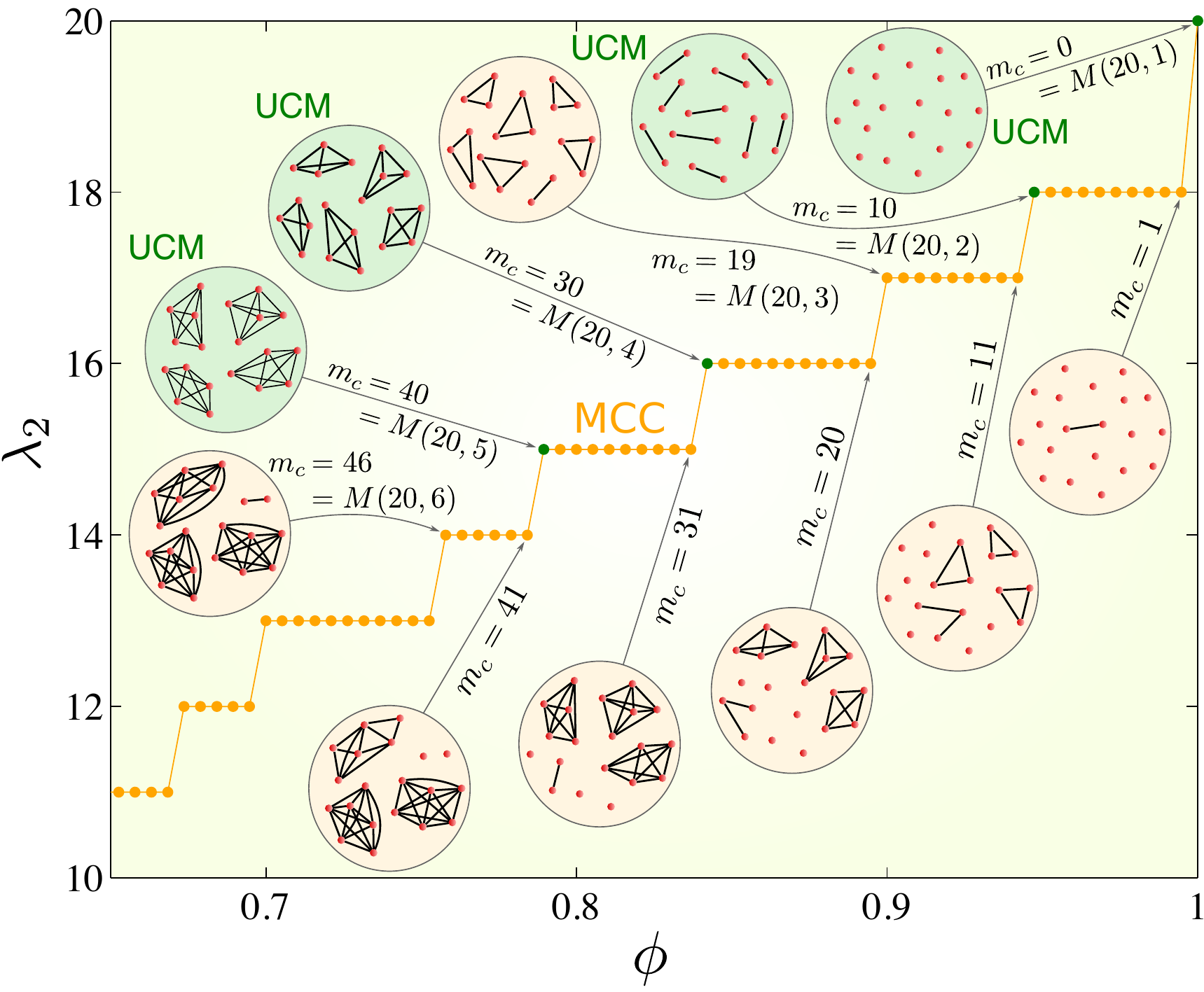}
\end{center}\vspace{-5mm}
\caption{\label{fig:construction}
UCM and MCC networks of size $n=20$ for the maximization of smallest non-zero Laplacian eigenvalue $\lambda_2$.
For a given link density $\phi$, the orange dot indicates $\lambda_2$ for MCC networks.
A UCM network, when possible for that value of $\phi$, is indicated by a green dot.
As $\phi$ increases, the eigenvalue $\lambda_2$ experiences discrete jumps, corresponding to sudden changes in the structure of the network.
The changes in the link configuration of the network's complement, as well as the associated jumps in the size of their largest clusters, are illustrated in the circles.
At $\phi$ values just above and below the jumps, the complement has $m_c = M(n, k_{n,m_c})$ links and $m_c = M(n, k_{n,m_c}) + 1$ links, respectively.}
\end{figure}

\begin{figure}[t]
\begin{center}
\includegraphics[width=0.45\textwidth]{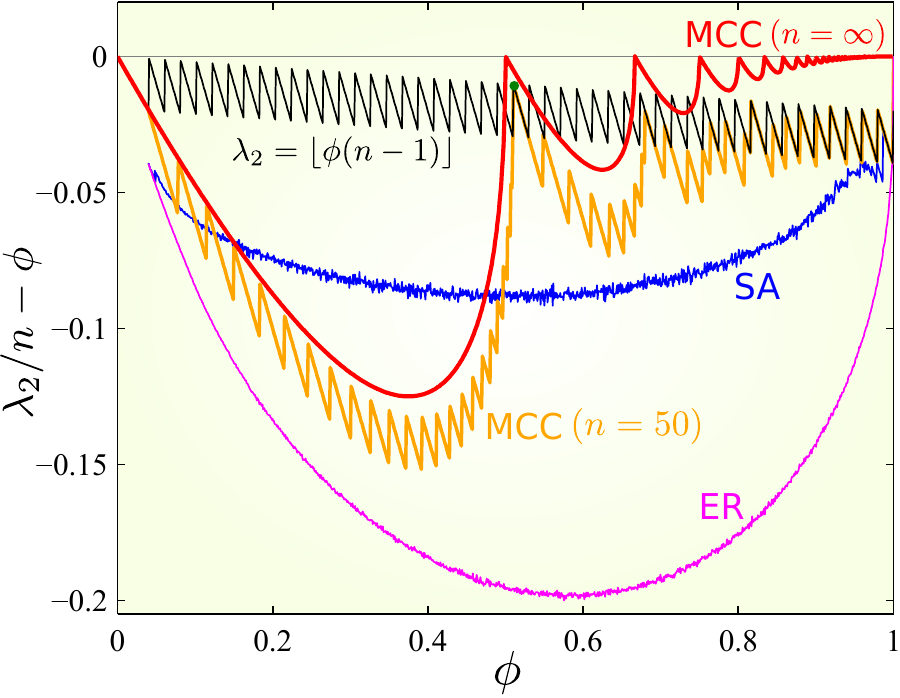}
\end{center}\vspace{-3mm}
\caption{\label{fig:comparison}
Sensitive dependence of the Laplacian eigenvalue $\lambda_2$ on link density $\phi$ for undirected networks of size $n=50$. 
Each curve indicates $\lambda_2$ normalized by the network size $n$, relative to the 
link density $\phi$. (Plots
of $\lambda_2/n$ itself can be found in Supplemental Material, Fig.~\ref{fig:comparisonA}~\cite{sm}.)
The blue curve shows the result of a single run of SA to maximize $\lambda_2$ with a fixed number of links.
Each point on the
magenta curve is the average over $1{,}000$ realizations of the ER random networks with connection probability $\phi$.
The orange and red curves indicate the MCC networks for $n = 50$ [Eq.~\eqref{eqn:lambda2}] and in the limit of $n\to\infty$ [Eq.~\eqref{eqn:lambda2inf}], respectively.
Notice the square-root singularity on the left of points $\phi = \phi_\ell = \frac{\ell-1}{\ell}$, $\ell = 2,3,\ldots,$ on the red curve.
The green dot near one of these singularity 
points
indicates the UCM network with $\ell=2$ and $k=25$, which achieves the upper bound $\lambda_2 \le \lfloor \phi(n-1) \rfloor$ shown by the black curve.
}
\end{figure}

Figure~\ref{fig:comparison} demonstrates that for a wide range of $\phi$, the MCC networks improve $\lambda_2$ significantly over the Erd\H{o}s-R\'enyi (ER) random networks, as well as those identified by direct numerical optimization of $\lambda_2$ using simulated annealing (SA).
The difference is particularly large for $\phi$ near certain special values such as $1/2$.
Note that the optimal value of $\lambda_2$ given by the upper bound (black curves) is achieved not only by the UCM network (for example, the one indicated by the green dot for $k=25$, $\ell=2$) at $\phi = 25/49 \approx 0.51$, but also by MCC networks (orange curves) for a finite range of $\phi$ around this value.
The optimal $\lambda_2$ value, however, is sensitive to changes in the link density $\phi$, and it departs quickly from its value at $\phi = 25/49$ as $\phi$ moves away from $25/49$, particularly for $\phi < 25/49$.

In fact, $\lambda_2$ has many points exhibiting such sensitivity, which
becomes more prominent for larger networks and turns into a singularity as $n\to\infty$ with fixed $\phi$.
To see this, we take the limit in Eq.~\eqref{eqn:lambda2} to obtain
\begin{equation}\label{eqn:lambda2inf}
\lim_{n\to\infty} \frac{\lambda_2}{n} 
= \frac{\ell^2 - \sqrt{\ell^2 - \phi\ell(\ell+1)}}{\ell(\ell+1)},
\end{equation}
where $\ell$ is the unique integer determined by 
$\phi_\ell \le \phi < \phi_{\ell+1}$, where we define $\phi_\ell := 1 - \frac{1}{\ell} = \frac{\ell-1}{\ell}$ for any positive integer $\ell$.
This function of $\phi$, shown in Fig.~\ref{fig:comparison} (red curve), has a cusp-like dependence on $\phi$ around $\phi = \phi_\ell$, at which it achieves the asymptotic upper bound 
$\lim_{n\to\infty} \lambda_2/n \le \phi$ [which follows directly from Eq.~\eqref{bound}]
and has a \emph{square-root singularity} on the 
left, i.e., the
derivative on the left 
diverges (while
the derivative on the right equals $1/2$).
This singularity is inherently different from the discrete jumps observed above for finite $n$.
Indeed, as the network size increases, the size of the jumps and the distance between consecutive jumps both tend to zero (as in the microtransition cascades \cite{PhysRevLett.112.155701} in percolation 
problems). The function thus becomes 
increasingly closer to a (piecewise) smooth function, while the square-root singularity becomes progressively more visible (verified numerically
in Fig.~\ref{fig:convergence}(a) of Supplemental Material~\cite{sm}).
For each singularity point $\phi = \phi_\ell$, there is a sequence of UCM networks with increasing $k$ (and thus increasing network size $n=k\ell$), for which the link density $\phi = (\ell - 1)/(\ell - \frac{1}{k})$ approaches $\phi_\ell$ as $k \to \infty$.

The UCM networks associated with these singularities also exhibit sensitivity to the removal of an {\it arbitrary} link. 
As shown in the previous section, the UCM networks are the only networks that attain the upper bound in Eq.~\eqref{bound} and satisfy $\lambda_2 = \lfloor 2m/n \rfloor = \lfloor k(\ell-1) \rfloor = k(\ell-1)$ [where the last equality holds because $k(\ell-1)$ is an integer].
The removal of any single link reduces the bound to $\lfloor 2(m-1)/n \rfloor = \lfloor k(\ell-1) - 2/n \rfloor =  k(\ell-1) - 1$ and thus the normalized eigenvalue $\lambda_2/n - \phi$ by at least $1/n$.
Since the link removal reduces $\phi$ by $2/[n(n-1)]$, the derivative of the normalized eigenvalue with respect to $\phi$ (in the limit of large $n$) is greater than or equal to
\begin{equation}
\lim_{n\to\infty} \frac{1/n}{2/[n(n-1)]} = \lim_{n\to\infty} \frac{n-1}{2} = \infty.
\end{equation}
In terms of the complement, this can be understood as coming from the unavoidable increase of the component size, since the link removal in the network corresponds to a link addition in the complement.
We note that the argument above is valid only for UCM networks, since the UCM network is the only one that attains the bound for any $\phi$ value at which the upper bound is discontinuous, i.e., when $2m/n$ is an integer 
(proof given in
Appendix~\ref{sec:ucm_optimality}).
In summary, we have the following result:
\begin{quote}
{\it The 
UCM networks, which maximize $\lambda_2$ and 
correspond to singularities in the $\lambda_2$ vs.\ $\phi$ curve for MCC 
networks,
are sensitive to link removals}.
\end{quote}

The UCM networks show similar sensitivity to node additions as well.
When $m$ is fixed, the expression for $\lambda_2$ given in Eq.~\eqref{eqn:lambda2}, considered now as a function of $n$, has a square-root dependence on the right of the points 
$n = \sqrt{2m/\phi_{\ell}}$, $\ell = 2,3,\ldots$
(corresponding to the UCM networks), as illustrated in Fig.~\ref{fig:n-dep-fixed-m} of Supplemental Material~\cite{sm}.
Similarly to the case of link removals, it can be shown that the bound in Eq.~\eqref{bound} suddenly drops from $k(\ell-1)$ to $k(\ell-1) - 1$ when a new node is connected to the network as long as the number of new links is less than $m/n$, and that this drop leads to an infinite derivative for $\lambda_2/n - \phi$ with respect to $\phi$ in the limit of large $n$.

\subsection{Directed networks}\label{sec:directed-nongeneric}

\subsubsection{The optimization problem}

For the class of networks with a fixed number of directed links $m_d=\sum_i\sum_{j\neq i} A_{ij}$, 
the matrix $A$ can be asymmetric in general.
In this case, the problem of maximizing the rate of convergence to the network-homogeneous state
can be expressed as
\begin{equation}
\begin{array}{ll}
\text{Minimize} & \Lambda_{\max}(A)\\[2mm]
\text{subject to} & A_{ij}\in\{0,1\}, \,\, A_{ii} = 0,\\[2mm]
&\displaystyle\sum_i\sum_{j\neq i} A_{ij} = m_d.
\end{array}
\end{equation}
The solution of this problem generally depends on the specific shape of the stability function.
However, the problem is equivalent to maximizing Re$(\lambda_2)$, the smallest real part among the eigenvalues of $L$ excluding the identically null eigenvalue $\lambda_1$, if the stability function $\Lambda(\alpha)$ is strictly decreasing in Re$(\alpha)$ and independent of Im$(\alpha)$ for Re$(\alpha) \le \bar{\lambda} := m_d/(n-1)$.
This condition is satisfied, e.g., for 
consensus and diffusion 
processes (details presented in Supplemental Material~\cite{sm}, Secs.~\ref{si:sec:example-systems}D and \ref{si:sec:example-systems}E, respectively).
This equivalence is a consequence of the upper bound, 
\begin{equation}\label{eqn:bound_dir}
\text{Re}(\lambda_2) \le \bar{\lambda}, 
\end{equation}
which follows from the fact that the sum of the eigenvalues equals the trace of $L$, which in turn equals 
$m_d$.  [We note that the tighter bound in Eq.~\eqref{bound} is not applicable to directed networks in general.]

\subsubsection{Optimal networks\label{sec:opt-dir}}

The optimization problem just formulated can be solved if $m_d$ is ``quantized,'' i.e., equals an integer multiple of $n-1$, in which case there are networks that satisfy $\lambda_2 = \cdots = \lambda_n = \bar{\lambda}$~\cite{Nishikawa:2010fk}.
Such networks attain the upper bound in Eq.~\eqref{eqn:bound_dir} and thus are optimal.
The class of directed networks satisfying $\lambda_2 = \cdots = \lambda_n = \bar{\lambda}$
has previously been studied within the context of network synchronization using 
objective functions that are not defined by Re$(\lambda_2)$ and different from
the convergence rate considered here~\cite{Nishikawa:2010fk,PhysRevLett.107.034102}.
If $m_d$ is not an integer multiple of $n-1$, the maximization of Re$(\lambda_2)$, like the maximization of $\lambda_2$ for undirected networks, is a hard combinatorial optimization problem.
Here we compute the Laplacian eigenvalues symbolically (and thus exactly) for all directed networks of size $n=3$, $4$, and $5$.
For the quantized values of $m_d$, we verify that the upper bound $\bar{\lambda}$ is indeed attained [Figs.~\ref{fig:geom-dist-dir}(a)--\ref{fig:geom-dist-dir}(c)], in which case $\lambda_2 = \cdots = \lambda_n = \bar{\lambda}$ is not only real but also an integer. For intermediate values of $m_d$, 
the maximum Re$(\lambda_2)$ does not appear to follow a simple rule; it can be strictly less than $\bar{\lambda}$, have nonzero imaginary part, and/or be non-integer.

\begin{figure}[t]
\begin{center}
\includegraphics[width=\columnwidth]{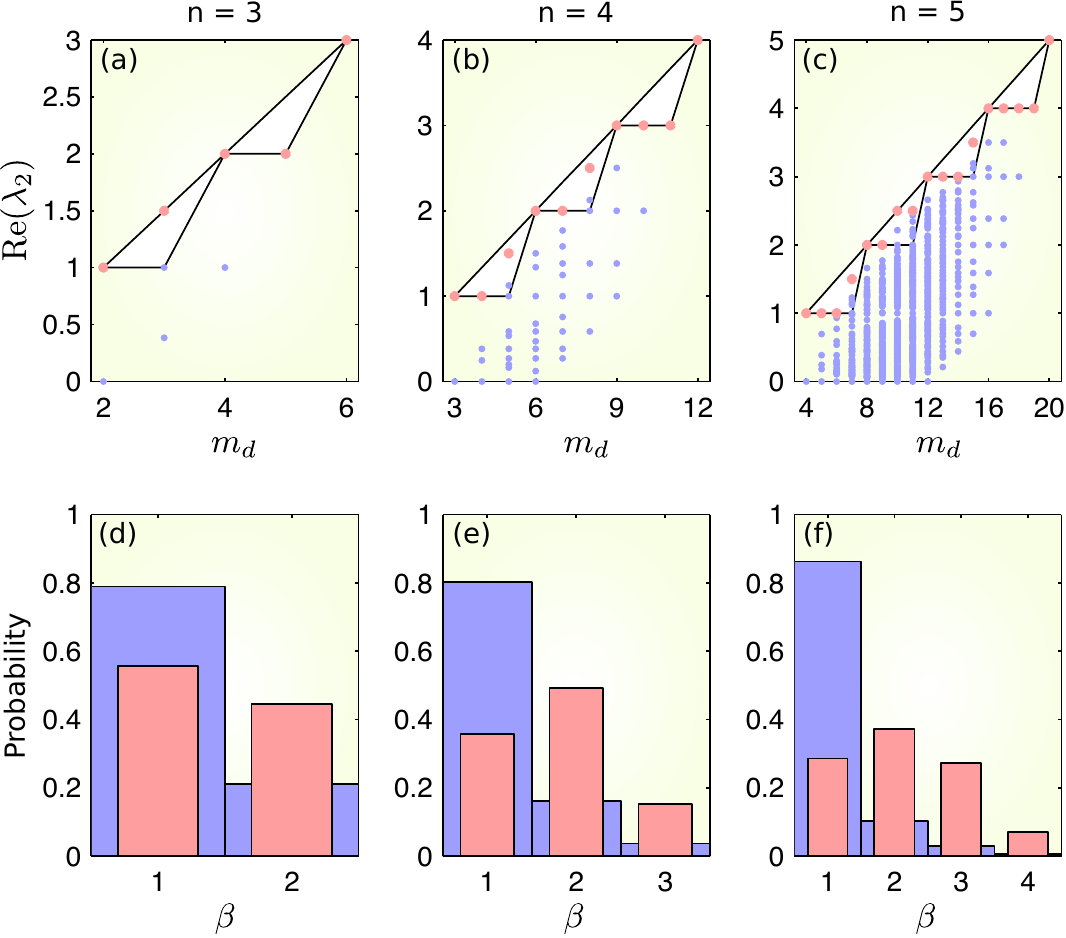}
\end{center}\vspace{-18pt}
\caption{\label{fig:geom-dist-dir}
Optimization induces spectral degeneracy.
The maximum $\text{Re}(\lambda_2)$ is computed exactly through symbolic calculation of eigenvalues for all directed networks of size $n=3$, $4$, and $5$. 
(a--c) Maximum values of $\text{Re}(\lambda_2)$ as a function of the number of directed links, $m_d$ (red dots).
All the other possible values of $\text{Re}(\lambda_2)$ are indicated by blue dots.
The black curves indicate the upper bound $\bar{\lambda}$ and lower bound $\lfloor \bar\lambda \rfloor$ for the maximum $\text{Re}(\lambda_2)$.
The upper bound comes from Eq.~\eqref{eqn:bound_dir}, while the
lower bound is derived in 
Sec.~\ref{sec:sensitivity-opt-dir}.
(d--f) Probability distribution of geometric degeneracy $\beta$ for $\lambda_2$ of all networks of a given size (blue bars) and of the $\text{Re}(\lambda_2)$-maximizing networks (red bars).
For each network size, the difference between the two distributions shows that the maximization tends to increase $\beta$.}
\end{figure}

\subsubsection{Non-sensitivity of optimal networks\label{sec:sensitivity-opt-dir}}

In the limit of large networks, however, there is a simple rule: we show below that the maximum value of Re$(\lambda_2)$, normalized by $n$, converges to the link density $\phi := m_d/[n(n-1)]$ as $n \to \infty$ with $\phi$ fixed.
This in particular implies that {\em the normalized maximum} Re$(\lambda_2)$ {\em has no sensitive dependence on $\phi$}, in sharp contrast to the sensitivity observed in the same limit for undirected networks~\cite{comment}. 

To establish this non-sensitivity result, we first note that $\phi = \bar{\lambda}/n$ is an upper bound for the maximum value of Re$(\lambda_2)/n$, which follows immediately from Eq.~\eqref{eqn:bound_dir}.
We show that the maximum value approaches the upper bound by showing that there is a lower bound that approaches the upper bound.
The lower bound is established by constructing a specific network with $n$ nodes and $m_d$ directed links.
To construct this network, we start with a variant of directed star networks, in which a core of $s$ fully connected nodes are all connected to all the other nodes, where we define
$s :=\lfloor \bar{\lambda} \rfloor$.
Since such a network involves exactly $s(n-1)$ links, the remaining $r$ links, where $r:=m_d-s(n-1)$, are added to the network.
The network can thus be constructed as follows:
1) For each $i=1,\ldots,s$, add $n-1$ links from node $i$ to all the other nodes.
2) Add $r$ links from node $s+1$ to nodes $1,\ldots,r$ if $r \le s$ and to nodes $1,\ldots,s$, $s+2,\ldots,r+1$ if $r > s$.
This
network satisfies
$\lambda_2 = s = \lfloor \bar{\lambda} \rfloor$
(proof given in Appendix~\ref{si:sec:non-sensitive-directed}),
which provides a lower bound for the maximum value of Re$(\lambda_2)$.
This lower bound, as well as the upper bound $\bar{\lambda}$, is indicated by black curves in Figs.~\ref{fig:geom-dist-dir}(a)--\ref{fig:geom-dist-dir}(c) for $3 \le n \le 5$.
Thus,
the maximum value of Re$(\lambda_2)/n$ 
is at least $s/n$, and this lower bound approaches the upper bound for large networks:
$s/n = \lfloor \bar{\lambda} \rfloor /n = \lfloor \phi n \rfloor/n \to \phi$ as $n \to \infty$.
This proves our claim that Re$(\lambda_2)/n$ for optimal networks is a smooth function of $\phi$ in the limit of large networks, thus establishing the absence of sensitivity.

\section{Sensitivity to weighted perturbations\label{sec:generic}}

To demonstrate the second type of sensitivity, we
now study how the convergence rate behaves when a small 
weighted
perturbation is applied to the network structure, particularly when the initial network is optimal or close to being optimal.
Since the convergence rate is determined by the Laplacian eigenvalues through the stability function $\Lambda(\alpha)$ and Eq.~\eqref{eqn:lambda-max}, it suffices to analyze how the Laplacian eigenvalues respond to 
such
perturbations, which we formulate as perturbations of the adjacency matrix in the form $A + \delta\Delta A$, 
where the small parameter $\delta$ is positive (unless noted otherwise) and $\Delta A$ is a fixed 
matrix.
This type of structural perturbations can represent imperfections in the strengths of couplings in real networks, such as power grids and networks of chemical~\cite{Showalter2015}, electrochemical~\cite{Kiss2002}, or optoelectronic~\cite{PhysRevLett.107.034102} oscillators.

\subsection{Eigenvalue scaling for arbitrary networks\label{sec:gen-scaling-eig}}

Here we show that for a given Laplacian eigenvalue $\lambda$ of a directed network and a generic choice of $\Delta A$, the change $\Delta\lambda$ of the eigenvalue due to the perturbation generally follows a scaling relation, $|\Delta\lambda| \sim \delta^\gamma$.
We also provide a rigorous bound for the scaling exponent $\gamma$. This scaling exponent determines the nature of the dependence of the perturbed eigenvalue on $\delta$: if $0 < \gamma < 1$, the dependence is sensitive and characterized by an infinite derivative at $\delta=0$, and if $\gamma \ge 1$, it is non-sensitive and characterized by a finite derivative.

\subsubsection{Bound on scaling exponent}

We provide an informative bound on $\gamma$ by proving the following general result on matrix perturbations.
Suppose $\lambda$ is an eigenvalue of an arbitrary matrix $M$ with geometric degeneracy $\beta$~\cite{PhysRevLett.107.034102}, defined as the largest number of repetitions of $\lambda$ associated with the same eigenvector (i.e., the size of the largest diagonal block associated with $\lambda$ in the Jordan canonical form of $M$).
For perturbations of the form $M + \delta \Delta M$ with an arbitrary matrix $\Delta M$, there exists a constant $C$ such that 
the corresponding change $\Delta\lambda = \Delta\lambda(\delta)$ in the eigenvalue, as a function of $\delta$, satisfies 
\begin{equation}\label{eqn:scaling_bound}
\limsup_{\delta\rightarrow{0}}\frac{|\Delta\lambda(\delta)|}{\delta^{1/\beta}} \leq C
\end{equation}
(proof given in Appendix~\ref{si:sec:mat-pert}).
Applying this result to an eigenvalue $\lambda$ of the Laplacian matrix $L$, we see that $\gamma \ge 1/\beta$, implying that the set of perturbed eigenvalues that converge to $\lambda$ as $\delta\to 0$ do so at a rate no slower than $\delta^{1/\beta}$.

\subsubsection{Typical scaling behavior\label{sec:typical-behavior}}

The bound established above suggests 
that the scaling $|\Delta\lambda|\sim \delta^{1/\beta}$ would be observed for all perturbed eigenvalues that converge to $\lambda$ as $\delta\to 0$.
In fact, our numerics supports a more refined statement for networks under generic weighted structural perturbations: for each eigenvector (say, the $j$th one) associated with $\lambda$, there is a set of $\beta_j$ perturbed eigenvalues that converge to $\lambda$ as $\delta\to 0$ and follows the scaling,
\begin{equation}\label{eqn:generic_scaling}
|\Delta\lambda| \sim \delta^{1/\beta_j},
\end{equation}
where $\beta_j$ is the number of repetitions of $\lambda$ associated with the $j$th eigenvector (i.e., the size of the $j$th Jordan block associated with $\lambda$).
We numerically verify this individual scaling for Laplacian eigenvalues using random perturbations applied to all off-diagonal elements of $A$.
We consider two examples of directed networks of size $n=6$, shown at the top of Fig.~\ref{fig:gen-pert}, which are both optimal because $\lambda_2=\cdots=\lambda_6$.
For each of these networks, the left column plots in the corresponding panel of Fig.~\ref{fig:gen-pert} show the distributions of the scaling exponent $\gamma$ in the relation $|\Delta\lambda| \sim \delta^\gamma$ for random choices of $\Delta A$, where $\gamma$ is estimated from fitting the computed values of perturbed $\lambda_i$ over different ranges of $\delta$.
We see that the 
distributions are
sharply peaked around $\beta_j$ (indicated by the gray inverted triangles) with smaller spread for narrower ranges of $\delta$, supporting the asymptotic scaling in Eq.~\eqref{eqn:generic_scaling} in the limit $\delta\to 0$.

\begin{figure*}[tb]
\centering
\includegraphics*[width=0.9\textwidth]{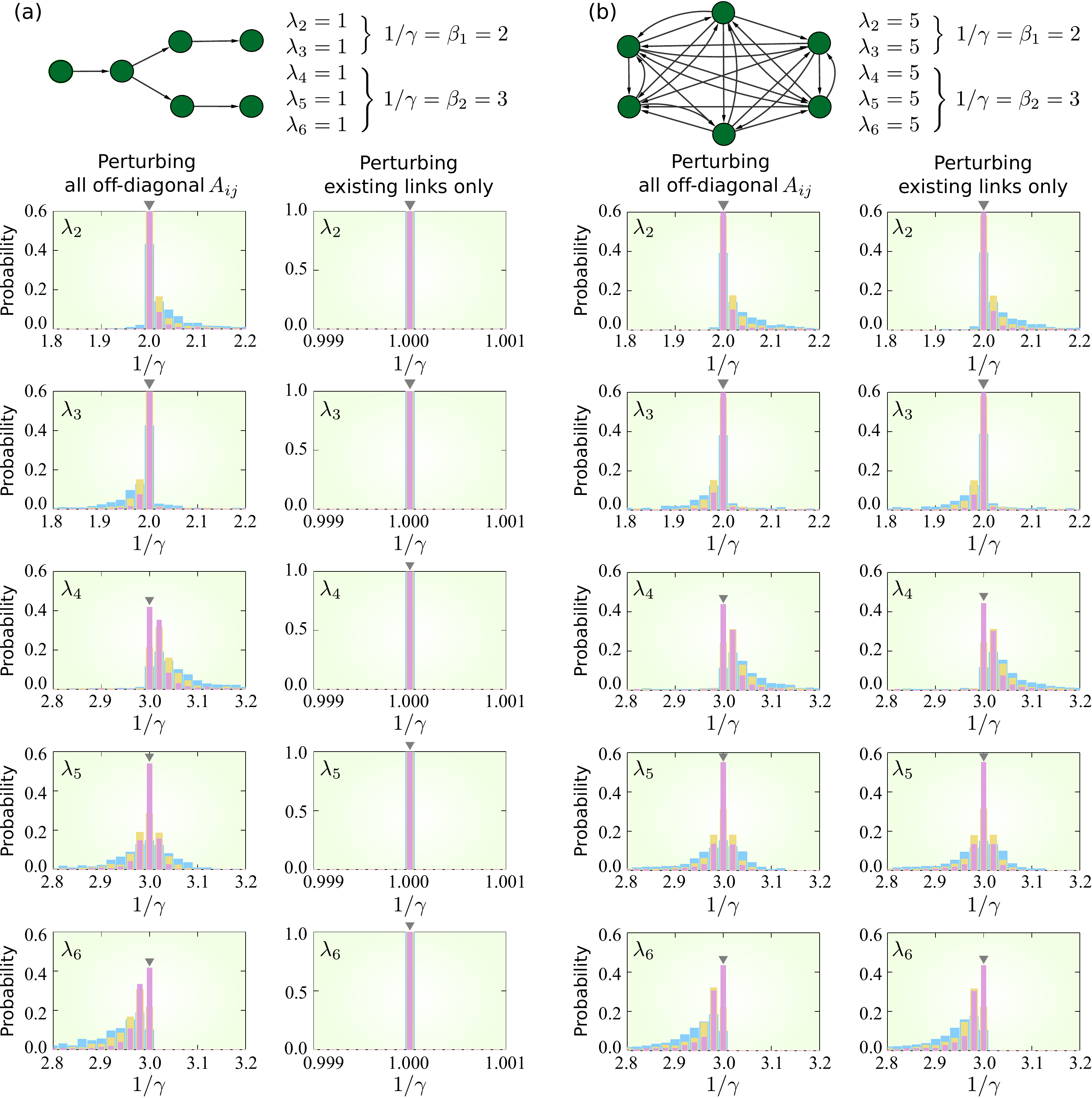}
\caption{\label{fig:gen-pert}
Distribution of scaling exponents for the Laplacian eigenvalues of $6$-node optimal directed networks under random structural perturbations.
(a) Example network (a directed tree) with $\lambda_2=\cdots=\lambda_6=1$.
(b) Example network with $\lambda_2=\cdots=\lambda_6=5$.
Each plot shows histograms of $1/\gamma$, where $\gamma$ is the scaling exponent numerically estimated
from $\lambda_i$ (the $i$th smallest perturbed eigenvalue) computed at
$1{,}000$ equally spaced values of $\delta$ in the intervals $[0,10^{-3}]$ (blue), $[0,10^{-4}]$ (yellow), and $[0,10^{-5}]$ (pink).
We determine $\gamma$ by applying MATLAB's built-in linear least-squares algorithm in log scale~\cite{MATLAB:2010}.
Each histogram is generated by estimating $\gamma$ for $10{,}000$ realizations of $\Delta A$, where each element of $\Delta A$ (corresponding to the perturbation of the weight of an existing link or the addition of a new link with small weight) is chosen randomly from the uniform distribution on $[-1,1]$.
When perturbing all the off-diagonal elements of the adjacency matrix $A$ (left column plots in each panel),
the results support $1/\gamma=\beta_j$ for both networks.
When perturbing only the existing links (right column plots in each panel), the scaling exponent depends on the initial network: the plots support $1/\gamma=1$ for the directed tree in panel (a) and $1/\gamma=\beta_j$ for the network in panel~(b).
}
\end{figure*}

We 
note that, for non-generic weighted perturbations (e.g., if the perturbation is constrained to a subset of the off-diagonal elements of $A$), the exponent may be different from $1/\beta_j$ in Eq.~\eqref{eqn:generic_scaling}.
For example, when perturbing only the existing links of a directed tree (which is optimal with $\lambda_2 = \cdots = \lambda_n = 1$), the exponent is one, and thus the network is not sensitive to this type of perturbations even if the degeneracy $\beta > 1$, as illustrated in Fig.~\ref{fig:gen-pert}(a) (right column plots).
This follows from the fact that the Laplacian matrix of a directed tree is triangular under appropriate indexing of its nodes, which remains true after perturbing the existing links.
This non-sensitivity result can be extended to certain other cases,
e.g., when $P^{-1} \Delta L P$ is a triangular matrix, where $P$ is the non-singular matrix in the Jordan decomposition of $L$ and $\Delta L$ is the perturbation of the Laplacian matrix 
(proof presented in Appendix~\ref{sec:non-sensitivity-constrained}).
In other cases, the scaling with exponent $1/\beta_j$, as in Eq.~\eqref{eqn:generic_scaling}, can be observed even when perturbing only the existing links, as illustrated in Fig.~\ref{fig:gen-pert}(b) (right column plots).

\begin{figure*}[ht]
\begin{center}
\includegraphics[width=0.9\textwidth]{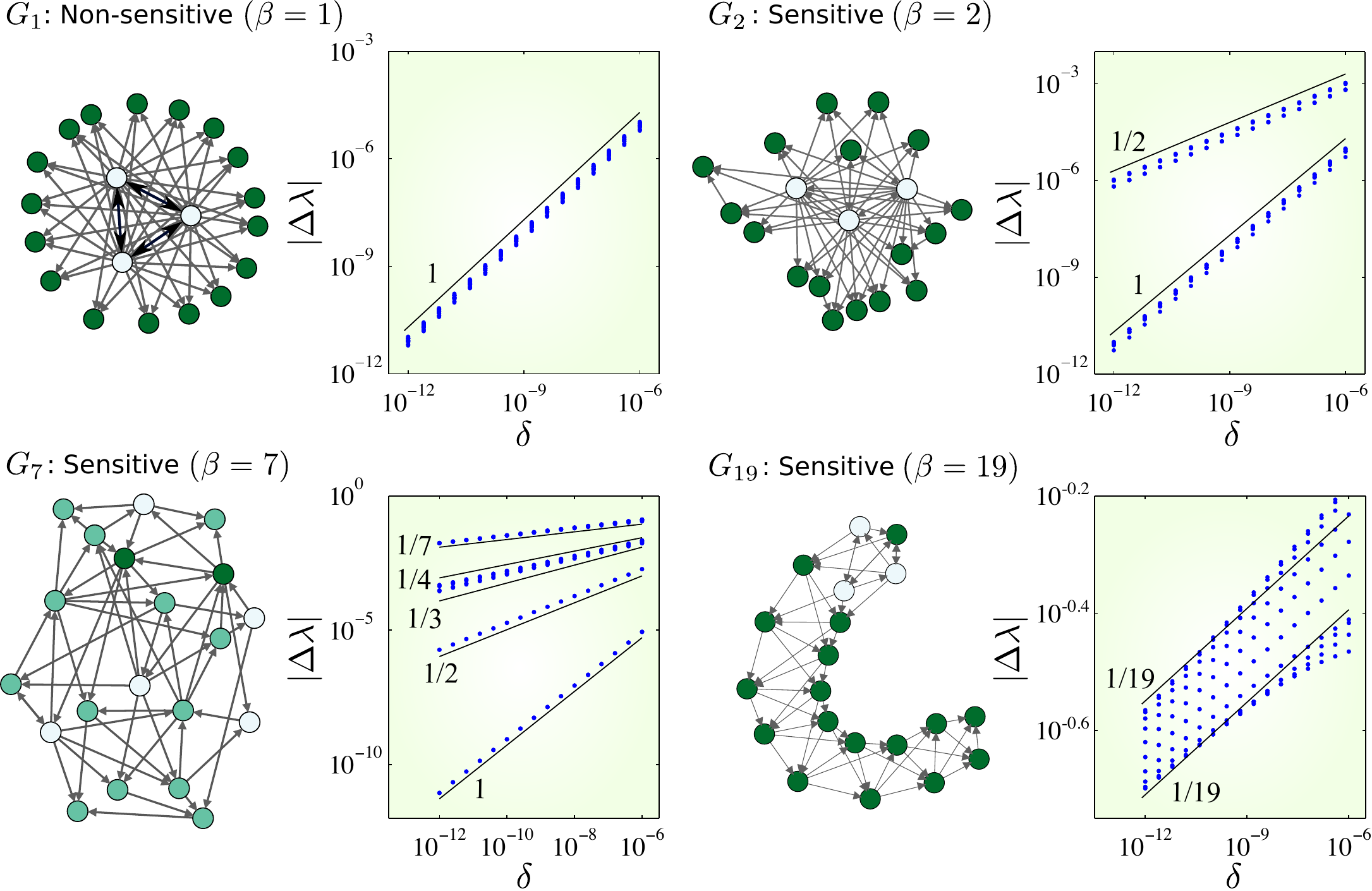}
\end{center}\vspace{-5mm}
\caption{\label{fig:generic_perturbation}
Optimal directed networks with various levels of sensitivity to generic weighted perturbations.
Example networks $G_1$, $G_2$, $G_7$, and $G_{19}$ (each with $20$ nodes and $57$ links) all satisfy $\lambda_2 = \cdots = \lambda_{20}=3$ (and thus are optimal) but have different geometric degeneracy. Perturbing the adjacency matrix of each network as $A + \delta\Delta A$, we plot in double logarithmic scale the resulting change $|\Delta\lambda(\delta)|=|\lambda(\delta)-3|$ for all 19 Laplacian eigenvalues (blue dots, many of which are overlapping). The same randomly chosen $\Delta A$ is used for all four networks, where each $\Delta A_{ij}$ is drawn uniformly from the interval $[-1,1]$ if the link exists from node $j$ to node $i$, and from $[0,1]$ otherwise. The perturbations thus allow small increase and decrease of the weight of existing links, as well as the addition of new links with small weight. In each network, nodes of the same color indicates the same in-degree, and a bidirectional arrow represents two directed links in opposite directions. The three nodes in the center of $G_1$ form a fully connected triangle and each of the other nodes has three in-links from these center nodes. 
In each plot, we observe the expected scaling behavior in Eq.~\eqref{eqn:generic_scaling}, indicated by black lines, each labeled with the corresponding scaling exponent, $1/\beta_j$
($G_1$: $\beta=1$ with $19$ eigenvectors and $\beta_1 = \cdots=\beta_{19}=1$; $G_{2}$: $\beta=2$ with $15$ eigenvectors and $\beta_1 = \cdots=\beta_{11}=1$, $\beta_{12}=\cdots=\beta_{15}=2$; $G_7$: $\beta=7$ with $7$ eigenvectors and $\beta_{1}=\beta_{2}=\beta_{3}=1$, $\beta_{4}=2$, $\beta_{5}=3$, $\beta_{6}=4$, $\beta_{7}=7$; $G_{19}$: $\beta=19$ with $1$ eigenvector and $\beta_1=19$).}
\end{figure*}

\subsection{Classification of networks by their sensitivity}

The general scaling results in the previous section indicate that the overall sensitivity of a Laplacian eigenvalue is determined by its geometric degeneracy $\beta$.
This is because larger $\beta_j$ means more sensitive dependence on $\delta$ in Eq.~\eqref{eqn:generic_scaling} and because $\beta$ is by definition the largest among all the associated $\beta_j$'s.
Thus, we summarize as follows:
\begin{quote}
{\it A Laplacian eigenvalue is sensitive to generic weighted perturbations if and only if the geometric degeneracy $\beta>1$, i.e., the associated eigenvector is degenerate}.
\end{quote}

\subsubsection{Sensitivity in directed networks\label{sec:sensitivity-directed}}

We now show that optimal directed networks are often sensitive to generic weighted perturbations.
Figure~\ref{fig:generic_perturbation} shows examples from the class of optimal networks satisfying $\lambda_2 = \cdots = \lambda_n$.
The geometric degeneracy $\beta$ can be different for different optimal networks in this class and provides a measure of how sensitive an eigenvalue is when $\beta>1$.
Some of these networks are non-sensitive, including simple cases such as the fully connected networks and directed star networks, as well as other networks with more complicated structure, such as the network $G_1$ in Fig.~\ref{fig:generic_perturbation}.
Other optimal networks in this class are sensitive, and
there is a hierarchy of networks having different levels of sensitivity, from $\beta=2$ (e.g., network $G_2$ in Fig.~\ref{fig:generic_perturbation}) all the way up to the maximum possible value $\beta=n-1$ (e.g., network $G_{19}$ in Fig.~\ref{fig:generic_perturbation}), including all intermediate cases (e.g., network $G_7$ in Fig.~\ref{fig:generic_perturbation}).
Such scaling behavior and the resulting sensitivity for $\beta>1$ are robust in the sense that they would be observed even if the associated eigenvector is only approximately degenerate 
(proved in Appendix~\ref{sec:approximate-optimal}).%

How often does an optimal network (including those not satisfying $\lambda_2 = \cdots = \lambda_n$) have $\beta > 1$ and thus exhibit sensitivity?
To study this systematically, we
compute $\beta_j$ symbolically and thus exactly for each Laplacian eigenvalue of all possible directed networks with $n\le 5$.
We find that a large fraction of the Re$(\lambda_2)$-maximizing networks are indeed sensitive due to geometric degeneracy: $44.4$\%, $64.3$\%, and $71.5$\% of them have $\beta>1$ for $n=3$, $4$, and $5$, respectively [red bars in Figs.~\ref{fig:geom-dist-dir}(d)--\ref{fig:geom-dist-dir}(f)].
These fractions are significantly higher than the corresponding fractions among all directed networks (including non-optimal ones): $21.1$\%, $19.7$\%, and $13.7$\%, respectively [blue bars in Figs.~\ref{fig:geom-dist-dir}(d)--\ref{fig:geom-dist-dir}(f)].
Since $\beta$ is bounded by the algebraic degeneracy (multiplicity) of $\lambda_2$, an interesting question is to ask how often $\beta$ attains this bound, giving the network the maximum possible level of sensitivity. Among those networks that are both optimal and sensitive, $74.5$\% and $60.0$\% achieve the maximal sensitivity for $n=4$ and $5$, respectively. 
(The fraction is trivially $100$\% for $n=3$.)
These results thus suggest that optimal directed networks are much more likely to exhibit higher sensitivity than non-optimal ones.

\subsubsection{Non-sensitivity in undirected networks\label{sec:generic-undirected}}

The situation is drastically different when the network is undirected.
For an arbitrary undirected network, for which we have the constraint that the matrix $A$ is symmetric, {\em all of its Laplacian eigenvalues are non-sensitive to any (generic or non-generic) perturbation of the form $A+\delta\Delta A$}, since symmetric matrices are diagonalizable~\cite{golub2013matrix} and thus $\gamma \ge 1/\beta = 1$. 
This in particular implies that there is no sensitivity even for optimal undirected networks, including the UCM and MCC networks.
However, this is not in contradiction with the results in Sec.~\ref{sec:sensitivity-undirected}, as they concern finite-size perturbations (i.e., addition or removal of whole links) in the limit of large networks, while here we consider infinitesimal perturbations on link weights for finite-size networks.

\subsection{Generality of the scaling}

The scaling bound in Eq.~\eqref{eqn:scaling_bound} is applicable to both directed and undirected networks, regardless of whether the links are weighted or unweighted. 
We also expect the scaling in Eq.~\eqref{eqn:generic_scaling} to generically hold true across these classes of networks.
Moreover, while the results for unweighted perturbations in Sec.~\ref{sec:opt} are specific to the Laplacian eigenvalue $\lambda_2$, Eq.~\eqref{eqn:scaling_bound} applies to any eigenvalue of an arbitrary matrix, 
including the adjacency matrix and any other matrix that may characterize a particular system. 
For example, the largest eigenvalue (in absolute value) of the adjacency matrix for a strongly connected (directed) network is non-degenerate 
(by, e.g., the Perron-Frobenious Theorem~\cite{Brualdi:2008})
and therefore non-sensitive.
In general, the degree to which the scaling holds is likely to be related to the normality of the matrix, which can range from completely normal matrices with orthogonal eigenvectors (as in undirected networks) to highly non-normal matrices with parallel, degenerate eigenvectors (as in many optimal networks)~\cite{Milanese2010:kgd,trefethen2005spectra}. 
The result in Appendix~\ref{sec:approximate-optimal} implies that the network does not need to be perfectly degenerate,
which opens the door for observing the sensitivity we identified in real-world applications where exact degeneracy is 
unlikely~\cite{PhysRevLett.107.034102}.
Combining all these with the tendency of optimization to cause geometric degeneracy and with
the wide range of systems that can be described by 
Eq.~\eqref{eqn:coupled-syst}, 
we expect to observe sensitivity to 
weighted
perturbations in many applications.

\section{Sensitivity in example physical systems}
\label{sec:phys-examples}

As summarized in Fig.~\ref{fig:summary}, we have established two cases in which sensitive dependence on network structure arises: undirected networks under 
unweighted
perturbations (Sec.~\ref{sec:sens-opt-undirected}) and directed networks under 
weighted
perturbations (Sec.~\ref{sec:sensitivity-directed}). 
Here we discuss implications of these cases for concrete examples of physical networked systems. 

\begin{figure*}[t]
\begin{center}
\includegraphics[width=0.9\textwidth]{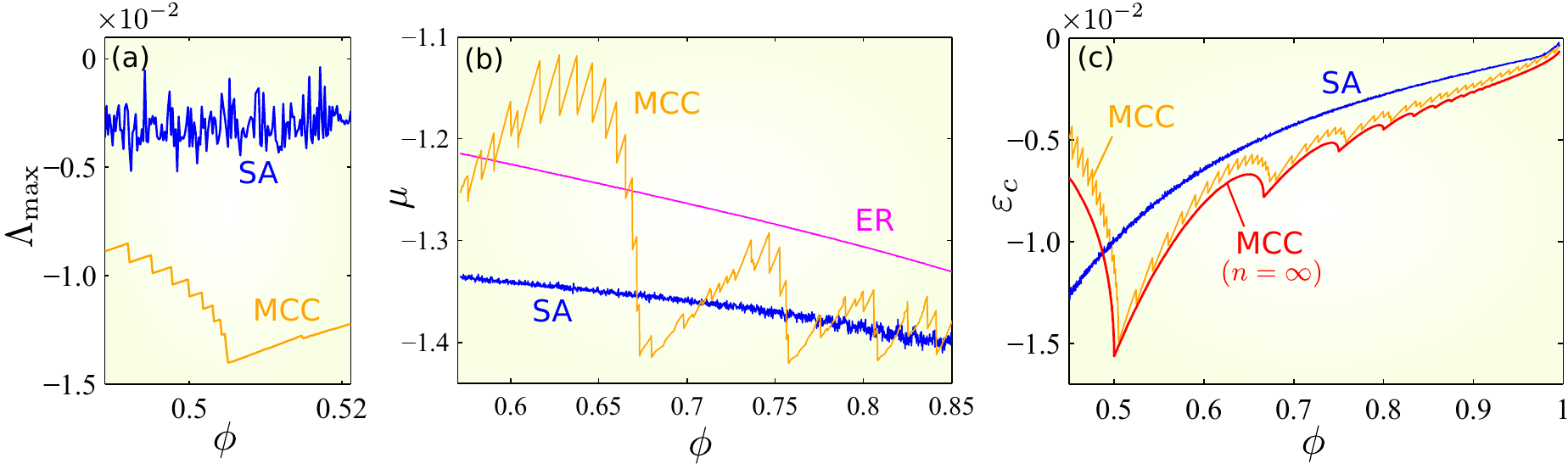}
\end{center}\vspace{-5mm}
\caption{\label{fig:example-systems}
Sensitive dependence on link density $\phi$ in physical examples of undirected networks.
(a) Exponential rate of convergence $\Lambda_\text{max}$ to a synchronous state of power-grid networks.
(b) Mean finite-time convergence rate $\mu$ toward synchronization in networks of optoelectronic oscillators.
(c) Critical diffusivity threshold $\eps_c$ for Turing instability in networks of activator-inhibitor systems.
Description of these three systems can be found in Supplemental Material~\cite{sm}, Secs.~\ref{si:sec:example-systems}A--\ref{si:sec:example-systems}C.
The orange curves indicate the values of these quantities for the MCC networks with $n=100$ constructed by the procedure described in  Appendix~\ref{si:sec:construction_MCC}, while the red curve in panel~(c) is the finite-$n$ approximation obtained from the asymptotic formula in Eq.~\eqref{eqn:lambda2inf}.
The blue curves indicate the corresponding values for networks found by SA.
The magenta curve in panel~(b) is the mean value for the ER random networks estimated from $1{,}000$ realizations, while in panels (a) and (c) the values are relative to the corresponding mean value for the ER random networks (and thus zero corresponds to the ER mean value).}
\end{figure*}

\subsection{Undirected networks under unweighted perturbations}

For undirected networks, the sensitivity of $\lambda_2$ observed for UCM networks is relevant for a wide range of networked systems, since the stability function formalism establishes that, in many systems, $\lambda_2$ determines the stability properties of relevant network-homogeneous states.
Typically the asymptotic rate of convergence $|\Lambda_\text{max}|$ is a smooth, monotonically increasing function of $\lambda_2$ 
(concrete examples given in Supplemental Material~\cite{sm}, Secs.~\ref{si:sec:example-systems}C--\ref{si:sec:example-systems}F), 
and thus the maximized convergence rate exhibits 
sensitivity.
Below we list specific cases in which sensitivity is observed in $|\Lambda_\text{max}|$ or a related quantity:
\begin{enumerate}[topsep=3pt, partopsep=0pt, leftmargin=12pt]
\item \textit{Convergence rate}. For networks of phase oscillators, including models of power-grid networks, the convergence rate to a frequency-synchronized, phase-locked state is a function of the Laplacian eigenvalue $\widetilde{\lambda}_2$ associated with an effective interaction matrix $\widetilde{A}$ for the system 
(details presented in Supplemental Material~\cite{sm}, Sec.~\ref{si:sec:example-systems}A). 
While $\lambda_2$ is generally different from $\widetilde{\lambda}_2$, it is strongly correlated with $\widetilde{\lambda}_2$, and hence 
with
$\Lambda_\text{max}$. We thus expect to observe sensitive dependence of $\Lambda_\text{max}$, which is indeed confirmed in Fig.~\ref{fig:example-systems}(a) for power-grid networks with a prescribed network topology and realistic parameters for the generators and other electrical components in the system.
\item \textit{Transient dynamics}. In addition to the asymptotic convergence rate $\Lambda_\text{max}$, sensitive dependence can be observed for the convergence rate in the transient dynamics of the network, which depends  not only on $\lambda_2$ but on all Laplacian eigenvalues. This is illustrated in Fig.~\ref{fig:example-systems}(b) using the example of coupled optoelectronic oscillator networks 
(system details described in Supplemental Material~\cite{sm}, Sec.~\ref{si:sec:example-systems}B).
\item \textit{Critical coupling threshold}. Another physical quantity that can exhibit sensitive dependence is the critical coupling threshold for the stability of the network-homogeneous state in systems with a global coupling strength $\eps$. In such systems, the functions $\mathbf{H}_{ij}$ are proportional to $\eps$.
For identical oscillators capable of chaotic synchronization, the minimum coupling strength for stable synchronization is inversely proportional to $\lambda_2$.
For the activator-inhibitor systems~\cite{Nakao:2010fk}, the parameter $\eps$ is interpreted as the common diffusivity constant associated with the process of diffusion over individual links.
As $\eps$ is decreased from a value sufficiently large for the uniform concentration state to be stable, there is a critical diffusivity, $\eps = \eps_c$, corresponding to the onset of Turing instability. This $\eps_c$ is inversely proportional to $\lambda_2$ 
(derivation given in Supplemental Material~\cite{sm}, Sec.~\ref{si:sec:example-systems}C).
Such a critical threshold thus depends sensitively on the link density of the network [as illustrated 
in Fig.~\ref{fig:example-systems}(c)] as well as on the number of nodes.
\end{enumerate}

\begin{figure*}[tb]
\begin{center}
\includegraphics[width=0.8\textwidth]{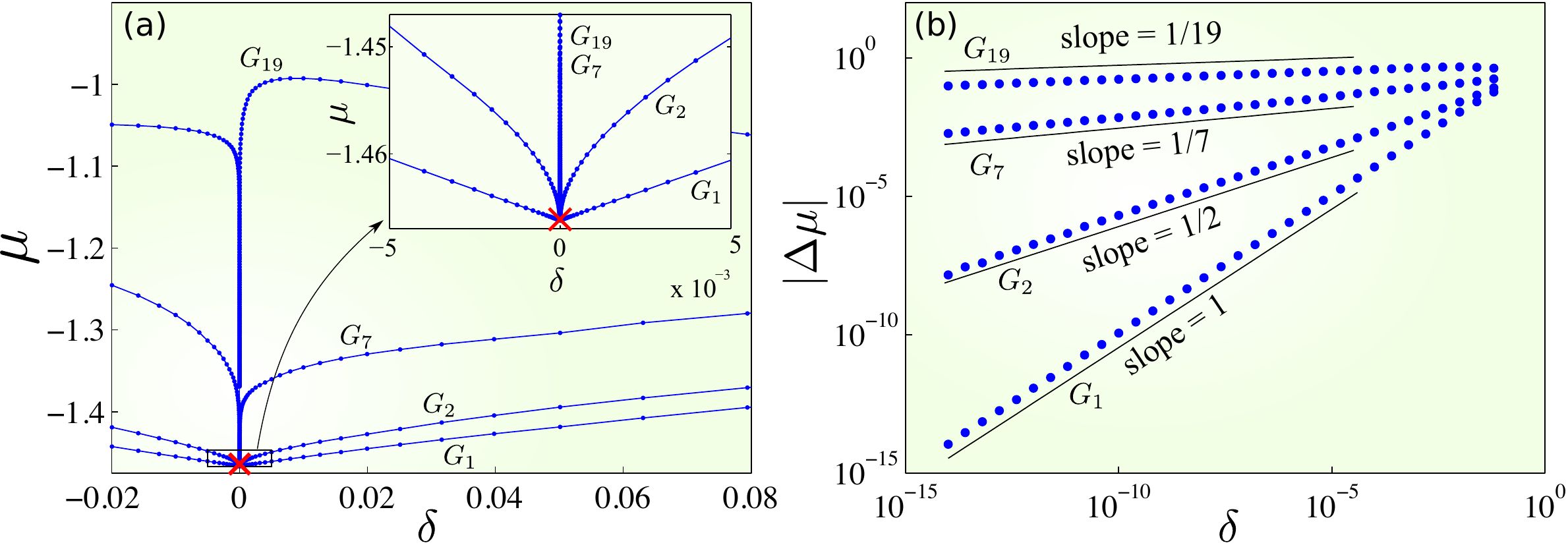}
\end{center}\vspace{-4mm}
\caption{\label{fig:generic_perturbation_phys}
Sensitivity to 
weighted
perturbations in directed networks of optoelectronic oscillators.
(a) Mean convergence rate $\mu$ as a function of $\delta$, illustrating the qualitative difference
between sensitive networks ($G_2$, $G_7$, and $G_{19}$ from Fig.~\ref{fig:generic_perturbation}) and non-sensitive networks ($G_1$ from Fig.~\ref{fig:generic_perturbation}).
The red cross symbol (\emph{``$\times$''}) indicates the value of $\mu$ at $\delta = 0$ corresponding to the case of no perturbation, which is the same for all four networks.
The inset shows a zoom-in plot of the marked rectangular region surrounding the red cross.
The perturbation matrix $\Delta A$ was chosen randomly following the same procedure used in Fig.~\ref{fig:generic_perturbation}.  To facilitate visualization, in this figure we allow negative $\delta$, which corresponds to considering a perturbation term of the form $|\delta|(-\Delta A)$.
(b) Log-log plot of the change in convergence rate $|\Delta\mu|$ versus $\delta$, which confirms the scaling $|\Delta\mu| \sim \delta^{1/\beta}$ for small $\delta$.}
\end{figure*}

\subsection{Directed networks under weighted perturbations}

For directed networks, the sensitivity of Laplacian eigenvalues under generic perturbations
is typically inherited by the convergence rate $\Lambda_\text{max}$
for many systems and processes governed by Eq.~\eqref{eqn:coupled-syst}, including most of the examples described in Supplemental Material~\cite{sm}, Sec.~\ref{si:sec:example-systems}.
In fact, $\Lambda_\text{max}$ would have the same sensitivity as the Laplacian eigenvalue $\lambda_j$ whenever $\Lambda_\text{max}$ has a smooth (non-constant) dependence on $\lambda_j$ near the unperturbed values of $\lambda_j$.
Figure~\ref{fig:generic_perturbation_phys} illustrates the sharp contrast between sensitive and non-sensitive cases using the example of synchronization in networks of chaotic optoelectronic oscillators~\cite{PhysRevLett.107.034102} 
(system details described in Supplemental Material~\cite{sm}, Sec.~\ref{si:sec:example-systems}B).

\section{Discussion}
\label{sec:discussion}

The sensitive dependence of collective dynamics on the network structure, characterized here by a derivative that diverges at an optimal point, has several implications.  On the one hand, it implies that the dynamics can be manipulated substantially by small structural adjustments, which we suggest has the potential to lead to new control approaches based on modifying the effective structure of the network in real time; indeed, the closer the system is to being optimal, the larger the range of manipulation possible with the same amount of structural adjustment.  On the other hand, the observed cusp-like behavior imposes constraints on how close one can get to the ultimate optimum in practice, given unavoidable parameter mismatches, resolution
limits, and numerical uncertainty. 
 
It is insightful to interpret our results in the context of living systems.  The apparent conundrum that follows from this study is that biological networks (such as genetic, neuronal, and ecological ones) are believed to have evolved under the pressure to both optimize fitness and be robust to structural perturbations~\cite{kitano2004biological}.  The latter means that the networks would not undergo significant loss of function (hence, of optimality) when perturbed.  For example, a mutation in a bacterium (i.e., a structural change to a genetic network) causes the resulting strain to be nonviable in only a minority of cases~\cite{keio_collection}. A plausible explanation is that much of the robustness of living systems comes from the plasticity they acquire from optimizing their fitness under varying conditions~\cite{Stelling:2004,Yang:2015}.  In the case of bacterial organisms, for example, it is believed that the reason most of their genes are not essential for a given environmental condition is because they are required under different conditions.  Bacteria kept under stable conditions, such as 
those that live inside other living cells (i.e., intracellular bacteria),
have evolved to virtually have only those genes essential under that 
condition~\cite{Glass2006} and are thus sensitive to gene removals;
they are a close analog of the optimization of a fixed objective function considered here~\cite{comment2}.  
While there is therefore no conflict between our results and the optimization-robustness trade-off expected for biological networks, investigating the equivalent of the sensitive dependence on network structure in the case of varying conditions or varying objective function would likely provide further insights.

In general, the optimization-robustness relation may depend on the type of robustness 
considered.
In this study we focused on how stable a state is, and hence on how resistant the network is to small changes in its dynamical 
state, which can be regarded as a form of robustness (terminology used, for example, in Ref.~\cite{Bar-Yam30032004}).
It is quite remarkable that, in seeking to optimize 
the network for 
this ``dynamical'' robustness, 
the network
would lose ``structural'' 
robustness, where the latter is a measure of how resistant the stability of the network state is to changes in the network structure.
But is the observed sensitive dependence on network structure really a sign of non-robustness?  The answer is both yes and no.  It is ``yes'' in the sense that, because of the non-differentiability of this dependence, small parameter changes cause stability to change significantly.  It is ``no'' in the sense that, because the cusps appear at valleys rather than at peaks, the stability in the vicinity of the local best parameter choices are still generally better than at locations farther away
(that is, specific parameters lead to significant improvement but not to significant deterioration).  By considering both the dynamical and the structural robustness in the sense above, we can interpret our results as a manifestation of the ``robust-yet-fragile'' property that has been suggested as a general feature of complex systems~\cite{carlson1999highly}.

Finally, it is instructive to compare 
sensitive dependence on network structure with the phenomenon of chaos, which can exhibit multiple forms of sensitive dependence~\cite{OttChaosBook}. Sensitive dependence on initial conditions, where small changes in the initial state lead to large changes in the subsequent evolution of the state, is a phenomenon that concerns trajectories in the phase space of a fixed system. Sensitive dependence on parameters may concern a similar change in trajectories across different systems even when the initial conditions are the same, as in the case of the map $\theta_{n+1} = 2\theta_n + c$ (mod $2\pi$) when  $c$ rather than $\theta_0$ is changed. But sensitive dependence on parameters may also concern a change in the nature of the dynamics, which has a qualitative rather than 
merely
quantitative impact on the trajectories; this is the case for the logistic map $x_{n+1} = rx_n(1-x_n)$, whose behavior can change from chaotic to periodic by arbitrarily small changes in $r$ and, moreover, whose Lyapunov exponent exhibits a cusp-like dependence on $r$ within each periodic window.  The latter concerns sensitive dependence of the stability (or the level of stability) of the states under consideration, and therefore is a low-dimensional analog of the sensitive dependence of network dynamics on network structural parameters investigated here.  In the case of networks, however, they emerge not from bifurcations but instead from optimization.  Much in the same way the discovery of sensitive dependence on initial conditions in the context of (what is now known as) chaos sets constraints on long-term predictability and on the reliability of simple models for weather forecast~\cite{Lorenz63}, the sensitive dependence on network structure calls for a careful evaluation of the constraints it sets on predictability and model reliability~\cite{Babtie1:2014} in the presence of noise and uncertainties in real network systems.
We thus believe that the interplay between network structure, optimization, sensitivity, and robustness is a promising topic of future research that can offer fundamental insights into the 
properties
of complex systems.

\begin{acknowledgments}
This work was supported by the U.S. National Science Foundation under Grant No.~DMS-1057128, by the U.S. Army Research Office under Grants No.~W911NF-15-1-0272, No.~W911NF-12-1-0276, and No.~W911NF-16-1-0081, and by the Simons Foundation under Grant No.~318812.
\end{acknowledgments}

\appendix

\section{\texorpdfstring{Derivation of the stability function $\boldsymbol{\Lambda(\alpha)}$}{Derivation of the stability function Lambda(alpha)}}\label{sec:stab_func}

The two assumptions we make in Sec.~\ref{sec:class-systems} regarding the coupling functions $H_{ij}$ can be mathematically formulated as follows:
\begin{itemize}[topsep=6pt, partopsep=0pt, leftmargin=26pt]
\item {\it Formulation of assumption (A-1).}
$D_{\mathbf{u}}\mathbf{H}_{ij}(\mathbf{x}^*,\mathbf{x}^*) = - D_{\mathbf{v}}\mathbf{H}_{ij}(\mathbf{x}^*,\mathbf{x}^*)$, where $D_{\mathbf{u}}\mathbf{H}_{ij}$ and $D_{\mathbf{v}}\mathbf{H}_{ij}$ denote the derivatives with respect to the first and second argument, respectively, of the function $\mathbf{H}_{ij}(\mathbf{u}, \mathbf{v})$. We also assume $\mathbf{H}_{ij}(\mathbf{u},\mathbf{u}) = 0$ for all $\mathbf{u}$, which ensures that the network-homogeneous state is a valid solution of Eq.~\eqref{eqn:coupled-syst}.
Together, these assumptions are equivalent to 
assuming that $\mathbf{H}_{ij}$ can be approximated as $\mathbf{H}_{ij}(\mathbf{u},\mathbf{v}) \approx D_{\mathbf{v}}\mathbf{H}_{ij}(\mathbf{x}^*,\mathbf{x}^*) \cdot (\mathbf{v} - \mathbf{u})$ to the first order in $\mathbf{v} - \mathbf{u}$.
\item {\it Formulation of assumption (A-2).} 
$D_{\mathbf{v}}\mathbf{H}_{ij}(\mathbf{x}^*,\mathbf{x}^*) = A_{ij} \cdot \mathbf{G}(t)$, where the scalar $A_{ij}$ is independent of $t$, and the function $\mathbf{G}(t)$ is independent of $i$ and $j$.
\end{itemize}
Under these assumptions, the variational equation of the system~\eqref{eqn:coupled-syst} around a given network-homogeneous state $\mathbf{x}_1 = \cdots = \mathbf{x}_n = \mathbf{x}^*(t)$ becomes
\begin{equation}\label{eqn:var-eqn}
\dot{\boldsymbol{\xi}}_i 
= D_{\mathbf{x}}\mathbf{F} \cdot \boldsymbol{\xi}_i 
 - D_{\mathbf{y}}\mathbf{F} \cdot
 \sum_{j=1}^n L_{ij} \mathbf{G}(t) \cdot \boldsymbol{\xi}_j,
\end{equation}
where $\boldsymbol{\xi}_i$ is the perturbation to the state of 
node
$i$, $D_{\mathbf{x}}\mathbf{F}$ and $D_{\mathbf{y}}\mathbf{F}$ are the derivatives of the function $\mathbf{F}$ with respect to the first and any of the other arguments, respectively, evaluated at $(\mathbf{x},\mathbf{y}_1,\ldots,\mathbf{y}_n) = (\mathbf{x}^*,\mathbf{0},\ldots,\mathbf{0})$, and $L = (L_{ij})$ is the Laplacian matrix of the network given by Eq.~\eqref{eqn:laplacian}.
An argument based on the Jordan canonical form of $L$ similar to the one used in Ref.~\cite{Nishikawa:2006kx} then leads to a stability function $\Lambda(\alpha)$, defined for given (complex-valued) auxiliary parameter $\alpha$ as the maximum Lyapunov exponent of the solution $\boldsymbol{\eta} = \mathbf{0}$ of
\begin{equation}\label{eqn:mse}
\dot{\boldsymbol{\eta}}
= [D_{\mathbf{x}}\mathbf{F} -\alpha D_{\mathbf{y}}\mathbf{F} \cdot \mathbf{G}(t)]\boldsymbol{\eta}.
\end{equation}
The
exponential rate of convergence or divergence is then given by $\Lambda(\lambda_j)$ for the perturbation mode corresponding to the $j$th (possibly complex) eigenvalue $\lambda_j$ of the Laplacian matrix $L$.
Thus, the perturbation mode with the slowest convergence (or fastest divergence) determines the stability of the network-homogeneous state through $\Lambda_{\max}$ defined in Eq.~\eqref{eqn:lambda-max}.
A key aspect of this approach is that the functional form of $\Lambda(\alpha)$ does not depend on the network structure, implying that the network structure influences the stability only through the Laplacian eigenvalues~\cite{Pecora:1998zp}. 

For a system with a global coupling strength parameter $\eps$, such as the networks of identical oscillators and networks of activator-inhibitor systems described in 
Secs.~\ref{si:sec:example-systems}B and \ref{si:sec:example-systems}C of Supplemental Material~\cite{sm}, respectively,
the derivative $D_{\mathbf{v}}\mathbf{H}_{ij}(\mathbf{x}^*,\mathbf{x}^*)$ in the condition (A-2) above is proportional to $\eps$, and $\mathbf{G}(t)$ can be chosen to include the factor $\eps$ [thus making the stability function $\Lambda(\alpha)=\Lambda_{\eps}(\alpha)$ dependent on $\eps$].
We note that the class of systems treated in Ref.~\cite{Pecora:1998zp} is an important special case of our formulation in which $\mathbf{F}(\mathbf{x},\mathbf{y}_1,\ldots,\mathbf{y}_n)$ depends linearly on the $\mathbf{y}$ variables and the coupling function $\mathbf{H}_{ij}$ is proportional to the difference in (some function of) the state of the nodes 
(details presented in Sec.~\ref{si:sec:example-systems}B of Supplemental Material~\cite{sm}).
We also note that the same stability condition $\Lambda_{\max} \le 0$ is derived in Ref.~\cite{PhysRevE.61.5080} for a general class of systems that is different from the class of systems treated here.
An advantage of our formulation is that the assumptions on the nature of pairwise interactions encoded in the coupling functions $\mathbf{H}_{ij}$ are intuitive and have clear relation to the network structure encoded in the adjacency matrix $A$.

\section{Uniqueness of networks attaining the bound\label{sec:ucm_optimality}}

Here we show that, if the mean degree $\bar{d} := 2m/n$ of the network is a (non-negative) integer, the UCM network is the only one that attains the bound in Eq.~\eqref{bound} among all networks with the same $n$ and $m$.
For $n = k\ell$ and $m = k^2 \ell(\ell - 1)/2$, this claim implies that the UCM network is the only $\lambda_2$ optimizer.
For other combinations of $n$ and $m$, no UCM network exists, and the claim implies that there is no network that can achieve the upper bound.

To prove the claim, we assume that the network attains the bound, i.e., $\lambda_2 = \lfloor 2m/n \rfloor  = \bar{d}$, and aim to show that it must be a UCM network.
We first observe that $\lambda_n^c = n - \lambda_2 = n - \bar{d}$.
Also, since $\bar{d}$ is an integer, so is the mean degree of the complement, $\bar{d}^c = (n-1) - \bar{d}$, and thus Eq.~\eqref{eqn:bound-proof} becomes
\begin{equation}\label{eqn:bound-proof2}
n - \bar{d} = \lambda_n^c \ge d_\text{max}^c + 1 \ge \bar{d}^c + 1 = n - \bar{d}.
\end{equation}
Since this implies that the maximum and the mean degree of the complement match, i.e., $d_\text{max}^c = \bar{d}^c =: d^c$, all nodes must have the same degree $d^c$ in the complement.
Equation~\eqref{eqn:bound-proof2} also implies $\lambda_n^c = d^c + 1$.
Next we consider an arbitrary connected component of the complement and show that its maximum Laplacian eigenvalue equals $d^c + 1$.
On the one hand, 
since the Laplacian spectrum of any network is the union of the Laplacian spectra of its connected components (stated and proved as Proposition~\ref{prop:union} in Supplemental Material~\cite{sm}, Sec.~\ref{si:sec:prop}B),
we see that the maximum Laplacian eigenvalue of this component is at most $\lambda_n^c$ ($= d^c + 1$).
On the other hand, by applying Eq.~\eqref{prop393} to the component and noting that its maximum degree is $d^c$, we see that its maximum Laplacian eigenvalue is at least $d^c + 1$.
Combining these, we conclude that the maximum Laplacian eigenvalue of this component equals $d^c + 1$. 
We now
use the part of Proposition 3.9.3 in Ref.~\cite{brouwer2012spectra} stating that the equality in Eq.~\eqref{prop393} holds true only if $d_\text{max} + 1 = n$.
Applying this to the component and combining with the result above, we see that the component size must be $k := d^c + 1$.
Since each node has degree $d^c$, the component must be fully connected. 
Since the choice of the component was arbitrary, the same holds true for all components in the complement, implying that they form $\ell$ isolated, fully connected clusters of size $k$ (for some positive integer $\ell$).
Therefore, the network 
must be
a UCM network.

\section{Explicit construction of MCC networks}\label{si:sec:construction_MCC}

To construct an MCC network for given $n$ and $m$,
we first compute the function 
$M(n,k)$, which we recall is 
the maximum number of links possible for a network of size $n$ when the largest size of connected components is $\le k$. 
For a given $k$, 
the maximum number of fully connected clusters of size $k$ that one can form with $n$ nodes is 
$\ell := \lfloor n/k\rfloor$. 
Forming 
$\ell$ 
such clusters requires 
$\ell\cdot k(k-1)/2$ 
links, and completely connecting the remaining 
$n_r := n - k\ell$ 
nodes requires $n_r(n_r - 1)/2$ links. Since any additional link would necessarily make the size of some component greater than $k$, this network has the maximum possible number of links, and we thus have
\begin{equation}\label{eqn:Mnk}
M(n,k) = \ell\cdot\frac{k(k-1)}{2} + \frac{n_r(n_r - 1)}{2}
\end{equation}
(proof given in Supplemental Material~\cite{sm}, Sec.~\ref{si:sec:prop}A).
This formula allow us to compute $M(n,k)$ for each $k = 1,\ldots,n$.  The computed $M(n,k)$ can then be used to
determine $k_{n,m_c}$ for 
the
given $m$ 
directly from the definition: $k_{n,m_c}$ is
the smallest integer $k$ for which $m_c \le M(n,k)$, where $m_c = \frac{n(n-1)}{2} - m$. 

The complement of an MCC network 
is then constructed so as to have as many fully connected clusters of size $k_{n,m_c}$ as possible using all the $m_c$ available links. If one or more links remain, we recursively apply the procedure to these links and the set of remaining isolated nodes. If no cluster of size $k_{n,m_c}$ can be formed (which occurs only when $k_{n,m_c} \ge 3$), we first construct a fully connected cluster of size $k_{n,m_c}-1$, which is always possible since $m_c > M(n,k_{n,m_c}-1) \ge (k_{n,m_c}-1)(k_{n,m_c}-2)/2$ by the definition of $k_{n,m_c}$. We then connect the remaining links arbitrarily while ensuring that the size of the largest connected component is $k_{n,m_c}$. The resulting Laplacian eigenvalues are independent of the configuration of these links, since all possible configurations are equivalent up to permutation of node indices. 
The procedure thus generates an MCC network with the given number of nodes and links, $n$ and $m$, respectively.
Note that, in the special case of $n = k\ell$ and $m = k^2 \ell(\ell - 1)/2$ with given positive integers $\ell$ and $k$, the procedure described here results in the UCM network with $\ell$ groups of size $k$, as it is the only MCC network in that case.
A MATLAB implementation for the procedure [including the relevant functions such as $M(n,k)$ and $k_{n,m}$] is available for download~\cite{software}.

\section{\texorpdfstring{Lower bound for maximum $\text{Re}\boldsymbol{(\lambda_2)}$}{Lower bound for maximum Re(lambda2)}\label{si:sec:non-sensitive-directed}}

Here we show that the network constructed in Sec.~\ref{sec:sensitivity-opt-dir} to establish the lower bound satisfies
\begin{equation}\label{eq:kDSeigs}
\begin{split}
&\lambda_1 = 0, \quad \lambda_2 =\dots=\lambda_{n-r}=s,\\
&\lambda_{n-r+1} =\dots=\lambda_{n}=s+1,
\end{split}
\end{equation}
which in particular implies that $\lambda_2 = s$.
We first note that 
$\bar{\lambda}-1 < s \le \bar{\lambda}$, since we have $s =\lfloor \bar{\lambda} \rfloor$ by definition.
From the definition of $r$, we can write $r = m_d-s(n-1) = \phi n(n-1)-s(n-1) = (\bar{\lambda}-s)(n-1)$.
Combining these, we see that $0 \le r \le n-2$.
We thus divide the proof into two cases: $0\leq r\leq s-1$ and $s\leq r\leq n-2$.
In the following, we use the notation $O_{n_1 \times n_2}$ for the zero matrix of size $n_1 \times n_2$ and $I_{n_1 \times n_1}$ for the identity matrix of size $n_1$.

\smallskip\noindent
\textbf{Case 1:}\, If $0\leq r\leq s-1$, the matrix $L$ has the lower block triangular form
\begin{equation}\label{eqn:block-L-1}
L=\begin{pmatrix}
L'_{n_1 \times n_1} & O_{n_1 \times n_2} \\
B_{n_2 \times n_1} & s I_{n_2 \times n_2}
\end{pmatrix},
\end{equation}
where we use the notations $n_1=s+1$ and $n_2=n-s-1$.
Here $L'_{n_1 \times n_1}$ and $B_{n_2 \times n_1}$ are matrices of size $n_1 \times n_1$ and $n_2 \times n_1$, respectively.
The set of eigenvalues of $L$ is thus the union of the set of eigenvalues of $L'$ and the set $\{s,\dots,s\}$ (repeated $n_2$ times, owing to the diagonal block $s I_{n_2 \times n_2}$). 
To obtain the eigenvalues of $L'_{n_1 \times n_1}$, we apply a sequence of row operations to the matrix $L'_{n_1 \times n_1}-\lambda I_{n_1 \times n_1}$.
Denoting the $i$th row of this matrix by $R_i$, we first replace $R_i$ with $R_i-R_{s+1}$ for each $i=1,\ldots,s$, and then replace $R_{s+1}$ with $R_{s+1}+\sum_{i=1}^{r}R_i/(s+1-\lambda)+\sum_{i=r+1}^{s}R_i/(s-\lambda)$ [or with $R_{s+1}+\sum_{i=1}^{s}R_i/(s-\lambda)$, if $r=0$].
Because of the specific form of $L'_{n_1 \times n_1}-\lambda I_{n_1 \times n_1}$, this results in an upper triangular matrix whose diagonal elements are $s+1-\lambda$ (first $r$), $s-\lambda$ (next $s-r$), and $-\lambda$.
Since none of these row operations involve switching two rows or multiplying a row by a nonzero constant, the determinant is invariant, and hence the eigenvalues of $L'_{n_1 \times n_1}$ are $s+1$ (repeated $r$ times), $s$ (repeated $s-r$ times), and $0$ (simple).
Combining with the $n_2$ repetitions of $s$ from the block $s I_{n_2 \times n_2}$ in Eq.~\eqref{eqn:block-L-1}, the eigenvalues of $L$ are $0$ (simple), $s$ (repeated $s-r+n_2=n-r-1$ times), and $s+1$ (repeated $r$ times), satisfying Eq.~\eqref{eq:kDSeigs}.

\smallskip\noindent
\textbf{Case 2:}\, If $s\leq r\leq n-2$, the matrix $L$ has the lower block triangular form
\begin{equation}
L=\begin{pmatrix}
K_{n_1 \times n_1} & O_{n_1 \times n_2} & O_{n_1 \times n_3} \\
B'_{n_2 \times n_1} & (s+1) I_{n_2 \times n_2} & O_{n_2 \times n_3}\\
B''_{n_3 \times n_1} & O_{n_3 \times n_2} & sI_{n_3 \times n_3}
\end{pmatrix},
\end{equation}
where we use the notations $n_1=s+1$ and $n_2=r-s$, and $n_3=n-r-1$.
Here, $K_{n_1 \times n_1}$ is the $n_1 \times n_1$ Laplacian matrix of a complete graph of $n_1$ nodes, with eigenvalues $0$ (simple) and $n_1 = s+1$ (repeated $n_1-1$ times). Therefore, the eigenvalues of matrix $L$ are $0$ (simple), $s$ (repeated $n_3=n-r-1$ times), and $s+1$ (repeated $n_1-1+n_2=r$ times), satisfying Eq.~\eqref{eq:kDSeigs}.

\section{Scaling for eigenvalues with geometric degeneracy\label{si:sec:mat-pert}}

We can
establish Eq.~\eqref{eqn:scaling_bound} for an arbitrary eigenvalue
of an arbitrary matrix. 
Given an $n \times n$ matrix $M$, its Jordan decomposition can be written as
\begin{equation}\label{eqn:sim-trans}
	M = PJP^{-1},
\end{equation}
where $P$ is an invertible matrix and 
\begin{equation}
	J = 
 \begin{pmatrix}
  J^{(1)} & \cdots & 0 \\
  \vdots  & \ddots & \vdots  \\
  0 & \cdots & J^{(p)}
 \end{pmatrix}
\end{equation}
is the block-diagonal Jordan matrix with $p$ Jordan blocks~\cite{Filippov1971}. 
The $j$th Jordan block is of size $\beta_j \times \beta_j$ and has the form
\begin{equation}
	J^{(j)} = 
	 \begin{pmatrix}
	  \lambda_j & 1 & 0 & \cdots & 0 \\
	   0 & \lambda_j & 1 & \cdots & 0 \\
	   \vdots & \vdots & \ddots & \ddots & \vdots \\	   
	   0 & 0 & \cdots & \lambda_j & 1 \\
	   0 & 0 & \cdots & 0 & \lambda_j
\end{pmatrix}.
\end{equation}
Since Eq.~\eqref{eqn:sim-trans} is a similarity transformation, the eigenvalues of $M$ are  the same as those of $J$, which are the diagonal elements $\lambda_1,\lambda_2,\dots,\lambda_p$ of $J$ with corresponding multiplicities $\beta_1, \beta_2,\ldots,\beta_p$, respectively.
Note that $\beta_j$ can be smaller than the algebraic multiplicity of $\lambda_j$, since we may have $\lambda_j = \lambda_{j'}$ for some $j \neq j'$.

As in the main text, we
consider the matrix perturbation of the form $\hat{M}(\delta) = M + \delta\Delta M$, where $\delta>0$ and $\Delta M$ is an $n\times n$ matrix. 
For a given eigenvalue $\lambda$ of $M$, let $\alpha$ and $\beta$ denote its algebraic and geometric degeneracy, respectively.
The geometric degeneracy is defined as the size of the largest Jordan block associated with $\lambda$, or equivalently, as the largest number of repetitions of $\lambda$ associated with the same eigenvector.
Since the roots of a polynomial depend continuously on the coefficients, each eigenvalue of a matrix changes continuously as the elements of that matrix change~\cite{Horn1985}. 
Therefore, there are exactly $\alpha$ eigenvalues of the matrix $\hat{M}(\delta)$ that approach $\lambda$ as $\delta\rightarrow{0}$.
Below we prove that there exists a constant $C\geq{0}$ such that Eq.~\eqref{eqn:scaling_bound} holds true for each eigenvalue $\hat\lambda(\delta)$ of $\hat{M}(\delta)$ that converges to $\lambda$, where we denote $\Delta\lambda(\delta) := \hat\lambda(\delta)-\lambda$.

We
first use the same $P$ that transforms $M$ into $J$ in Eq.~\eqref{eqn:sim-trans} to transform $\hat{M}(\delta)$ for each $\delta$ as
\begin{equation}
	P^{-1}\hat{M}(\delta)P = J + \delta Q,
\end{equation}
where $Q$ is the matrix given by $Q:=P^{-1}\Delta M P$. 
Thus, the eigenvalues of $\hat{M}(\delta)$ are the same as those of $J+\delta Q$. 
To further transform the matrix, consider the block-diagonal matrix
\begin{equation}
T = 
 \begin{pmatrix}
  T^{(1)} & \cdots & 0 \\
  \vdots  & \ddots & \vdots  \\
  0 & \cdots & T^{(p)}
 \end{pmatrix},
\end{equation}
where the $j$th block $T^{(j)}$ is a $\beta_j \times \beta_j$ diagonal matrix with elements $T^{(j)}_{ii} = \delta^{-1 + i/\beta_j}$, $1\leq{i}\leq\beta_j$.
The matrix $T$ is invertible for all $\delta\neq{0}$. Therefore, the eigenvalues of $\hat{M}(\delta)$ are the same as those of the matrix
\begin{equation}\label{eqn:transformed}
	T^{-1}P^{-1}\hat{M}(\delta)PT = T^{-1}JT + \delta T^{-1}QT.
\end{equation}
From the definition of $T$, it follows that the matrix $T^{-1}JT$ has the same block-diagonal structure as $J$ and $T$, and the $j$th diagonal block is the matrix
\begin{equation}\label{eqn:T-block}
	\begin{pmatrix}
	  \lambda_j & \delta^{1/\beta_j} & 0 & \cdots & 0 \\
	   0 & \lambda_j & \delta^{1/\beta_j} & \cdots & 0 \\
	   \vdots & \vdots & \ddots & \ddots & \vdots \\	   
	   0 & 0 & \cdots & \lambda_j & \delta^{1/\beta_j} \\
	   0 & 0 & \cdots & 0 & \lambda_j
\end{pmatrix}.
\end{equation}
It also follows that the $(i,k)$-element of the matrix $\delta T^{-1}QT$ is upper-bounded by $|Q_{ik}|\delta^{1/\beta_j}$, where $j$ is the index for the Jordan block that intersects with the $k$th column of the matrix $J$.
Applying the Gershgorin Theorem~\cite{Horn1985} to the right-hand side of Eq.~\eqref{eqn:transformed}, we see that each eigenvalue of $\hat{M}(\delta)$ must be contained in the disk centered at $\lambda_j$ with radius $C\delta^{1/\beta_j}$ for some $j = 1,\ldots,p$, where $C := 1 + \max_k \sum_i |Q_{ik}|$. [The first term in the expression for $C$ comes from the off-diagonal elements in Eq.~\eqref{eqn:T-block}.]

Now the algebraic and geometric multiplicity of the given eigenvalue $\lambda$ of $M$ can be expressed as $\alpha = \sum_{j} \beta_{j}$ and $\beta = \max_{j} \beta_{j}$, respectively, where the sum and the maximum are both taken over all $j$ for which $\lambda_{j} = \lambda$.
Choose $\hat\lambda(\delta)$ to be any of the $\alpha$ eigenvalues of $\hat{M}(\delta)$ that converge to $\lambda$ as $\delta \to 0$.
Also choose a fixed $\delta$ value sufficiently small to ensure that any two disks with different centers among those mentioned above in connection with the Gershgorin Theorem are disjoint (which can be achieved if $\max_j C\delta^{1/\beta_j}$ is less than half the minimum distance between distinct eigenvalues of $M$).
With this choice, the disk centered at $\lambda$ with radius $C\delta^{1/\beta}$ is  disjoint from all the others and must contain $\hat\lambda(\delta)$; otherwise $\hat\lambda(\delta)$ would have to jump discontinuously from another disk as $\delta \to 0$ since it must remain in at least one of these disks, and this would violate the continuity of $\hat\lambda(\delta)$ with respect to $\delta$.
Having $\hat\lambda(\delta)$ in the disk centered at $\lambda$ with radius $C\delta^{1/\beta}$ immediately gives the 
inequality~\eqref{eqn:scaling_bound}.%

\section{Non-sensitivity under weighted constrained perturbations\label{sec:non-sensitivity-constrained}}

We can
show that all eigenvalues are non-sensitive under a certain class of weighted perturbations even when $\beta>1$.
If the matrix $P$ for the Jordan decomposition of $M$ in Eq.~\eqref{eqn:sim-trans} transforms the perturbation matrix $\Delta M$ into an upper triangular matrix, then the matrix $T^{-1}QT$ in Eq.~\eqref{eqn:transformed} is also upper triangular. In this case, we have the stronger result that the perturbed eigenvalues are given precisely by $\hat{\lambda}(\delta) = \lambda + \delta Q_{ii}$, where $i$ is the index for any column of $J$ that intersects with a Jordan block associated with the eigenvalue $\lambda$. 
The change of each eigenvalue is thus proportional to $\delta$, i.e., the scaling exponent is one, independently of $\beta$ [which is consistent with the general 
result in Eq.~\eqref{eqn:scaling_bound}
since $\beta \ge 1$].
The result for non-generic perturbations in Sec.~\ref{sec:typical-behavior}
follows 
from this if
$M$ is replaced by the Laplacian matrix $L$ and $\Delta M$ by $\Delta L$.
In particular, the result applies to the case of a directed tree with each link having equal weight and $\Delta L$ representing a perturbation of the weights of the existing links.

\section{Scaling for approximately degenerate networks\label{sec:approximate-optimal}}

Here we show that the scaling in Eqs.~\eqref{eqn:scaling_bound} and \eqref{eqn:generic_scaling} is observed even when the eigenvector is only approximately degenerate.
More precisely, we show that, when the matrix is close to one with exact degeneracy, the scaling remains valid over a range of $\delta$ much larger than the distance between the two matrices.

Suppose that a matrix $M_0$ has an eigenvalue $\lambda(M_0)$ with exact geometric degeneracy $\beta$.
We consider a perturbation of $M_0$ in the form $M_1=M_0+\eps\Delta M_1,$ where $\Delta M_1$ is a fixed matrix satisfying $||\Delta M_1||=1$.
Thus, the distance between $M_0$ and $M_1$ is $\eps$, and for small $\eps$ (and a generic choice of $\Delta M_1$) the matrix $M_1$ is approximately degenerate.
We now apply a perturbation of size $\delta$ to $M_1$ in the form $M_2=M_1+\delta\Delta M_2$, where $\Delta M_2$ is another fixed matrix satisfying $||\Delta M_2||=1$.
Denoting $\eta:=\eps/\delta$, we can write $M_2$ as a perturbation of $M_0$ rather than $M_1$, namely, $M_2=M_0+\delta(\eta\Delta M_1+\Delta M_2)$.

When taking the limit $\delta\to 0$ with $\eta$ fixed, matrices $\Delta M_1$ and $\eta\Delta M_1+\Delta M_2$ are both fixed, so we can apply the result in 
Eq.~\eqref{eqn:scaling_bound}.
We thus have
\begin{equation}\label{eqn:limsup-cond}
\begin{split}
\limsup_{\delta\rightarrow{0}}\frac{|\lambda(M_1)-\lambda(M_0)|}{(\eta\delta)^{1/\beta}} &\leq C_1,\\
\limsup_{\delta\rightarrow{0}}\frac{|\lambda(M_2)-\lambda(M_0)|}{\delta^{1/\beta}} &\leq C_2,
\end{split}
\end{equation}
for some constants $C_1,C_2\ge 0$, where $\lambda(M_1)$ and $\lambda(M_2)$ denote eigenvalues of $M_1$ and $M_2$, respectively, that approach $\lambda(M_0)$ as $\delta\to 0$.
This means that for an arbitrary $\xi>0$, we can find $\delta_1>0$ and $\delta_2>0$ (which can depend on $\eta$) such that
\begin{equation}\label{eqn:u-bound}
\begin{split}
\frac{|\lambda(M_1)-\lambda(M_0)|}{(\eta\delta)^{1/\beta}} &< C_1+\frac{\xi}{2\eta^{1/\beta}},\quad\text{if $\delta<\delta_1$},\\
\frac{|\lambda(M_2)-\lambda(M_0)|}{\delta^{1/\beta}} &< C_2+\frac{\xi}{2},\quad\text{if $\delta<\delta_2$}.
\end{split}
\end{equation}
Then,
\begin{equation}\label{eqn:u-bound-2}
\begin{split}
\MoveEqLeft\frac{|\lambda(M_2)-\lambda(M_1)|}{\delta^{1/\beta}}\\
&=\frac{|[\lambda(M_2)-\lambda(M_0)]-[\lambda(M_1)-\lambda(M_0)]|}{\delta^{1/\beta}}\\
&\le \frac{|\lambda(M_2)-\lambda(M_0)|}{\delta^{1/\beta}}+\frac{|\lambda(M_1)-\lambda(M_0)|}{(\eta\delta)^{1/\beta}}\cdot\eta^{1/\beta}\\
&<\left(C_2+\frac{\xi}{2}\right)+\left(C_1+\frac{\xi}{2\eta^{1/\beta}}\right)\cdot\eta^{1/\beta}\\
&=C_2+C_1\eta^{1/\beta}+\xi,
\end{split}
\end{equation}
if $\delta<\min(\delta_1,\delta_2)$.
Since $\xi$ can be made arbitrarily small by making $\delta$ sufficiently small, we have
\begin{equation}
\limsup_{\delta\rightarrow{0}}\frac{|\lambda(M_2)-\lambda(M_1)|}{\delta^{1/\beta}} \leq C_2+C_1\eta^{1/\beta}.
\end{equation}
Thus, Eq.~\eqref{eqn:scaling_bound} and the corresponding bound on the scaling exponent, $\gamma \ge 1/\beta$, remain valid for any fixed $\eta$ (i.e., with $\eps\to 0$ as $\delta\to 0$ while holding $\eta=\eps/\delta$ constant).
For finite $\eps$ and $\delta$, this result suggests that we should observe the scaling $|\lambda(M_2)-\lambda(M_1)|\sim \delta^\gamma$ with $\gamma \ge 1/\beta$ when $\eps \ll \delta \ll 1$.

Now consider the stronger scaling property in Eq.~\eqref{eqn:generic_scaling}, which can be formalized for $M_1$ and $M_2$ as
\begin{equation}\label{eqn:stronger-scaling-cond}
\begin{split}
\lim_{\delta\rightarrow{0}}\frac{|\lambda(M_1)-\lambda(M_0)|}{(\eta\delta)^{1/\beta}} = C_1,\\
\lim_{\delta\rightarrow{0}}\frac{|\lambda(M_2)-\lambda(M_0)|}{\delta^{1/\beta}} = C_2.
\end{split}
\end{equation}
Replacing Eq.~\eqref{eqn:limsup-cond} with Eq.~\eqref{eqn:stronger-scaling-cond} and using the resulting lower bounds analogous to those in Eq.~\eqref{eqn:u-bound}, we obtain a lower bound analogous to that in Eq.~\eqref{eqn:u-bound-2}. Combining this with Eq.~\eqref{eqn:u-bound-2}, we obtain
\begin{equation}\label{eqn:u-bound-3}
\left|\frac{|\lambda(M_2)-\lambda(M_1)|}{\delta^{1/\beta}} - C_2\right| < C_1\eta^{1/\beta}+\xi.
\end{equation}
Since $\xi$ can be made arbitrarily small by making $\delta$ sufficiently small, we see that 
$$\limsup_{\delta\rightarrow{0}}\Bigl|\frac{|\lambda(M_2)-\lambda(M_1)|}{\delta^{1/\beta}} - C_2\Bigr| \le C_1\eta^{1/\beta}.$$
This implies the scaling in Eq.~\eqref{eqn:generic_scaling}, or more precisely, $|\lambda(M_2)-\lambda(M_1)| = C\delta^{1/\beta}$ with a prefactor $C$ that can vary with $\delta$ but is bounded between $C_2 \pm C_1\eta^{1/\beta}$ as $\delta\to 0$.
The ratio $\eta$ of perturbation sizes thus determines the range of variation of this scaling prefactor.
In the limit of both $\eta\to 0$ and $\delta\to 0$, Eq.~\eqref{eqn:u-bound-3} implies $\lim_{\eta,\delta\rightarrow{0}}\frac{|\lambda(M_2)-\lambda(M_1)|}{\delta^{1/\beta}}=C_2$.
Therefore, we have the scaling $|\lambda(M_2)-\lambda(M_1)| \approx C_2 \delta^{1/\beta}$ when $\eps \ll \delta \ll 1$.

Altogether, we have shown that the scaling properties in Eqs.~\eqref{eqn:scaling_bound} and \eqref{eqn:generic_scaling} are observed for the eigenvalues of $M_2$ when the size $\delta$ of the perturbation applied to $M_1$ is much larger compared to the distance $\eps$ between $M_1$ and the exactly degenerate matrix $M_0$.


\end{bibunit}

\clearpage
\begin{bibunit}[plain]

\setcounter{page}{1}
\onecolumngrid

\begin{center}
\textbf{\Large Supplemental Material}\\[3mm]
\textit{Sensitive Dependence of Optimal Network Dynamics on Network Structure}\\[2pt] 
Takashi Nishikawa, Jie Sun, and Adilson E. Motter
\end{center}

\counterwithout{equation}{section} 
\setcounter{equation}{0}
\renewcommand{\theequation}{S\arabic{equation}}

\setcounter{section}{0}
\renewcommand{\thesection}{S\arabic{section}}

\renewcommand*{\citenumfont}[1]{S#1}
\renewcommand*{\bibnumfmt}[1]{[S#1]}

\bigskip\noindent
In the following sections we provide descriptions of several 
examples of
systems and processes 
to which our analysis applies
(Sec.~\ref{si:sec:example-systems}),
proofs of key properties of 
MCC networks (Sec.~\ref{si:sec:prop}), 
as well as 
supplemental figures (Figs.~\ref{fig:comparisonA}--\ref{fig:n-dep-fixed-m}).

\medskip
\section{Example systems and processes}
\label{si:sec:example-systems}

\subsection{Power grids and networks of non-identical phase oscillators}\label{sec:pg}
Given the initial state of a power grid (defined by the voltage phase and magnitude as well as 
active
and reactive power generation/consumption at each node) 
obtained from the standard power flow calculation, the short-term dynamics of generators 
are
governed by the so-called swing equation~\cite{Anderson:1986fk}.
This equation can be expressed as
\begin{equation}\label{eqn:pg}
\frac{2H_i}{\omega_R}\ddot{\theta}_i = -D_i\dot{\theta}_i
+ P_{m,i} - E_i^2 G_{ii} - \sum_{\substack{j=1\\ j \neq i}}^n E_i E_j[B_{ij}\sin(\theta_i - \theta_j) 
+ G_{ij}\cos(\theta_i - \theta_j)],
\end{equation}
where, for each generator $i$, the angle $\theta_i$ is the voltage phase relative to a reference phase rotating at frequency $\omega_R$, $H_i$ is the inertia constant, $D_i$ is the effective damping coefficient (representing various damping and damping-like effects, including 
the actions of power system stabilizers), 
and $P_{m,i}$ is the constant mechanical power provided 
to 
the generator.
Each generator $i$ is modeled as a voltage source of constant magnitude $E_i$ connected to the rest of the network through a transient reactance $x'_{d,i}$ and a terminal node, where $E_i$ is determined from the initial state of the system obtained through power flow calculation. 
The interactions between generators are represented by an effective admittance matrix, whose real and imaginary components are denoted by $G_{ij}$ and $B_{ij}$, respectively.
This matrix is obtained by reducing the matrix representing the physical network of transmission lines and transformers, which accounts for the network topology encoded by the adjacency matrix $A$, the node-to-node impedances, and the nodes' power consumption modeled as equivalent constant impedances to the ground~\cite{Motter:2013fk}.
While Eq.~\eqref{eqn:pg} does not directly fit into the framework of Eq.~\eqref{eqn:coupled-syst} in the main text, it can be transformed into the same form by rewriting it in terms of the deviations from a frequency-synchronized solution, as we will show below.

Equation~\eqref{eqn:pg}, which describes the dynamics of power generators, belongs to the class of coupled second-order phase oscillators, as it can be expressed as
\begin{equation}\label{eqn:second-order-kuramoto}
\ddot{\theta}_i + \beta_i\dot{\theta}_i = \Omega_i
  + \sum_{j=1}^n \hat{H}_{ij}(\theta_i - \theta_j),
\end{equation}
where
\begin{align}
&\beta_i = \frac{D_i\omega_R}{2H_i},\\
&\Omega_i = \frac{P_{m,i}\omega_R}{2H_i},\\
&\hat{H}_{ij}(\theta) = - \frac{\omega_R E_i E_j}{2H_i}[ B_{ij}\sin(\theta) + G_{ij}\cos(\theta)].
\end{align}
Besides modeling power-grid dynamics, the general class of systems described by Eq.~\eqref{eqn:second-order-kuramoto} can be used to model disordered Josephson junction arrays~\cite{PhysRevE.71.016215}.
Note that this class of systems does not directly fit into the framework of Eq.~\eqref{eqn:coupled-syst} because the oscillators are not identical and the coupling function $\hat{H}_{ij}$ generally does not satisfy the requirement $\mathbf{H}_{ij}(\mathbf{u},\mathbf{u}) = 0$.
This is indeed the case with the power-grid equation~\eqref{eqn:pg} because of the cosine term.
The difficulty, however, can be overcome 
in any system of the form in Eq.~\eqref{eqn:second-order-kuramoto} by changing the frame of reference.  Specifically, we rewrite
the equation in terms of the deviations from a frequency-synchronized solution $\theta_i = \theta^*_i + \Omega_s t$, satisfying 
\begin{equation}
\beta_i \Omega_s = \Omega_i + \sum_{j=1}^n \hat{H}_{ij}(\theta^*_i - \theta^*_j)
\end{equation}
for each $i$.
Writing $\theta'_i = \theta_i - (\theta^*_i + \Omega_s t)$, we obtain 
\begin{equation}
\ddot{\theta}'_i + \beta_i\dot{\theta}'_i = \sum_{j=1}^n \widetilde{H}_{ij}(\theta'_i - \theta'_j),
\end{equation}
 where $\widetilde{H}_{ij}(\theta) 
= \hat{H}_{ij}(\theta^*_i - \theta^*_j + \theta)
 - \hat{H}_{ij}(\theta^*_i - \theta^*_j)$.
If we assume that $\beta_i = \beta$ is independent of $i$ (which can hold even if the individual generator parameters vary widely across the network), then this equation fits in the framework of Eq.~\eqref{eqn:coupled-syst} with 
\begin{gather}
\mathbf{x}_i = \begin{pmatrix}\theta'_i\\ \dot{\theta}'_i \end{pmatrix}, \\
\mathbf{F}( \mathbf{x},\mathbf{y}_1,\ldots,\mathbf{y}_n) = \begin{pmatrix} \dot{\theta}\\ -\beta\dot{\theta} \end{pmatrix} + \sum_{j=1}^n \mathbf{y}_j, \quad \mathbf{x} = \begin{pmatrix} \theta\\ \dot{\theta} \end{pmatrix}, \\
\mathbf{H}_{ij}\bigl( \mathbf{u}, \mathbf{v} \bigr) = \begin{pmatrix} 0\\ \widetilde{H}_{ij}(u_1 - v_1) \end{pmatrix},
\quad \mathbf{u} = \begin{pmatrix} u_1\\ u_2 \end{pmatrix},
\quad \mathbf{v} = \begin{pmatrix} v_1\\ v_2 \end{pmatrix},
\end{gather} 
and this $\mathbf{H}_{ij}$ satisfies 
the assumptions (A-1) and (A-2) in the main text.
In particular, we have 
$D_{\mathbf{v}}\mathbf{H}_{ij} = \widetilde{A}_{ij}\mathbf{G}$, 
where $\widetilde{A}_{ij} = - \hat{H}'_{ij}(\theta^*_i - \theta^*_j)$ and $\mathbf{G} = \bigl({0 \atop 1} {0 \atop 0}\bigr)$.
For the network-homogeneous state $\theta'_i = \theta^*$, $\forall i$, Eq.~\eqref{eqn:mse} in Appendix~\ref{sec:stab_func} thus becomes
\begin{equation}
\dot{\boldsymbol{\eta}}
= [D_{\mathbf{x}}\mathbf{F} -\alpha \mathbf{G}]\boldsymbol{\eta} = \begin{pmatrix} 0 & 1\\ -\alpha & -\beta \end{pmatrix}\boldsymbol{\eta},
\quad D_{\mathbf{x}}\mathbf{F} = 
\begin{pmatrix} 0 & 1\\ 0 & -\beta \end{pmatrix}.
\end{equation}
The stability function can then be explicitly calculated as 
\begin{equation}
\Lambda(\alpha) = \text{Re}\left(\frac{-\beta + \sqrt{\beta^2 - 4\alpha}}{2}\right),
\end{equation}
 which shows that the  stability function derived in Ref.~\cite{Motter:2013fk} remains the same for the general class of second-order phase oscillator 
networks in Eq.~\eqref{eqn:second-order-kuramoto}
and that it is independent of the choice of the coupling function $\hat{H}_{ij}$.
The stability function is strictly decreasing as a function of $\text{Re}(\alpha)$ for $\text{Re}(\alpha) \le \beta^2/4$ and strictly increasing as a function of $|\text{Im}(\alpha)|$.
The convergence rate to the frequency-synchronized solution $\theta_i = \theta^*_i + \Omega_s t$ can then be determined to be $\Lambda_{\max} = \max_{j\ge 2} \Lambda(\widetilde{\lambda}_j)$, where $\widetilde{\lambda}_j$ are the eigenvalues of 
the (weighted) Laplacian 
matrix associated with the network of effective interactions represented by 
the (weighted) effective adjacency 
matrix $\widetilde{A} = (\widetilde{A}_{ij})$.

In the case of power-grid dynamics~\eqref{eqn:pg}, we have 
\begin{equation}\label{eqn:Aij-power-grid}
\widetilde{A}_{ij} = \frac{\omega_R E_i E_j}{2H_i}[ B_{ij}\cos(\theta^*_i - \theta^*_j) - G_{ij}\sin(\theta^*_i - \theta^*_j)].
\end{equation}
For the results shown in Fig.~\ref{fig:example-systems}(a) of the main text on the dependence of $\Lambda_{\max}$ on the network topology, we assume for simplicity that all generators, transformers, and transmission lines have identical parameters.
The parameters are taken from the components in the 9-bus example system discussed in Ref.~\cite{Anderson:1986fk}, and we stressed the system (power demand and consumption uniformly increased by a factor of $8$), while using a smaller value of generators' transient reactances ($x'_{d,i} = 0.02$ per unit) to ensure stability.  
To construct a power-grid network for a given adjacency matrix $A$, we first connect the terminal node of each generator to a unique load node through a transformer.
These $n$ load buses are then connected to each other by transmission lines according to the topology given by $A$.
Note that the eigenvalues to be inserted into $\Lambda(\alpha)$ to determine $\Lambda_{\max}$ are the Laplacian eigenvalues $\widetilde{\lambda}_j$ corresponding to $\widetilde{A}$, which are not generally equal to $\lambda_j$.
However, the variations of $E_i$ and $\theta^*_i$ are both small for a typical power grid, 
making $\widetilde{A}$ approximately equal to
a constant multiple of 
$B = (B_{ij})$ [see Eq.~\eqref{eqn:Aij-power-grid}].
Since $B$ is tightly related to $A$ through the network reduction process mentioned 
after Eq.~\eqref{eqn:pg},
we expect $\lambda_j$ to be strongly correlated with $\widetilde{\lambda}_j$, and therefore to the convergence rate $\Lambda_{\max}$.
Note that if 
we assume 
that $\beta_i$ is independent of $i$ and the transmission lines are lossless 
(i.e., $G_{ij}=0$ for $i\neq j$), 
then Eq.~\eqref{eqn:second-order-kuramoto} 
would take the form of a network of coupled second-order Kuramoto-type phase oscillators~\cite{PhysRevLett.110.218701}, which is the form used to model power grids in Ref.~\cite{Filatrella2008}. If we take the limit of strong damping (large $D_i$) instead of assuming $\beta_i = \beta$, the node dynamics would reduce to that of a network of first-order Kuramoto-type oscillators~\cite{fd-fb:09z}.

We also note that the technique used above to put Eq.~\eqref{eqn:second-order-kuramoto} in the form of Eq.~\eqref{eqn:coupled-syst} can also be used to treat the phase-reduced model for networks of weakly coupled limit-cycle 
oscillators~\cite{hoppensteadt1997weakly}
\begin{equation}\label{eqn:first-order-kuramoto}
\dot{\theta}_i = \Omega_i + \sum_{j=1}^n \hat{H}_{ij}(\theta_i - \theta_j),
\end{equation}
which can be regarded as a first-order version of Eq.~\eqref{eqn:second-order-kuramoto}.
In this case, Eq.~\eqref{eqn:mse} becomes $\dot{\boldsymbol{\eta}} = -\alpha\boldsymbol{\eta}$, and the stability function is 
$\Lambda(\alpha) = - \text{Re}(\alpha)$ [which is strictly decreasing in $\text{Re}(\alpha)$], 
again independent of the choice of the coupling function $\hat{H}_{ij}$.
The convergence rate to the frequency-synchronized solution can be computed as 
$\Lambda_{\max} = \max_{j\ge 2} \Lambda(\widetilde{\lambda}_j)$, where $\widetilde{\lambda}_j$ are the eigenvalues of the Laplacian matrix associated with the effective adjacency matrix defined by $\widetilde{A}_{ij} = - \hat{H}'_{ij}(\theta^*_i - \theta^*_j)$.
In the special case of the Kuramoto model with arbitrary network structure $A_{ij}$ and pairwise frustration parameters $\delta_{ij}$, the coupling function is $\hat{H}_{ij}(\theta) = A_{ij}\sin(\theta + \delta_{ij})$, and the effective interaction matrix is given by $\widetilde{A}_{ij} = A_{ij} \cos(\theta^*_i - \theta^*_j + \delta_{ij})$.

\subsection{Networks of identical oscillators}
The class of coupled oscillator networks considered in Ref.~\cite{Pecora:1998zp} can be described by Eq.~\eqref{eqn:coupled-syst} with 
\begin{gather}
\mathbf{F}(\mathbf{x}, \mathbf{y}_1, \ldots, \mathbf{y}_n) = \mathbf{F}(\mathbf{x}) + \sum_{j=1}^n \mathbf{y}_j, \\
\mathbf{H}_{ij}(\mathbf{u},\mathbf{v}) = \eps A_{ij} [\mathbf{H}(\mathbf{v}) - 
\mathbf{H}(\mathbf{u})], \label{eqn:Hij-coupled-osc}
\end{gather}
where the function $\mathbf{F}(\mathbf{x})$ describes the dynamics of a single oscillator in isolation, the function $\mathbf{H}$ represents the signal that a 
node
sends to other 
nodes,
and $\eps$ is the global coupling strength.
We note that $\mathbf{H}_{ij}$ defined in Eq.~\eqref{eqn:Hij-coupled-osc} satisfies the assumptions (A-1) and (A-2) in the main text for any (differentiable) $\mathbf{H}$.
In particular, we have $D_{\mathbf{v}}\mathbf{H}_{ij}(\mathbf{x}^*,\mathbf{x}^*) = A_{ij} \cdot \mathbf{G}(t)$, where $\mathbf{G}(t) = \eps D\mathbf{H}(\mathbf{x}^*(t))$.
In this case, the network-homogeneous state 
represents completely
synchronized periodic or chaotic motion of the oscillators.
The stability function $\Lambda(\alpha) = \Lambda_{\eps}(\alpha)$ derived in Appendix~\ref{sec:stab_func} is determined from Eq.~\eqref{eqn:mse}, which can be written as 
$\dot{\boldsymbol{\eta}}
= [D\mathbf{F} -\alpha\eps D\mathbf{H}]\boldsymbol{\eta}$
in this case~\cite{Pecora:1998zp}.
For many choices of $\mathbf{F}$ and $\mathbf{H}$, the stability function has a range of $\text{Re}(\alpha)$ for which it is a strictly decreasing function of $\text{Re}(\alpha)$.

The experimental system of coupled optoelectronic oscillators studied in Ref.~\cite{PhysRevLett.107.034102} does not directly fit into the class of systems described by Eq.~\eqref{eqn:coupled-syst} because non-negligible time delay is 
present
in the dynamics of individual oscillators, as well as in the coupling between them.
However, the convergence rate toward synchronization in this system 
is determined 
by the Laplacian eigenvalues through a stability function, which is slightly different from $\Lambda(\alpha)$ defined in Appendix~\ref{sec:stab_func} but can be derived in a similar fashion.
Indeed, the experimentally observed (finite-time) convergence rate toward synchronous state can be estimated as 
\begin{equation}
\mu := \frac{\sum_{j=2}^n \mu_j e^{-\mu_j T}}{\sum_{j=2}^n e^{-\mu_j T}},
\end{equation}
where $\mu_j = \widetilde\Lambda(\eps\lambda_j/d)$ is the convergence rate for the $j$th eigenmode of perturbations, $\widetilde\Lambda$ is the stability function, $\lambda_j$ is the $j$th Laplacian eigenvalue of the network with adjacency matrix $A_{ij}$, the constant $T = 2.0$ ms is the time scale for the experimental system to converge to the synchronous state, $\eps$ is the global coupling strength, and $d$ is the normalization factor defined as the average coupling per node, $\sum_i\sum_{j\neq i} A_{ij}/n$.
We use the stability function $\widetilde\Lambda$ measured experimentally in Ref.~\cite{PhysRevLett.107.034102} and choose the value of $\eps$ for which an optimal network with $\lambda_2 = \cdots = \lambda_n = \bar{\lambda}$ achieves the maximum possible convergence rate.
In Fig.~\ref{fig:generic_perturbation_phys}, which shows our results for directed networks, we take only the real part of $\lambda_j$, since $\widetilde\Lambda$ is available from the experiments only for real values of its argument. 
This approximation is justified by the fact that the dependence of $\widetilde\Lambda$ on the imaginary part is quadratic, and that the deviation of $\lambda_j$ from $\bar{\lambda}$, and thus the imaginary part of $\lambda_j$, is small for small perturbation size $\delta$.
In this system, the 
elements
of the perturbation matrix $\Delta A$ represents imperfection in prescribing the coupling strengths between dynamical units in experiments.

\subsection{Networks of activator-inhibitor systems}
We consider a network of coupled activator-inhibitor systems studied in Ref.~\cite{Nakao:2010fk}, which can be described by Eq.~\eqref{eqn:coupled-syst} with $\mathbf{x}_i = (z_i,w_i)^T$, where $z_i$ and $w_i$ denote the activator and inhibitor concentrations at node $i$, respectively, and 
\begin{gather}
\mathbf{F}(\mathbf{x},\mathbf{y}_1,\ldots,\mathbf{y}_n) = 
\begin{pmatrix} f(z,w)\\ g(z,w) \end{pmatrix} + \sum_{j=1}^n \mathbf{y}_j,\\
f(z,w) = [(a+bz-z^2)/c-w]z,\quad
g(z,w) = [z-(1+ew)]w,\\
\mathbf{H}_{ij}(\mathbf{u},\mathbf{v}) = \eps A_{ij} [\mathbf{H}(\mathbf{v}) - 
\mathbf{H}(\mathbf{u})],\quad
\mathbf{H}(\mathbf{x}) = \begin{pmatrix} z\\ \sigma w \end{pmatrix}. \label{eqn:Hij-activator-inhibitor}
\end{gather}
Here $\eps$ and $\sigma\eps$ are the diffusivity constants for the activator and inhibitor, respectively. 
The symmetry of the diffusion of these species implies that the network given by $A_{ij}$ must be undirected (and unweighted).
For the result shown in Fig.~\ref{fig:example-systems}(b), we used the node parameter values from Ref.~\cite{Nakao:2010fk} ($a=35$, $b=16$, $c=9$, and $e=2/5$) and set $\sigma=16$.
For any value of $\eps$ (and $\sigma$), this network has a state of uniform concentration given by $\mathbf{x}_i = \mathbf{x}^* =(5,10)^T$, $\forall i$.
The stability function is determined from Eq.~\eqref{eqn:mse} in Appendix~\ref{sec:stab_func}, which in this case becomes 
\begin{equation} \label{eqn:S21}
\dot{\boldsymbol{\eta}}
= [D_{\mathbf{x}}\mathbf{F} - \alpha\eps D\mathbf{H}]\boldsymbol{\eta},\quad
D_{\mathbf{x}}\mathbf{F} = \begin{pmatrix} f_z & f_w\\ g_z & g_w \end{pmatrix},\quad
D\mathbf{H} = \begin{pmatrix} 1 & 0\\ 0 & \sigma \end{pmatrix},
\end{equation}
where the partial derivatives of $f$ and $g$ evaluated at $\mathbf{x}^*$ are denoted by $f_z$, $f_w$, $g_z$, and $g_w$.
We note that $\mathbf{H}_{ij}$ defined in Eq.~\eqref{eqn:Hij-activator-inhibitor} has the same form as Eq.~\eqref{eqn:Hij-coupled-osc}.
It thus satisfies the assumptions (A-1) and (A-2) in the main text, and we have $D_{\mathbf{v}}\mathbf{H}_{ij}(\mathbf{x}^*,\mathbf{x}^*) = A_{ij} \mathbf{G}$, where $\mathbf{G} = \eps D\mathbf{H}(\mathbf{x}^*) = \eps\bigl({1 \atop 0} {0 \atop \sigma}\bigr)$.
Since $D_{\mathbf{x}}\mathbf{F} - \alpha\eps D\mathbf{H}$ in Eq.~\eqref{eqn:S21} is a $2 \times 2$ matrix that is independent of time, computing its eigenvalues leads to an explicit formula for the  stability function:
\begin{equation}
\Lambda(\alpha) = \Lambda_{\eps}(\alpha) \\
= \frac{1}{2}\left\{ f_z + g_w - (1+\sigma)\alpha\eps + \sqrt{4 f_w g_z + (f_z - g_w - (1-\sigma)\alpha\eps)^2} \right\}.
\end{equation}
For the parameters used, we find this function to be strictly decreasing for $0 \le \alpha \le (f_z  - g_w - 2\sqrt{-f_z g_w})/[(1-\sigma)\eps]$.
The function has two values of $\alpha$ for which $\Lambda(\alpha) = 0$, the larger of which can be written as $\alpha = \alpha_c/\eps$, where $\alpha_c$ is a constant (which equals $11/6$ for the parameter values we used).
Thus, as $\eps$ decreases from a sufficiently large value, the critical value of the diffusivity constant $\eps=\eps_c$ at the onset of Turing instability is determined by the condition $\lambda_2 = \alpha_c/\eps_c$, thus giving $\eps_c = \alpha_c/\lambda_2$, which is a strictly decreasing function of $\lambda_2$.

\subsection{Consensus protocol}
The continuous-time linear consensus 
protocol 
is described by Eq.~\eqref{eqn:coupled-syst} with 
\begin{equation}
\mathbf{F}(\mathbf{x},\mathbf{y}_1,\ldots,\mathbf{y}_n) = \sum_{j=1}^n \mathbf{y}_j,\quad
\mathbf{H}_{ij}(\mathbf{u},\mathbf{v}) = A_{ij}(\mathbf{v} - \mathbf{u}).
\end{equation}
This $\mathbf{H}_{ij}$ satisfies the assumptions (A-1) and (A-2) in the main text, and we have $D_{\mathbf{v}}\mathbf{H}_{ij}(\mathbf{x}^*,\mathbf{x}^*) = A_{ij} \mathbf{G}$, where $\mathbf{G}$ equals the identity matrix in this case.
Equation~\eqref{eqn:mse}
in Appendix~\ref{sec:stab_func} reads $\dot{\boldsymbol{\eta}}= -\alpha \boldsymbol{\eta}$, and the stability function is given by 
$\Lambda(\alpha) = -\text{Re}(\alpha)$ [which is strictly decreasing with respect to $\text{Re}(\alpha)$].
If there is a directed spanning tree embedded in the network, we have $\text{Re}(\lambda_j) > 0$ for all $j\ge 2$, and the system converges from an arbitrary initial condition to a 
network-homogeneous state,
$\mathbf{x}_i = \mathbf{x}^*$, $\forall i$~\cite{4140748} at a rate given by $\Lambda_{\max} = -\min_{j\ge 2} \text{Re}(\lambda_j)=-\text{Re}(\lambda_2)$.
If the network is strongly connected and balanced in the sense that the in- and out-degrees are equal for all nodes (i.e., $\sum_j A_{ij} = \sum_i A_{ij}$), then this state corresponds to the solution of the average consensus problem, $\mathbf{x}^* = \sum_i \mathbf{x}_i(0)/n$~\cite{4140748}.

\subsection{Diffusion over networks}
Assuming the Fick's law for the diffusive process over each link of the network, the flux from node $j$ to node $i$ ($\neq j$) is given by $\eps (x_j - x_i)$, where $\eps$ is the diffusivity constant and $x_i$ is the density of the diffusing species at node $i$.
The dynamics of the system 
are
then governed by Eq.~\eqref{eqn:coupled-syst} with 
\begin{equation}
\mathbf{F}(x,y_1,\ldots,y_n) = \sum_{j=1}^n y_j, \quad
\mathbf{H}_{ij}(u,v) = \eps A_{ij}(v - u).
\end{equation}
This $\mathbf{H}_{ij}$ satisfies the assumptions (A-1) and (A-2) in the main text, and we have $D_{\mathbf{v}}\mathbf{H}_{ij}(\mathbf{x}^*,\mathbf{x}^*) = A_{ij}\mathbf{G}$ with $\mathbf{G} = \eps$.
The symmetry of the diffusive process requires that the network be undirected, and hence that $\lambda_j$ be all real.
For a connected network (which has $\lambda_2 > 0$), the system has a 
network-homogeneous state,
$x_i(t) = x^* = \sum_i x_i(0)/n$, $\forall i,t$, for which Eq.~\eqref{eqn:mse} in Appendix~\ref{sec:stab_func} becomes $\dot{\boldsymbol{\eta}}= -\alpha\eps \boldsymbol{\eta}$, leading to the stability function 
$\Lambda(\alpha) = -\eps\text{Re}(\alpha)$ [which is strictly decreasing in $\text{Re}(\alpha)$].
The rate of exponential convergence to this state is thus given by $\Lambda_{\max}=\Lambda(\lambda_2) = -\eps\lambda_2$.

\subsection{Fluid networks}
The equation of motion for a system of vertical pipes partially filled with liquid and connected by horizontal pipes is given by Eq.~\eqref{eqn:coupled-syst}, where $\mathbf{x}_i = (z_i, \dot{z}_i)^T$, variable $z_i$ represents the liquid level of the $i$th vertical pipe, and
\begin{gather}
\mathbf{F}(\mathbf{x},\mathbf{y}_1,\ldots,\mathbf{y}_n) = \begin{pmatrix} \dot{z}\\ -2\nu \dot{z} \end{pmatrix} + \sum_{j=1}^n \mathbf{y}_j,\\
\mathbf{H}_{ij}(\mathbf{u},\mathbf{v}) = A_{ij} [\mathbf{H}(\mathbf{v}) - 
\mathbf{H}(\mathbf{u})], \quad \mathbf{H}(\mathbf{x}) = \begin{pmatrix} 0\\ z \end{pmatrix}, \label{eqn:Hij-fluid}
\end{gather}
 and $2\nu$ is the coefficient for the damping due to friction in the pipes~\cite{maas1987transportation}.
The difference in the pressure force in two connected vertical pipes with different liquid height drives the system to an equilibrium state of equal $z_i$.
Since the same pressure force acts on the liquid levels in the two pipes with opposite signs, the interaction matrix $A_{ij}$ is symmetric, implying that we have an undirected network.
Equation~\eqref{eqn:mse} in Appendix~\ref{sec:stab_func} becomes 
\begin{equation}\label{eqn:S27}
\dot{\boldsymbol{\eta}}
= [D_{\mathbf{x}}\mathbf{F} -\alpha D\mathbf{H}]\boldsymbol{\eta},\quad
D_{\mathbf{x}}\mathbf{F} = \begin{pmatrix}
0 & 1\\
0 & -2\nu
\end{pmatrix},\quad
D\mathbf{H} = 
\begin{pmatrix}
0 & 0\\
1 & 0
\end{pmatrix}.
\end{equation}
We note that $\mathbf{H}_{ij}$ defined in Eq.~\eqref{eqn:Hij-fluid} satisfies the assumptions (A-1) and (A-2) in the main text, and we have $D_{\mathbf{v}}\mathbf{H}_{ij}(\mathbf{x}^*,\mathbf{x}^*) = A_{ij} \cdot \mathbf{G}$ with $\mathbf{G} = D\mathbf{H}$.
From Eq.~\eqref{eqn:S27},
the stability function is derived to be $\Lambda(\alpha) = -\nu+\sqrt{\nu^2 - \alpha}$ if $\alpha \le \nu^2$ and $\Lambda(\alpha) = -\nu$ if $\alpha > \nu^2$, which is a strictly decreasing function of $\alpha$ for $\alpha \le \nu^2$.

\section{\texorpdfstring{Laplacian eigenvalue {\normalsize$\boldsymbol{\lambda_2}$} of MCC networks}{Laplacian eigenvalue lambda2 of MCC networks}}\label{si:sec:prop}

Here we first prove some basic properties of the function $M(n,k)$ in subsection A.  
We then derive the relation $\lambda_2 = n - k_{n,m_c}$ for all MCC networks for which $\phi < 1$ (subsection B) and establish the local $\lambda_2$-optimality for 
all MCC networks
(subsection C).
We also establish the global optimality for certain specific cases in subsection C.  Finally, we derive the formula in Eq.~\eqref{eqn:lambda2} of the main text in subsection D.

\subsection{\texorpdfstring{Basic properties of $\boldsymbol{M(n,k)}$}{Basic properties of M(n,k)}}
Recall the definition of the class of MCC networks for given $n$ and $m$ as the set of all $n$-node connected networks with $m$ links for which the largest connected component of the complement is of size $k_{n,m_c}$, where $m_c = n(n-1)/2 - m$ is the number of links in the complement.
Here, as in the main text, $k_{n,m}$ is the smallest integer $k$ for which $m \le M(n,k)$ for given $n$ and $m$, where $M(n,k)$ is the maximum number of links allowed for any $n$-node network whose connected components have size $\le k$.
The following proposition establishes basic properties of the function $M(n,k)$.
\begin{proposition}\label{prop:Mnk}
Given $n$ and $k$ with $1 \le k \le n$, the only  
network with $n$ nodes and $m = M(n,k)$ links whose connected components have size $\le k$
is a network consisting of $\ell$ isolated, fully connected clusters of size $k$ and an additional isolated, fully connected cluster of size $n_r := n-k\ell$, where $\ell=\lfloor n/k\rfloor$. It follows that 
\begin{equation}\label{eq:Mnk}
	M(n,k) = \ell\cdot\frac{k(k-1)}{2}+\frac{n_r(n_r-1)}{2}.
\end{equation}
Furthermore,
\begin{equation}\label{eq:Mnk-property}
	0 = M(n,1)<M(n,2)<\dots<M(n,n) = \frac{n(n-1)}{2}.
\end{equation}
\end{proposition}
\begin{proof}
Given a network of $n$ nodes and $m$ links, suppose that its largest component size is smaller than or equal to $k$.
By definition, we have $m\leq M(n,k)$. We shall show that 
if $m=M(n,k)$, 
the number of $k$-node, fully connected clusters in the network must equal $\lfloor n/k\rfloor$, which is the maximum number of isolated $k$-node clusters allowed in such a network, and the remaining nodes must form a single, fully connect cluster.

Suppose that $m=M(n,k)$.
We first show that the individual components of the network must all be fully connected.
If they are not, then we can add links to a component that is not fully connected and increase the total number of links to be larger than $M(n,k)$ while keeping the component sizes fixed (and hence less than or equal to $k$).
Since this contradicts the definition of $M(n,k)$ that it is the maximum number of links allowed given the number of nodes $n$ and the largest component size $k$, all components in the network must be fully connected.

We now show that $\ell=\lfloor n/k\rfloor$, where $\ell$ denotes the number of fully connected components of size $k$.
If $\ell<\lfloor n/k\rfloor$, then the number of remaining nodes 
$n_r=n-k\ell>k$.
Hence, there
must be at least two (fully connected) components of size $k_1$ and $k_2$ among these $n_r$ nodes which satisfy $k_1\leq k_2<k$.
Thus, we can
``relocate'' a node from the component of size $k_1$ to the component of size $k_2$, while rewiring and adding links to ensure that the two components are both fully connected. 
This ``relocation'' increases the total number of links by $k_2-(k_1-1)\geq 1$
while $k$ remains the largest component size of the 
network, which contradicts the definition of $M(n,k)$. 
Therefore, we must have $\ell=\lfloor n/k\rfloor$.

The $n_r$ remaining nodes must form a single, fully connected cluster, since we can otherwise increase the number of links by adding links to that part of the network, again contradicting the definition of $M(n,k)$.
Equation~\eqref{eq:Mnk} then follows immediately from counting the total number of links in the $\ell$ components of size $k$ and one component of size $n_r$, all of which are fully connected.

To prove Eq.~\eqref{eq:Mnk-property}, suppose
that $k < n$ and apply the same ``relocation'' argument to the network with $m=M(n,k)$ links.
This involves a cluster of size $k$ increasing its size to $k+1$ by incorporating a new node from another cluster (noting that $k < n$ guarantees at least two clusters) and adding links to keep the cluster fully connected.
We thus see that the total number of links must strictly increase.
Since $M(n,k+1)$ is larger than or equal to this number of links by definition, we have $M(n,k)<M(n,k+1)$ for $k=1,\ldots,n-1$.
The equalities in Eq.~\eqref{eq:Mnk-property} follow directly from Eq.~\eqref{eq:Mnk}.
\end{proof}

The class of MCC networks for given $n$ and $m$ contains the MCC network generated by the procedure described in Appendix~\ref{si:sec:construction_MCC}, but it generally contains more.
However, for certain combinations of $n$ and $m$, we can show that the class contains only one MCC network. 
Indeed, if $m_c = M(n,k)$ for some $1 \le k \le n$, then Proposition~\ref{prop:Mnk} applied to the complement implies that there is only one MCC network in the class.
Note that in this case we have $k_{n,m_c} = k$.
If $m_c$ increases by one, $m_c \le M(n,k)$ would no longer be satisfied, which forces $k_{n,m_c}$ to jump from $k$ to $k+1$ and thus causes a jump in $\lambda_2 = n - k_{n,m_c}$, as observed, e.g., for the orange curve in Fig.~\ref{fig:construction}.
In the case where $n_r = 0$ is satisfied in addition to $m_c = M(n,k)$, we can show that the unique MCC network is actually a UCM network, as stated in the following proposition.

\begin{proposition}\label{prop:UCM-unique}
Suppose that $n = k\ell$ and $m = k^2 \ell(\ell - 1)/2$ for positive integers $\ell$ and $k$.
Then, the only MCC network with $n$ nodes and $m$ links is the network whose complement consists of $\ell$ isolated, fully connected clusters of size $k$.
\end{proposition}
\begin{proof}
The assumption $n = k\ell$ is equivalent to $n_r = 0$ and implies $M(n,k) = \ell k(k-1)/2$.
Since $m_c = n(n-1)/2 - m = \ell k(k-1)/2$, we have $m_c = M(n,k)$.
We can thus apply Proposition~\ref{prop:Mnk} to the complement of an MCC network with $n$ nodes and $m$ links to conclude that the complement of any such network must consist of $\ell$ isolated, fully connected clusters of size $k$.
\end{proof}%

\subsection{\texorpdfstring{Proof of {\small$\boldsymbol{\lambda_2 = n - k_{n,m_c}}$}}{Proof of lambda2 = n - k n mc}}

For a given MCC network with $n$ nodes and $m$ links,
assume $m < n(n-1)/2$ (i.e., $\phi < 1$), which guarantees that $k_{n,m_c} \ge 2$.
(We note that if $\phi=1$ the MCC network must be the complete graph, for which $\lambda_2=n$, even though $k_{n,m_c}=1$ in that case.)
Below we show that $\lambda_2 = n - k_{n,m_c}$, or equivalently, $\lambda_n^c = k_{n,m_c}$, using the following propositions.
\medskip
\begin{proposition}\label{prop:union}
The Laplacian spectrum of a network is the union of the Laplacian spectra of its connected components.
\end{proposition}
\begin{proof}
This well-known result follows immediately from reordering the indices to make the Laplacian matrix block-diagonal (with blocks corresponding to connected components) and using the fact that the eigenvalue spectrum of a block diagonal matrix is the union of the spectra of its diagonal blocks.
\end{proof}

\begin{proposition}\label{prop:max-lambda}
The largest Laplacian eigenvalue of a network is at most the size of the largest connected components.
\end{proposition}
\begin{proof}
First note that the largest Laplacian eigenvalue of any network is bounded by the size of that network 
(a proof can be found, e.g., in Sec.~3.9 of Ref.~\cite{brouwer2012spectra}).
Applying this to each connected component of a given network, we see that the largest Laplacian eigenvalue of the component is at most the size of that component.
The conclusion now follows directly from Proposition~\ref{prop:union} by considering the union of the Laplacian spectra of the components.
\end{proof}

\begin{proposition}\label{prop:n}
For any network with $n \ge 2$ nodes and more than $(n-1)(n-2)/2$ links, the largest Laplacian eigenvalue is $n$.
\end{proposition}
\begin{proof}
Suppose that a given network has $n$ nodes and $m$ links, where $m > (n-1)(n-2)/2$.
From the identity $(n-1)(n-2)/2 = n(n-1)/2 - (n-1)$, the number of links in its complement  satisfies $m_c < n-1$.
Since the minimum number of links required for a connected $n$-node network is $n-1$, the complement must be disconnected.
This implies that the second smallest Laplacian eigenvalue $\lambda^c_2$ of the complement is zero.
Therefore, the largest Laplacian eigenvalue of the network is $\lambda_n = n - \lambda^c_2 = n$.
\end{proof}

To prove $\lambda_n^c = k_{n,m_c}$,
we separately show $\lambda_n^c \le k_{n,m_c}$ and $\lambda_n^c \ge k_{n,m_c}$.
First, by the definition of MCC networks, the size of each connected component in its complement is at most $k_{n,m_c}$.
By Proposition~\ref{prop:max-lambda}, we immediately obtain $\lambda_n^c \le k_{n,m_c}$.
Next, to show $\lambda_n^c \ge k_{n,m_c}$, suppose hypothetically that all connected components in the complement of the network have at most $(k_{n,m_c} - 1)(k_{n,m_c} - 2)/2$ links.
Then, the links in each component of size $k_{n,m_c}$ can be rewired to form a connected component of size $k_{n,m_c} - 1$ with an additional node that is isolated.
Doing this, the maximum component size becomes $k_{n,m_c} - 1$, and the total number of links forming these components is at most $M(n, k_{n,m_c}-1)$.
This implies that the original network must have $m_c \le M(n, k_{n,m_c}-1)$, which contradicts with the definition of $k_{n,m_c}$.
Therefore, the initial hypothesis is false, and there must be at least one connected component with more than $(k_{n,m_c} - 1)(k_{n,m_c} - 2)/2$ links.
Since $k_{n,m_c} \ge 2$, it follows from Proposition~\ref{prop:n} that the largest Laplacian eigenvalue of this component is $k_{n,m_c}$. 
By Proposition~\ref{prop:union}, we conclude that $\lambda_n^c \ge k_{n,m_c}$.
Altogether, we have $\lambda_n^c = k_{n,m_c}$, and hence $\lambda_2 = n - k_{n,m_c}$ for all MCC networks 
with
$\phi < 1$.

\subsection{Local and global optimality}

We now prove that the class of MCC networks optimizes $\lambda_2$.
We divide the proof into cases based on the link density, establishing the global optimality for certain cases (of lowest- and highest-density as well as of density close to that of 
a UCM network) 
and the local optimality for all other cases. For convenience we define $\Phi(n,k) := \frac{2M(n,k)}{n(n-1)}$, the link density of 
any
network having $n$ nodes and $M(n,k)$ links.

When a given network has $m = n-1$ links, we have $\phi = 1-\Phi(n,n-1) = 2/n$, the lowest possible link density for a connected network. Any such network is necessarily a tree. We have $k_{n,m_c} = n-1$, which implies that in this case there is only one MCC network, namely, a star network. This network has $\lambda_2 = 1$ and is the unique maximizer of $\lambda_2$ among all tree networks~\cite{maas1987transportation}. This establishes the global optimality of the MCC network in this case.

For a network with the highest possible link density, $\phi = 1-\Phi(n,1) = 1$, we have $m_c = M(n,1) = 0$ and $k_{n,m_c} = 1$.
The only such network is the fully connected one, which has $\lambda_2 = n$. Thus, the MCC network is trivially the global optimum in this case.

For the next highest range of link density, $1-\Phi(n,2) \le \phi < 1$, we have $0 < m_c \le M(n,2)$ and $k_{n,m_c} = 2$.
In this case, the complement of the network has at least one link, which can be regarded as a subnetwork whose largest Laplacian eigenvalue is $2$.
Since the largest Laplacian eigenvalue of a network 
is larger
than or equal to the largest Laplacian eigenvalue of any of its subnetworks (i.e., induced subgraphs)~\cite{brouwer2012spectra}, we have $\lambda_n^c \ge 2$. This implies $\lambda_2 \le n-2 = n - k_{n,m_c}$, which is the value of $\lambda_2$ for any MCC network.
This establishes the global optimality in this case as well.

For $\phi$ in the next highest range, $1-\Phi(n,3) \le \phi < 1-\Phi(n,2)$, we have $M(n,2) < m_c \le M(n,3)$ and $k_{n,m_c} = 3$.
In this case, the complement of the network has three nodes connected by two links as a subnetwork; otherwise all links would be isolated from one another, which would imply $m_c \le M(n,2)$.
Since the largest Laplacian eigenvalue of this subnetwork is $3$, we have $\lambda_n^c \ge 3$, and hence $\lambda_2 \le n-3$.
This shows that the eigenvalue $\lambda_2 = n - k_{n,m_c} = n-3$ for the MCC networks is globally optimal. 

In the neighborhood of each 
$\phi=1-\Phi(n,k)$,
which corresponds
to a UCM network, we can also prove global optimality, which is summarized in the following proposition.
\medskip
\begin{proposition}\label{prop:UCM-global-opt}
Suppose that $n=k\ell$ for some positive integers $k$ and $\ell$, and that the link density $\phi$ satisfies
\begin{equation}\label{eq:Phi-global-opt}
	1-\Phi(n,k+1) \le \phi < 1-\Phi(n,k)+\frac{1}{n-1}.
\end{equation}
Then, the MCC networks achieve the global maximum value of $\lambda_2$.
\end{proposition}
\begin{proof}
Consider an MCC network having $n = k\ell$ nodes and $m$ links, with $\phi$ satisfying the inequalities in Eq.~\eqref{eq:Phi-global-opt}.
First note that the inequalities~\eqref{eq:Phi-global-opt} can be expressed in terms of the number of links $m_c$ in the complement as
\begin{equation}\label{eq:Mmc}
	M(n,k)-\frac{n}{2} < m_c \leq M(n,k+1),
\end{equation}
Since $n=k\ell$, we have $M(n,k)=M(k\ell,k)=\frac{1}{2}\ell k(k-1)$.
Also recall that $\lfloor 2m/n\rfloor$ is an upper bound for $\lambda_2$ for all networks having $n$ nodes and $m$ links 
[see Eq.~\eqref{bound} of the main text].
We divide the proof into three distinct cases.

\smallskip\noindent
{\it Case 1}: $m_c=M(n,k)$. In this case, the upper bound $\lfloor 2m/n \rfloor=\lfloor k(\ell-1)\rfloor=k(\ell-1)$.
By Proposition~\ref{prop:Mnk}, the complement must consist of $\ell$ isolated, fully connected clusters of $k$ nodes each.
Therefore, 
the network is UCM and hence is globally optimal by 
the result proved in 
Sec.~\ref{sec:opt-undirected}
[and we have $\lambda_2=k(\ell-1)$].

\smallskip\noindent
{\it Case 2}: $M(n,k)-\frac{n}{2}<m_c<M(n,k)$. In this case, $\frac{1}{2}n(n-1)-M(n,k)<m<\frac{1}{2}n(n-1)-M(n,k)+\frac{n}{2}$, which implies that $k(\ell-1)<\frac{2m}{n}<k(\ell-1)+1$. Thus, $\lfloor 2m/n\rfloor=k(\ell-1)$.
On the other hand, the definition of 
$k_{n,m_c}$
implies that $k_{n,m_c}\leq k$, and 
hence $k(\ell-1)\geq\lambda_2=k\ell-k_{n,m_c}\geq k(\ell-1)$. 
This shows that $\lambda_2=k(\ell-1)$, which coincides with the upper bound and therefore is the global maximum value.

\smallskip\noindent
{\it Case 3}: $M(n,k)<m_c\leq M(n,k+1)$. In this case $m=\frac{1}{2}n(n-1)-m_c<\frac{1}{2}n(n-1)-M(n,k)=\frac{1}{2}k^2\ell(\ell-1)$, therefore the upper bound $\lfloor 2m/n \rfloor\leq k(\ell-1)-1$. On the other hand, by the definition of 
$k_{n,m_c}$,
we have
$k_{n,m_c}=k+1$, and thus $\lambda_2=n-(k+1)=k(\ell-1)-1$, which equals the upper bound and hence is the global maximum value.
\end{proof}

We now prove the local optimality of the MCC networks for $1-\Phi(n,n-1) < \phi < 1-\Phi(n,3)$, which covers all the other cases.
For such $\phi$ we have $m > n-1$ and $k_{n,m_c} \ge 4$.
To avoid heavy notation, we define $k_* := k_{n,m_c}$ in the rest of this section.
Our proof consists of two steps. The first step is to prove that at least one of the following holds for the complement of any such MCC network:
\begin{enumerate}
\item At least one connected component has at least $\frac{(k_* - 1)(k_* - 2)}{2} + 2$ links (and thus is necessarily of size $k_*$). 
\item At least two connected components have at least $\frac{(k_* - 1)(k_* - 2)}{2} + 1$ links each (and thus are both of size $k_*$). 
\item One component is of size $k_*$ and has exactly $\frac{(k_* - 1)(k_* - 2)}{2} + 1$ links and all the other components are fully connected clusters of size $k_* - 1$.
\end{enumerate}
The second step is to show that, under each of these conditions, rewiring one link cannot increase $\lambda_2$.

We prove the first step by contradiction. Suppose that a given MCC network satisfies none of the conditions~1--3.
As shown in the proof of $\lambda^c_n \ge k_*$ in subsection~B, there exists a connected component in its complement which has at least $\frac{(k_* - 1)(k_* - 2)}{2} + 1$ links.
Since condition~2 does not hold, this component is unique, and each of the other components has at most $\frac{(k_* - 1)(k_* - 2)}{2}$ links.
Since condition~1 does not hold, this component actually has exactly $\frac{(k_* - 1)(k_* - 2)}{2} + 1$ links.
Also, it must be of size $k_*$, 
because a
smaller component 
cannot
have that many links.
Since all the other components have at most $\frac{(k_* - 1)(k_* - 2)}{2}$ links each, their links can be rewired to ensure that their sizes are all $\le k_*-1$.
We thus
see that the total number of links in 
these components 
[consisting of a total of ($n - k_*$) nodes] is at most $M(n - k_*, k_* - 1)$.
This gives
\begin{equation}\label{eqn:si-inequality}
m_c \le \frac{(k_* - 1)(k_* - 2)}{2} + 1 + M(n - k_*, k_* - 1).
\end{equation}

Our strategy now is to rewrite $M(n - k_*, k_* - 1)$ in Eq~\eqref{eqn:si-inequality} in terms of $M(n, k_* - 1)$ and use the relation $m_c > M(n, k_* - 1)$ to show that 
there is
a contradiction.
Let us use the notation $\ell(n,k) := \lfloor n/k \rfloor$ and $n_r(n,k) := n - k \cdot \ell(n,k)$ to make the $n$- and $k$-dependence explicit for $\ell$ and $n_r$ introduced 
in Appendix~\ref{si:sec:construction_MCC}.
We have five different cases:
\begin{itemize}
\item Case 1: $n \ge 2(k_*-1)$,\, $n_r(n, k_* - 1) \ge 1$.
\item Case 2: $n \ge 2(k_*-1)$,\, $n_r(n, k_* - 1) = 0,\, \ell(n, k_* - 1) \ge 3$.
\item Case 3: $n \ge 2(k_*-1)$,\, $n_r(n, k_* - 1) = 0,\, \ell(n, k_* - 1) = 2$.
\item Case 4: $n < 2(k_*-1)$,\, $n_r(n, k_* - 1) \ge 1$.
\item Case 5: $n < 2(k_*-1)$,\, $n_r(n, k_* - 1) = 0$.
\end{itemize}%
\noindent
Together they cover all possible cases, since we have $\ell(n, k_* - 1) \ge 2$ whenever $n \ge 2(k_*-1)$.
Below we prove each of these cases.

\medskip\noindent
\textbf{Case 1:} $n \ge 2(k_*-1)$,\, $n_r(n, k_* - 1) \ge 1$.
Since $\ell(n, k_* - 1) - 1 \ge 1$ and $n_r(n, k_* - 1) - 1 \ge 0$ in this case, we can write
\begin{equation}
\begin{split}
n - k_* &= [(k_* - 1)\cdot\ell(n, k_* - 1) + n_r(n, k_* - 1)] - k_*\\ 
&= (k_* - 1)[\ell(n, k_* - 1) - 1] + n_r(n, k_* - 1) - 1,
\end{split}
\end{equation}
giving
\begin{equation}
\begin{split}
\ell(n - k_*, k_* - 1) &= \ell(n, k_* - 1) - 1,\\
n_r(n - k_*, k_* - 1) &= n_r(n, k_* - 1) - 1,
\end{split}
\end{equation}
and hence $\ell(n - k_*, k_* - 1) \ge 1$.
Thus, Eq.~\eqref{eqn:si-inequality} becomes
\begin{equation}\label{eqn:S35}
\begin{split}
m_c &\le \frac{(k_* - 1)(k_* - 2)}{2} + 1 \\
 &\quad+ [\ell(n, k_* - 1) - 1]\cdot \frac{(k_* - 1)(k_* - 2)}{2}\\
 &\quad+ \frac{[n_r(n, k_* - 1) - 1][n_r(n, k_* - 1) - 2]}{2}\\
 &= M(n, k_* - 1) - n_r(n, k_* - 1) + 2\\
 &< m_c - n_r(n, k_* - 1) + 2,
\end{split}
\end{equation}
which implies $n_r(n, k_* - 1) < 2$, and hence $n_r(n, k_* - 1) = 1$ [since we assumed $n_r(n, k_* - 1) \ge 1$ at the outset].
In this case, 
Eq.~\eqref{eqn:S35} reads
\begin{equation}
m_c \le \ell(n, k_* - 1)\cdot\frac{(k_* - 1)(k_* - 2)}{2} + 1 < m_c + 1,
\end{equation}
and hence 
\begin{equation}
m_c = \ell(n, k_* - 1)\cdot\frac{(k_* - 1)(k_* - 2)}{2} + 1.
\end{equation}
This implies that condition~3 holds, contradicting with our initial assumption that it does not.\qed

\medskip\noindent
\textbf{Case 2:} $n \ge 2(k_*-1)$,\, $n_r(n, k_* - 1) = 0,\, \ell(n, k_* - 1) \ge 3$.
Since $\ell(n, k_* - 1) - 2 \ge 1$ in this case, and since we have $k_* - 2 \ge 0$ from the assumption $k_* \ge 4$, we can write
\begin{equation}
\begin{split}
n - k_* &=(k_* - 1)\cdot\ell(n, k_* - 1) - k_*\\ 
&= (k_* - 1)[\ell(n, k_* - 1) - 2] + k_* - 2,
\end{split}
\end{equation}
giving
\begin{equation}
\begin{split}
\ell(n - k_*, k_* - 1) &= \ell(n, k_* - 1) - 2,\\
n_r(n - k_*, k_* - 1) &= k_* - 2.
\end{split}
\end{equation}
Thus, Eq.~\eqref{eqn:si-inequality} becomes
\begin{equation}
\begin{split}
m_c &\le \frac{(k_* - 1)(k_* - 2)}{2} + 1 \\
 &\quad+ [\ell(n, k_* - 1) - 2]\cdot \frac{(k_* - 1)(k_* - 2)}{2}\\
 &\quad+ \frac{(k_* - 2)(k_* - 3)}{2}\\
 &= M(n, k_* - 1) - k_* + 3\\
 &< m_c - k_* + 3,
\end{split}
\end{equation}
which implies $k_* < 3$, violating the assumption $k_* \ge 4$.\qed

\medskip\noindent
\textbf{Case 3:} $n \ge 2(k_*-1)$,\, $n_r(n, k_* - 1) = 0,\, \ell(n, k_* - 1) = 2$.
In this case, we have
\begin{equation}
n = (k_* - 1)\cdot\ell(n, k_* - 1) + n_r(n, k_* - 1) = 2k_* - 2,
\end{equation}
which makes $M(n - k_*, k_* - 1)$ undefined, since the size of the largest connected component in the complement, $k_* - 1$, is larger than the number of nodes, $n - k_* = k_* - 2$.
However, noting that we have $(k_* - 2)$ nodes outside the component of size $k_*$ with $\frac{(k_* - 1)(k_* - 2)}{2} + 1$ links, we see that
\begin{equation}
\begin{split}
m_c &\le \frac{(k_* - 1)(k_* - 2)}{2} + 1 + \frac{(k_* - 2)(k_* - 3)}{2}\\
 &= M(n, k_* - 1) - k_* + 3\\
 &< m_c - k_* + 3,
\end{split}
\end{equation}
which again implies $k_* < 3$, violating the assumption $k_* \ge 4$.\qed

\medskip\noindent
\textbf{Case 4:} $n < 2(k_*-1)$,\, $n_r(n, k_* - 1) \ge 1$.
Since $n < 2(k_* - 1)$ in this case, we have 
\begin{equation}
\ell(n, k_* - 1) = \left\lfloor \frac{n}{k_* - 1} \right\rfloor 
\le \frac{n}{k_* - 1} < 2,
\end{equation}
implying $\ell(n, k_* - 1) = 1$, and hence
\begin{equation}\label{eqn:S44}
\begin{split}
n &= (k_* - 1)\cdot\ell(n, k_* - 1) + n_r(n, k_* - 1)\\
 &= k_* + n_r(n, k_* - 1) - 1.
\end{split}
\end{equation}
While the component of size $k_*$ has $\frac{(k_* - 1)(k_* - 2)}{2} + 1$ links, the remaining $[n_r(n, k_* - 1) - 1]$ nodes can only have at most 
$\frac{[n_r(n, k_* - 1) - 1][n_r(n, k_* - 1) - 2]}{2}$ 
links.
Thus,
\begin{equation}\label{eqn:S45}
\begin{split}
m_c &\le \frac{(k_* - 1)(k_* - 2)}{2} + 1 + 
\frac{[n_r(n, k_* - 1) - 1][n_r(n, k_* - 1) - 2]}{2}\\
 &= M(n, k_* - 1) - n_r(n, k_* - 1) + 2\\
 &< m_c - n_r(n, k_* - 1) + 2.
\end{split}
\end{equation}
This implies $n_r(n, k_* - 1) < 2$, and hence 
$n_r(n, k_* - 1) = 1$. Substituting this into Eqs.~\eqref{eqn:S44} and \eqref{eqn:S45}
leads to $m = n - 2$, violating the assumption $m > n-1$.\qed

\medskip\noindent
\textbf{Case 5:} $n < 2(k_*-1)$,\, $n_r(n, k_* - 1) = 0$.
As in Case 4, we have $\ell(n, k_* - 1) = 1$, and in this case we have $n = k_* - 1 < k_*$.
This is impossible because the connected component size cannot be larger than the total number of nodes.\qed

\medskip
Having dealt with all five cases, we have proved that at least one of the conditions 1--3 holds for any MCC network for which $1-\Phi(n,n-1) < \phi < 1-\Phi(n,3)$.
We now show that after rewiring one arbitrary link in any such network, we have $\lambda_2 \le n - k_*$, or equivalently, $\lambda^c_n \ge k_*$.

First note that regardless of which of the conditions 1--3 are satisfied, there is a connected component of size $k_*$ with at least $\frac{(k_* - 1)(k_* - 2)}{2} + 1$ links.
By Proposition~\ref{prop:n}, the largest Laplacian eigenvalue of this component is $k_*$, which is not affected by rewiring links within this or any other component.
Also, this eigenvalue can only increase by the rewiring of a link outside of this component (even if the link is then connected to this component or moves entirely into this component), since the subgraph formed by the links originally in this component remains unchanged and the largest Laplacian eigenvalue of the network is lower-bounded by that of a subgraph~\cite{brouwer2012spectra}.  Then, by Proposition~\ref{prop:union}, we have
$\lambda^c_n \ge k_*$.
It remains to show that $\lambda^c_n \ge k_*$ holds after rewiring a link from  
the inside to 
outside of this component.

If condition~1 holds, the component would still have at least $\frac{(k_* - 1)(k_* - 2)}{2} + 1$ links, so the largest Laplacian eigenvalue remains $k_*$, and we therefore have $\lambda^c_n \ge k_*$.
If condition~2 holds, at least one of the two components that have $\frac{(k_* - 1)(k_* - 2)}{2} + 1$ links would still have that many links after the rewiring.
We therefore have $\lambda^c_n \ge k_*$ also in this case.
If condition~3 holds, since all the other connected components are fully connected, we are guaranteed that the rewired link connects to one of these components, forming a subgraph having the largest Laplacian eigenvalue $k_*$, which implies $\lambda^c_n \ge k_*$ (because the largest Laplacian eigenvalue is lower-bounded by that of a subgraph~\cite{brouwer2012spectra}).
Thus, we have proved $\lambda^c_n \ge k_*$ under all three conditions, and therefore have established the local optimality of the MCC networks among all networks with the same $n$ and $m$.

\subsection{\texorpdfstring{Dependence of {\small$\boldsymbol{\lambda_2}$} on {\small$\boldsymbol{\phi}$}}{Dependence of lambda2 on phi}}

We first note that for an arbitrary positive integer $k$, we can use the definition of the function $\Phi(n,k)$ in subsection C and the relation $n_r = n - k\ell$ to write 
\begin{equation}
\Phi(n,k) = 1 - \frac{f\bigl(\tfrac{k}{n}\bigr)}{1-\tfrac{1}{n}},
\end{equation} 
where $f$ is a piecewise quadratic function defined on the interval $(0,1]$ by 
\begin{equation}
f(x) := \ell x[2 - (\ell+1)x],
\end{equation}
with positive integer $\ell$ determined uniquely by 
$\frac{1}{\ell+1} < x \le \frac{1}{\ell}$ 
for each $x$ in the interval $(0,1]$.
The function $f$ can be shown to be strictly decreasing in its domain, and hence is invertible.
From the assumption $\phi < 1$, which is equivalent to $k_{n,m_c} \ge 2$, we can characterize $k_{n,m_c}$ as the unique integer satisfying $1-\Phi(n,k_{n,m_c}) \le \phi < 1-\Phi(n,k_{n,m_c}-1)$.
This relation can be expressed using $f$ as
\begin{equation}
f\bigl(\tfrac{k_{n,m_c}}{n}\bigr) \le \bigl(1-\tfrac{1}{n}\bigr)\phi < f\bigl(\tfrac{k_{n,m_c}-1}{n}\bigr).
\end{equation}
Since $f$ is strictly decreasing and invertible, we can apply $f^{-1}$ (and reverse the direction of the inequalities) to obtain
\begin{equation}
\frac{k_{n,m_c}-1}{n} < x_n \le \frac{k_{n,m_c}}{n},
\end{equation}
where $x_n$ is defined uniquely by $f(x_n) = \bigl( 1-\frac{1}{n} \bigr) \phi$.
From this and $\lambda_2 = n - k_{n,m_c}$ (established in subsection A above), it follows that $\lambda_2 \le C_{\ell,n}(\phi) \cdot n < \lambda_2 + 1$, which is equivalent to Eq.~\eqref{eqn:lambda2} of the main text, with the definition $C_{\ell,n}(\phi) := 1-x_n$.
Since $f$ is a quadratic function on each interval $\bigl(\frac{1}{\ell+1}, \frac{1}{\ell}\bigr]$, we can explicitly invert $f(x_n) = \bigl( 1-\frac{1}{n} \bigr) \phi$, which gives the expression for $C_{\ell,n}(\phi)$ in 
Eq.~\eqref{eqn:C_ell_n}.


\end{bibunit}

\vspace{10mm}

\section*{Supplemental Figures}

\renewcommand{\thefigure}{S\arabic{figure}}
\setcounter{figure}{0}

\begin{figure*}[ht]
\begin{center}
\includegraphics[width=2.9in]{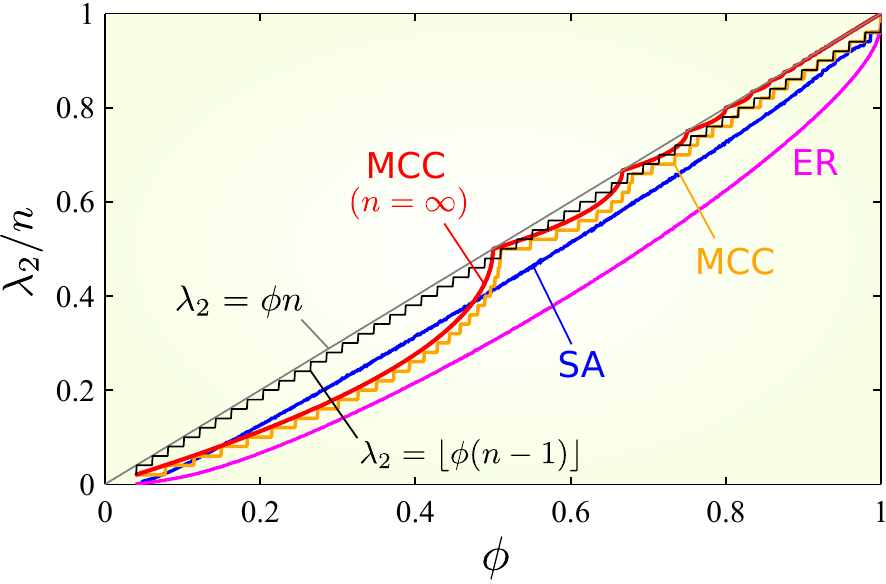}
\rule{22pt}{0pt}
\end{center}\vspace{-20pt}
\caption{\label{fig:comparisonA}
The Laplacian eigenvalue $\lambda_2$ normalized by the network size $n$ as a function of the link density $\phi$ for the same networks used in Fig.~\ref{fig:comparison}.
Color coding for the curves is the same as in Fig.~\ref{fig:comparison}.}
\end{figure*}

\vspace{10mm}

\begin{figure*}[ht]
\begin{center}
\includegraphics[width=0.71\textwidth]{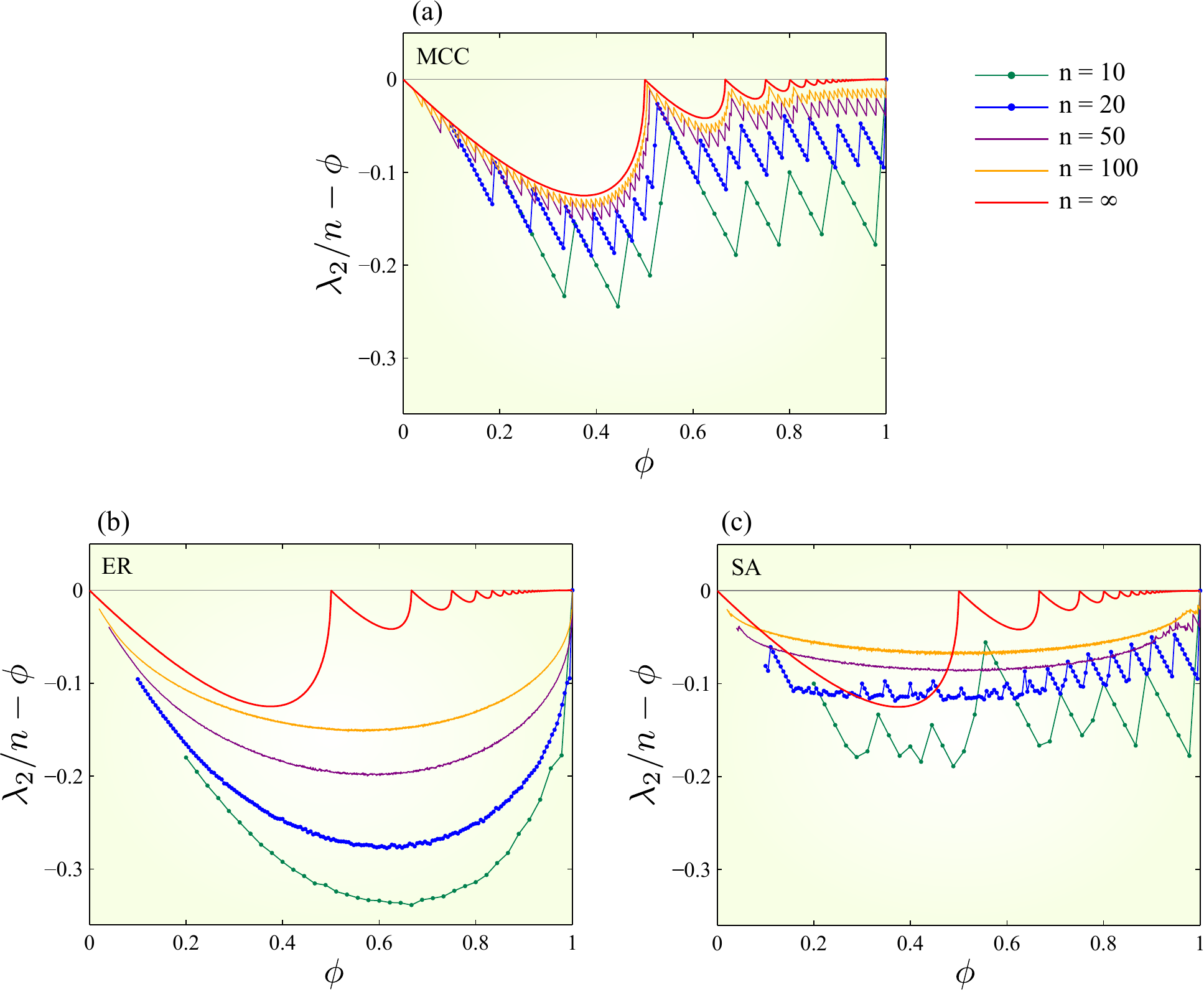}
\rule{22pt}{0pt}
\end{center}\vspace{-16pt}
\caption{\label{fig:convergence}
Convergence of $\lambda_2/n$ toward the asymptotic upper bound $\phi$ as network size $n$ increases.
For each value of $n$, we plot $\lambda_2/n - \phi$ as a function of the link density $\phi$ for the MCC networks (a), ER random networks (b), and networks obtained by simulated annealing~(c).
For reference, a red curve in each panel shows the values for the MCC networks in the limit $n \to \infty$, given by Eq.~\eqref{eqn:lambda2inf} of the main text.
For the ER random networks, 
each data point is an average
over $1{,}000$ realizations.
For SA with $n=10$ and $20$, we took the networks that maximize $\lambda_2$ over $1{,}000$ independent SA runs starting from randomly chosen initial networks, while we used $20$ runs for $n=50$ and a single run for $n=100$.
The raggedness of the $n=10$ curve for the networks obtained by SA is likely to reflect the true nature of the maximum $\lambda_2$ as a function of $\phi$, since a majority of these networks do achieve the upper bound $\lambda_2 = \lfloor \phi(n-1) \rfloor$, and hence have the true maximum $\lambda_2$.}
\end{figure*}

\begin{figure*}[ht]
\begin{center}
\includegraphics[width=2.9in]{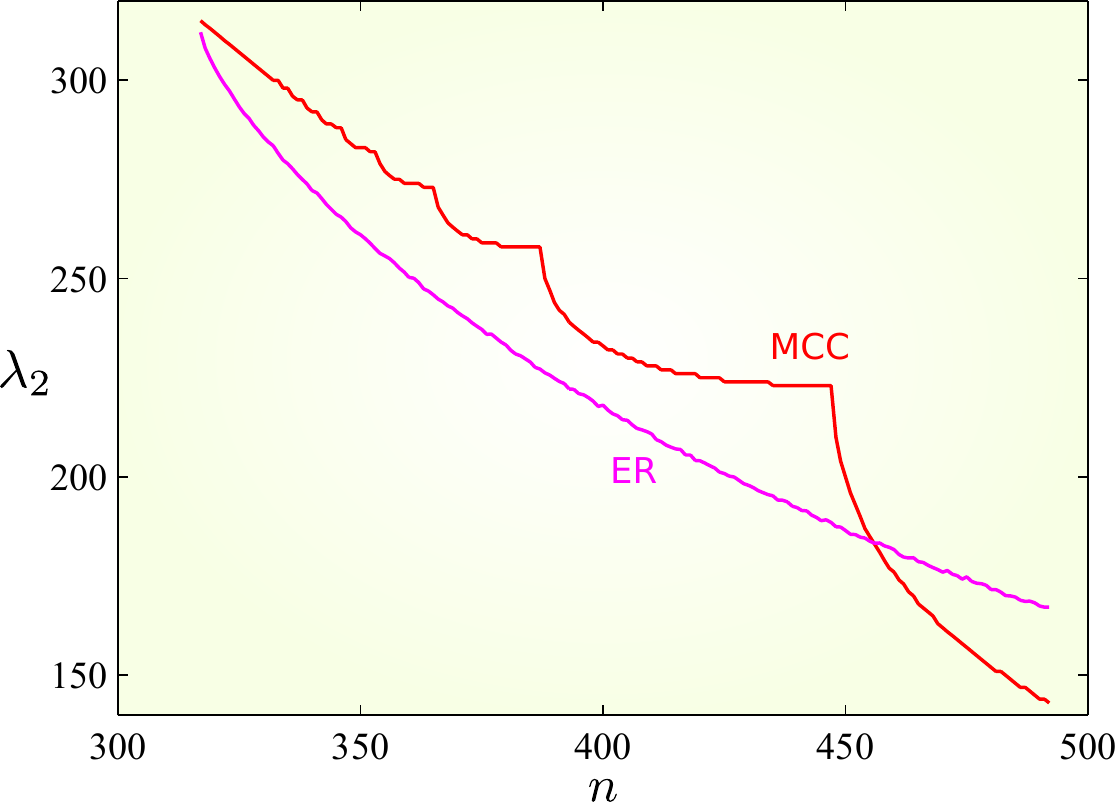}
\rule{22pt}{0pt}
\end{center}\vspace{-12pt}
\caption{\label{fig:n-dep-fixed-m}
Dependence of $\lambda_2$ on network size $n$ for fixed $m=5\times 10^4$.
The red curve indicates $\lambda_2$ for the MCC networks, given in Eq.~\eqref{eqn:lambda2}, considered as a function of $n$.
Note that, in addition to the explicit dependence on $n$, the function $C_{\ell,n}(\phi)$ in Eq.~\eqref{eqn:C_ell_n} depends on $n$ also through $\phi = 2m/[n(n-1)]$.
The non-smoothness of the curve between the singularity points is not a numerical artifact but instead reflects the fact that $\lambda_2$ always takes an integer value according to Eq.~\eqref{eqn:lambda2}.
Each point on the magenta curve is the average of $\lambda_2$ over $100$ realizations of the ER random networks.
}
\end{figure*}


\begin{thebibliography}{10}

\bibitem{Newman2010}
M.~E.~J. Newman,
{\it Networks: An Introduction} 
(Oxford University Press, Oxford, 2010).

\bibitem{Chen:2015}
G. Chen, X. Wang, and X. Li,
{\it Fundamentals of Complex Networks: Models, Structures and Dynamics}
(John Wiley \& Sons, Singapore, 2015).

\bibitem{Strogatz:2001il}
S.~H. Strogatz, 
{\it Exploring complex networks}, 
Nature {\bf 410}, 268 (2001).

\bibitem{barrat2008dynamical}
A. Barrat, M. Barth{\'e}lemy, and A. Vespignani,
{\it Dynamical Processes on Complex Networks} 
(Cambridge University Press, Cambridge, 2008).

\bibitem{Porter:2016}
M.~A.~Porter and J.~P.~Gleeson,
{\it Dynamical Systems on Networks}
(Springer International Publishing, Switzerland, 2016).

\bibitem{Nishikawa:2003xr}
T. Nishikawa,  A.~E. Motter, Y.~C. Lai, and F.~C. Hoppensteadt, 
{\it Heterogeneity in oscillator networks{\rm :} Are smaller worlds easier to synchronize?}
Phys. Rev. Lett. {\bf 91}, 014101 (2003).

\bibitem{Restrepo2004}
J.~G.~Restrepo, E.~Ott, and B.~R.~Hunt,
{\it Spatial patterns of desynchronization bursts in networks},
Phys. Rev. E {\bf 69}, 066215 (2004).

\bibitem{PhysRevLett.93.254101}
H. Kori and A.~S. Mikhailov, 
{\it Entrainment of randomly coupled oscillator networks by a pacemaker}, 
Phys. Rev. Lett. {\bf 93}, 254101 (2004).

\bibitem{Belykh:2006qr}
I.~Belykh, V.~Belykh, and M.~Hasler,
{\it Generalized connection graph method for synchronization in asymmetrical networks},
Physica D {\bf 224}, 42 (2006).

\bibitem{Restrepo2006}
J.~G.~Restrepo, E.~Ott, and B.~R.~Hunt,
{\it Emergence of coherence in complex networks of heterogeneous dynamical systems},
Phys. Rev. Lett. {\bf 96}, 254103 (2006).

\bibitem{Wiley:2006fk}
D.~A. Wiley, S.~H. Strogatz, and M. Girvan, 
{\it The size of the sync basin}, 
Chaos {\bf 16}, 015103 (2006).

\bibitem{PhysRevE.74.066115}
H. Kori and A.~S. Mikhailov, 
{\it Strong effects of network architecture in the entrainment of coupled oscillator systems}, 
Phys. Rev. E {\bf 74}, 066115 (2006).

\bibitem{Arenas2008}
A. Arenas, A. D\'{i}az-Guilera, J. Kurths, Y. Moreno, and C. Zhou,
{\it Synchronization in complex networks}, 
Phys. Rep. {\bf 469}, 93 (2008).

\bibitem{PhysRevLett.110.174102}
V. Nicosia, M. Valencia, M. Chavez, A. D\'{i}az-Guilera, and V. Latora, 
{\it Remote synchronization reveals network symmetries and functional modules}, 
Phys. Rev. Lett. {\bf 110}, 174102 (2013).

\bibitem{Pecora:2014zr}
L.~M. Pecora, F. Sorrentino, A.~M. Hagerstrom, T.~E. Murphy, and R. Roy, 
{\it Cluster synchronization and isolated desynchronization in complex networks with symmetries}, 
Nat. Commun. {\bf 5}, 4079 (2014).

\bibitem{Skardal2014}
P.~S. Skardal, D. Taylor, and J. Sun, 
{\it Optimal synchronization of complex networks}, 
Phys. Rev. Lett. {\bf 113}, 144101 (2014).

\bibitem{Colizza:2007uq}
V. Colizza, R. Pastor-Satorras, and A. Vespignani, 
{\it Reaction-diffusion processes and metapopulation models in heterogeneous networks}, 
Nat. Phys. {\bf 3}, 276 (2007).

\bibitem{PhysRevLett.110.028701}
S. G\'omez, A. D\'{i}az-Guilera, J. G\'omez-Garde\~nes, C.~J. P\'erez-Vicente, Y. Moreno, and A. Arenas, 
{\it Diffusion dynamics on multiplex networks}, 
Phys. Rev. Lett. {\bf 110}, 028701 (2013).

\bibitem{Youssef:2013}
M. Youssef, Y. Khorramzadeh, and S. Eubank,
{\it Network reliability{\rm :} The effect of local network structure on diffusive processes},
Phys. Rev. E {\bf 88}, 052810 (2013).

\bibitem{PhysRevE.89.020801}
S. Hata, H. Nakao, and A.~S. Mikhailov, 
{\it Advection of passive particles over flow networks}, 
Phys. Rev. E {\bf 89}, 020801 (2014).

\bibitem{Pomerance:2009fk}
A. Pomerance, E. Ott, M. Girvan, and W. Losert, 
{\it The effect of network topology on the stability of discrete state models of genetic control}, 
Proc. Natl. Acad. Sci. USA {\bf 106}, 8209 (2009).

\bibitem{Bunimovich2012}
L.~A. Bunimovich and B.~Z. Webb, 
{\it Isospectral graph transformations, spectral equivalence, and global stability of dynamical networks}, 
Nonlinearity {\bf 25}, 211 (2012).

\bibitem{PhysRevE.75.046103}
F. Sorrentino, M. di~Bernardo, F. Garofalo, and G. Chen, 
{\it Controllability of complex networks via pinning}, 
Phys. Rev. E {\bf 75}, 046103 (2007).

\bibitem{Whalen:2015}
A.~J. Whalen, S.~N. Brennan, T.~D. Sauer, and S.~J. Schiff,
{\it Observability and controllability of nonlinear networks{\rm :} The role of symmetry},
Phys. Rev. X {\bf 5}, 011005 (2015).

\bibitem{PhysRevE.75.036105}
S. Sreenivasan, R. Cohen, E. L\'opez, Z. Toroczkai, and H.~E. Stanley, 
{\it Structural bottlenecks for communication in networks}, 
Phys. Rev. E {\bf 75}, 036105 (2007).

\bibitem{Sun2015}
J. Sun, D. Taylor, and E.~M. Bollt, 
{\it Causal network inference by optimal causation entropy}, 
SIAM J. Appl. Dyn. Syst. {\bf 14}, 73 (2015).

\bibitem{PhysRevE.87.032909}
C. Fu, Z. Deng, L. Huang, and X. Wang, 
{\it Topological control of synchronous patterns in systems of networked chaotic oscillators}, 
Phys. Rev. E {\bf 87}, 032909 (2013).

\bibitem{PhysRevE.81.036101}
Y. Qian, X. Huang, G. Hu, and X. Liao, 
{\it Structure and control of self-sustained target waves in excitable small-world networks}, 
Phys. Rev. E {\bf 81}, 036101 (2010).

\bibitem{PhysRevLett.96.208701}
T. Gross, C.~J.~D. D'Lima, and B. Blasius, 
{\it Epidemic dynamics on an adaptive network}, 
Phys. Rev. Lett. {\bf 96}, 208701 (2006).

\bibitem{RisauGusman200952}
S. Risau-Gusman and D.~H. Zanette, 
{\it Contact switching as a control strategy for epidemic outbreaks},
J. Theor. Biol. {\bf 257}, 52 (2009).

\bibitem{Hagberg:2008wd}
A. Hagberg and D.~A. Schult, 
{\it Rewiring networks for synchronization}, 
Chaos {\bf 18}, 037105 (2008).

\bibitem{Nishikawa:2010fk}
T. Nishikawa and A.~E. Motter, 
{\it Network synchronization landscape reveals compensatory structures, quantization, and the positive effect of negative interactions}, 
Proc. Natl. Acad. Sci. USA {\bf 107}, 10342 (2010).

\bibitem{Watanabe2010lsg}
T. Watanabe and N. Masuda, 
{\it Enhancing the spectral gap of networks by node removal}, 
Phys. Rev. E {\bf 82}, 046102 (2010).

\bibitem{Hart:2015}
J.~D. Hart, J.~P. Pade, T. Pereira, T.~E. Murphy, and R. Roy,
{\it Adding connections can hinder network synchronization of time-delayed oscillators}, 
Phys. Rev. E {\bf 92}, 022804 (2015). 

\bibitem{PhysRevLett.110.064106}
N.~A.~M. Ara\'ujo, H. Seybold, R.~M. Baram, H.~J. Herrmann, and J.~S. Andrade, 
{\it Optimal synchronizability of bearings}, 
Phys. Rev. Lett. {\bf 110}, 064106 (2013).

\bibitem{PhysRevLett.107.034102}
B. Ravoori, A.~B. Cohen, J. Sun, A.~E. Motter, T.~E. Murphy, and R. Roy, 
{\it Robustness of optimal synchronization in real networks}, 
Phys. Rev. Lett. {\bf 107}, 034102 (2011).

\bibitem{PhysRevLett.108.214101}
M. Nixon et~al., 
{\it Controlling synchronization in large laser networks}, 
Phys. Rev. Lett. {\bf 108}, 214101 (2012).

\bibitem{Dobson:2007}
I. Dobson, B.~A. Carreras, V.~E. Lynch, and D.~E. Newman, 
{\it Complex systems analysis of series of blackouts{\rm :} Cascading failure, critical points, and self-organization}, 
Chaos {\bf 17}, 026103 (2007).

\bibitem{Yamins10062014}
D.~L.~K. Yamins et~al., 
{\it Performance-optimized hierarchical models predict neural responses in higher visual cortex}, 
Proc. Natl. Acad. Sci. USA {\bf 111}, 8619 (2014).

\bibitem{Buzsaki:2004uq}
G. Buzs{\'a}ki, C. Geisler, D.~A. Henze, and X.~J. Wang, 
{\it Interneuron diversity series{\rm :} Circuit complexity and axon wiring economy of cortical interneurons}, 
Trends Neurosci. {\bf 27}, 186 (2004).

\bibitem{Tononi24051994}
G. Tononi, O. Sporns, and G.~M. Edelman, 
{\it A measure for brain complexity{\rm :} relating functional segregation and integration in the nervous system}, 
Proc. Natl. Acad. Sci. USA {\bf 91}, 5033 (1994).

\bibitem{Bullmore:2009fk}
E. Bullmore and O. Sporns, 
{\it Complex brain networks{\rm :} Graph theoretical analysis of structural and functional systems}, 
Nat. Rev. Neurosci. {\bf 10}, 186 (2009).

\bibitem{10.1371/journal.pcbi.0030017}
S. Achard and E. Bullmore, 
{\it Efficiency and cost of economical brain functional networks}, 
PLoS Comput. Biol. {\bf 3}, e17 (2007).

\bibitem{Golub:2013}
G.~H. Golub and C.~F. Van Loan,
{\it Matrix Computations}
(The Johns Hopkins University Press, Baltimore, 2013).

\bibitem{Eslami:1994}
M. Eslami, {\it Theory of Sensitivity in Dynamic Systems} (Springer-Verlag, Berlin, 1994).

\bibitem{Chui:1997}
C. K. Chui and G. Chen, {\it Discrete $H^{\infty}$ Optimization} (Springer-Verlag, Berlin, 1997).

\bibitem{Ottino-Loffler:2016}
B. Ottino-L\"offler and S.~H. Strogatz,
Comparing the locking threshold for rings and chains of oscillators,
Phys. Rev. E {\bf 94}, 062203 (2016).

\bibitem{MacArthur:2009}
B.~D. MacArthur and R.~J. S\'{a}nchez-Garc\'{i}a, 
{\it Spectral characteristics of network redundancy},
Phys. Rev. E {\bf 80}, 026117 (2009).

\bibitem{Pecora:1998zp}
L.~M. Pecora and T.~L. Carroll, 
{\it Master stability functions for synchronized coupled systems}, 
Phys. Rev. Lett. {\bf 80}, 2109 (1998).

\bibitem{4140748}
W. Ren, R. Beard, and E. Atkins, 
{\it Information consensus in multivehicle cooperative control}, 
IEEE Control Syst. Mag. {\bf 27}, 71 (2007).

\bibitem{4700861}
P. Yang, R. Freeman, and K. Lynch, 
{\it Multi-agent coordination by decentralized estimation and control}, 
IEEE Trans. Automat. Control {\bf 53}, 2480 (2008).

\bibitem{Nakao:2010fk}
H. Nakao and A.~S. Mikhailov, 
{\it Turing patterns in network-organized activator-inhibitor systems}, 
Nat. Phys. {\bf 6}, 544 (2010).

\bibitem{maas1987transportation}
C. Maas, 
{\it Transportation in graphs and the admittance spectrum}, 
Discrete Appl. Math. {\bf 16}, 31 (1987).

\bibitem{Motter:2013fk}
A.~E. Motter, S.~A. Myers, M. Anghel, and T. Nishikawa, 
{\it Spontaneous synchrony in power-grid networks}, 
Nat. Phys. {\bf 9}, 191 (2013).

\bibitem{kuramoto1984chemical}
Y. Kuramoto, 
{\it Chemical Oscillations, Waves, and Turbulence}
(Springer-Verlag, Berlin, 1984).

\bibitem{PhysRevE.61.5080}
K.~S. Fink, G. Johnson, T. Carroll, D. Mar, and L. Pecora,
{\it Three coupled oscillators as a universal probe of synchronization stability in coupled oscillator arrays}, 
Phys. Rev. E {\bf 61}, 5080 (2000).

\bibitem{Nishikawa:2006fk}
T. Nishikawa and A.~E. Motter, 
{\it Synchronization is optimal in non-diagonalizable networks}, 
Phys. Rev. E {\bf 73}, 065106 (2006).

\bibitem{sm}
Supplemental Material contains 
the description of several example systems and processes 
(Sec.~\ref{si:sec:example-systems}),
proofs of key properties of the 
MCC networks (Sec.~\ref{si:sec:prop}), 
and
supplemental figures (Figs.~\ref{fig:comparisonA}--\ref{fig:n-dep-fixed-m}).

\bibitem{Sun:2009hc}
J. Sun, E.~M. Bollt, and T. Nishikawa, 
{\it Master stability functions for coupled nearly identical dynamical systems}, 
Europhys. Lett. {\bf 85}, 60011 (2009).

\bibitem{PhysRevLett.95.188701}
L. Donetti, P.~I. Hurtado, and M.~A. Mu\~noz, 
{\it Entangled networks, synchronization, and optimal network topology}, 
Phys. Rev. Lett. {\bf 95}, 188701 (2005).

\bibitem{Nishikawa:2006kx}
T. Nishikawa and A.~E. Motter, 
{\it Maximum performance at minimum cost in network synchronization}, 
Physica D {\bf 224}, 77 (2006).

\bibitem{Wang:2007kx}
B. Wang, T. Zhou, Z.~L. Xiu, and B.~J. Kim, 
{\it Optimal synchronizability of networks}, 
Eur. Phys. J. B {\bf 60}, 89 (2007).

\bibitem{PhysRevE.81.025202}
M. Brede, 
{\it Optimal synchronization in space}, 
Phys. Rev. E {\bf 81}, 025202 (2010).

\bibitem{6561538}
Q. Xinyun, W. Lifu, G. Yuan, and W. Yaping, 
{\it The optimal synchronizability of a class network}, 
in {\it 25th Chinese Control and Decision Conference}
(IEEE, Guiyang, China, 2013), p.~3414.

\bibitem{Alon:1986uq}
N. Alon, 
{\it Eigenvalues and expanders}, 
Combinatorica {\bf 6}, 83 (1986).

\bibitem{Lubotzky:1988fk}
A. Lubotzky, R. Phillips, and P. Sarnak,
{\it Ramanujan graphs}, 
Combinatorica {\bf 8}, 261 (1988).

\bibitem{Friedman:1989:SER:73007.73063}
J. Friedman, J. Kahn, and E. Szemer{\'e}di,
{\it On the second eigenvalue of random regular graphs}, 
in {\it STOC '89 Proceedings of the Twenty-first Annual ACM Symposium on Theory of Computing},
edited by D.~S. Johnson (ACM, New York, 1989), p.~587.

\bibitem{Kolokolnikov:2015}
T. Kolokolnikov,
{\it Maximizing algebraic connectivity for certain families of graphs},
Linear Algebra Appl. {\bf 471}, 122 (2015).

\bibitem{MR2159259}
R. Diestel, 
{\it Graph Theory}
(Springer-Verlag, Berlin, 2005).

\bibitem{Duan:2008fk}
Z. Duan, C. Liu, and G. Chen, 
{\it Network synchronizability analysis{\rm :} The theory of subgraphs and complementary graphs}, 
Physica D {\bf 237}, 1006 (2008).

\bibitem{brouwer2012spectra}
A. Brouwer, 
{\it Spectra of Graphs}
(Springer, New York, 2012).

\bibitem{Milo:2002}
R.~Milo, S.~Shen-Orr, S.~Itzkovitz, N.~Kashtan, D.~Chklovskii, and U.~Alon,
{\it Network motifs{\rm :} Simple building blocks of complex networks},
Science {\bf 298}, 824 (2002).

\bibitem{Sporns:2004}
O.~Sporns and R.~K\"{o}tter,
{\it Motifs in brain networks},
PLoS Biol. {\bf 2}, e369 (2004).

\bibitem{Kaluza:2007}
P.~Kaluza, M.~Ipsen, M.~Vingron, and A.~S.~Mikhailov,
{\it Design and statistical properties of robust functional networks{\rm :} A model study of biological signal transduction},
Phys. Rev. E {\bf 75}, 015101 (2007).

\bibitem{software}
\url{https://github.com/tnishi0/mcc-networks}

\bibitem{Achlioptas:2009ys}
D. Achlioptas, R.~M. D'Souza, and J. Spencer, 
{\it Explosive percolation in random networks}, 
Science {\bf 323}, 1453 (2009).

\bibitem{PhysRevLett.105.255701}
R.~A. da~Costa, S.~N. Dorogovtsev, A.~V. Goltsev, and J.~F.~F. Mendes,
{\it Explosive percolation transition is actually continuous}, 
Phys. Rev. Lett. {\bf 105}, 255701 (2010).

\bibitem{Riordan:2011kx}
O. Riordan and L. Warnke, 
{\it Explosive percolation is continuous}, 
Science {\bf 333}, 322 (2011).

\bibitem{Nagler:2011aa}
J. Nagler, A. Levina, and M. Timme, 
{\it Impact of single links in competitive percolation}, 
Nat. Phys. {\bf 7}, 265 (2011).

\bibitem{DSouza:2015fk}
R.~M. D'Souza and J. Nagler, 
{\it Anomalous critical and supercritical phenomena in explosive percolation}, 
Nat. Phys. {\bf 11}, 531 (2015).

\bibitem{Rozenfeld:2010uq}
H. Rozenfeld, L. Gallos, and H. Makse, 
{\it Explosive percolation in the human protein homology network}, 
Eur. Phys. J. B {\bf 75}, 305 (2010).

\bibitem{PhysRevLett.112.155701}
W. Chen, M. Schr\"oder, R.~M. D'Souza, D. Sornette, and J. Nagler, 
{\it Microtransition cascades to percolation}, 
Phys. Rev. Lett. {\bf 112}, 155701 (2014).

\bibitem{comment}
We note that directed networks can exhibit sensitive dependence on the network structure with respect to different objective functions (see, e.g., Ref.~\cite{Nishikawa:2010fk}).

\bibitem{Showalter2015}
A.~F. Taylor, M.~R. Tinsley, and K. Showalter, 
{\it Insights into collective cell behaviour from populations of coupled chemical oscillators}, 
Phys. Chem. Chem. Phys. {\bf 17}, 20047 (2015).

\bibitem{Kiss2002}
I.~Z. Kiss, Y. Zhai, and J.~L. Hudson, 
{\it Emerging coherence in a population of chemical oscillators}, 
Science {\bf 296}, 1676 (2002).

\bibitem{MATLAB:2010}
MATLAB, version 7.10.0 (R2010a), The MathWorks Inc., 2010.

\bibitem{golub2013matrix}
G. Golub and C. Van~Loan, 
{\it Matrix Computations}
(Johns Hopkins University Press, Baltimore, 2013).

\bibitem{Brualdi:2008}
R. A. Brualdi and D. Cvetkovic,
{\it A Combinatorial Approach to Matrix Theory and Its Applications}
(CRC Press, Boca Raton, 2008).

\bibitem{trefethen2005spectra}
L. Trefethen and M. Embree, 
{\it Spectra and Pseudospectra{\rm :} The Behavior of Nonnormal Matrices and Operators}
(Princeton University Press, Princeton, 2005).

\bibitem{Milanese2010:kgd}
A. Milanese, J. Sun, and T. Nishikawa, 
{\it Approximating spectral impact of structural perturbations in large networks}, 
Phys. Rev. E {\bf 81}, 046112 (2010).

\bibitem{kitano2004biological}
H. Kitano, 
{\it Biological robustness}, 
Nat. Rev. Genet. {\bf 5}, 826 (2004).

\bibitem{keio_collection}
T. Baba et al., 
{\it Construction of Escherichia coli K-12 in-frame, single-gene knockout mutants{\rm :} The Keio collection}, 
Mol. Syst. Biol. {\bf 2}, 2006.0008 (2006).

\bibitem{Stelling:2004}
J. Stelling, U. Sauer, Z. Szallasi, F.~J. Doyle, III, and J. Doyle,
{\it Robustness of cellular functions},
Cell {\bf 118}, 675 (2004). 

\bibitem{Yang:2015}
L. Yang, S. Srinivasan, R. Mahadevan, and W.~R. Cluett,
{\it Characterizing metabolic pathway diversification in the context of perturbation size},
Metab. Eng. {\bf 28}, 114 (2015).

\bibitem{Glass2006}
J.~I. Glass, et al., 
{\it Essential genes of a minimal bacterium}, 
Proc. Natl. Acad. Sci. USA {\bf 103}, 425 (2006).

\bibitem{comment2}
A complementary example of interplay between selection optimization and robustness is described in Ref.~\cite{Parvinen:2013lr}, which shows that optimization can drive a population to extinction through a global bifurcation in the structure of the state space. 

\bibitem{Parvinen:2013lr}
K. Parvinen and U. Dieckmann, 
{\it Self-extinction through optimizing selection}, 
J. Theor. Biol. {\bf 333}, 1 (2013).

\bibitem{Bar-Yam30032004}
Y. Bar-Yam and I.~R. Epstein, 
{\it Response of complex networks to stimuli}, 
Proc. Natl. Acad. Sci. USA {\bf 101}, 4341 (2004).

\bibitem{carlson1999highly}
J.~M. Carlson and J. Doyle, 
{\it Highly optimized tolerance{\rm :} A mechanism for power laws in designed systems}, 
Phys. Rev. E {\bf 60}, 1412 (1999).

\bibitem{OttChaosBook}
E. Ott, 
{\it Chaos in Dynamical Systems}, 2nd ed.
(Cambridge University Press, Cambridge, 2002).

\bibitem{Lorenz63}
E.~N. Lorenz, 
{\it Deterministic nonperiodic flow}, 
J. Atmos. Sci. {\bf 20}, 130 (1963).

\bibitem{Babtie1:2014}
A.~C. Babtie, P. Kirk, and M.~P.~H. Stumpf,
{\it Topological sensitivity analysis for systems biology},
Proc. Natl. Acad. Sci. USA {\bf 111}, 18507 (2014).

\bibitem{Filippov1971}
A.~F. Filippov, 
{\it A short proof of the theorem on reduction of a matrix to Jordan form},
Vestnik Moskov. Univ. Ser. I Mat. Meh. {\bf 26}, 18 (1971).

\bibitem{Horn1985}
R.~A. Horn and C.~R. Johnson,
{\it Matrix Analysis}
(Cambridge University Press, Cambridge, 
1985).

\end{thebibliography}

\begin{thebibliography}{10}

\bibitem{Anderson:1986fk}
P.~M. Anderson and A.~A. Fouad, {\it Power System Control and Stability}, 2nd ed. (IEEE Press, Piscataway, 2003).

\bibitem{Motter:2013fk}
A.~E. Motter, S.~A. Myers, M. Anghel, and T. Nishikawa, 
{\it Spontaneous synchrony in power-grid networks}, 
Nat. Phys. {\bf 9}, 191 (2013).

\bibitem{PhysRevE.71.016215}
B.~R. Trees, V. Saranathan, and D. Stroud, 
{\it Synchronization in disordered Josephson junction arrays{\rm :} Small-world connections and the Kuramoto model}, 
Phys. Rev. E {\bf 71}, 016215 (2005).

\bibitem{PhysRevLett.110.218701}
P. Ji, T.~K.~D. Peron, P.~J. Menck, F.~A. Rodrigues, and J. Kurths,
{\it Cluster explosive synchronization in complex networks}, 
Phys. Rev. Lett. {\bf 110}, 218701 (2013).

\bibitem{Filatrella2008}
G. Filatrella, A.~H. Nielsen, and N.~F. Pedersen, 
{\it Analysis of a power grid using a Kuramoto-like model}, 
Eur. Phys. J. B {\bf 61}, 485 (2008).

\bibitem{fd-fb:09z}
F. D\"orfler and F. Bullo, 
{\it Synchronization and transient stability in power networks and non-uniform Kuramoto oscillators}, 
SIAM J. Control. Optim. {\bf 50}, 1616 (2012).

\bibitem{hoppensteadt1997weakly}
F. Hoppensteadt and E. Izhikevich, 
{\it Weakly Connected Neural Networks}
(Springer-Verlag, 1997).

\bibitem{Pecora:1998zp}
L.~M. Pecora and T.~L. Carroll, 
{\it Master stability functions for synchronized coupled systems}, 
Phys. Rev. Lett. {\bf 80}, 2109 (1998).

\bibitem{PhysRevLett.107.034102}
B. Ravoori, A.~B. Cohen, J. Sun, A.~E. Motter, T.~E. Murphy, and R. Roy,
{\it Robustness of optimal synchronization in real networks}, 
Phys. Rev. Lett. {\bf 107}, 034102 (2011).

\bibitem{Nakao:2010fk}
H. Nakao and A.~S. Mikhailov,
{\it Turing patterns in network-organized activator-inhibitor systems}, 
Nat. Phys. {\bf 6}, 544 (2010).

\bibitem{4140748}
W. Ren, R. Beard, and E. Atkins, 
{\it Information consensus in multivehicle cooperative control}, 
IEEE Control Syst. Magazine {\bf 27}, 71 (2007).

\bibitem{maas1987transportation}
C. Maas, 
{\it Transportation in graphs and the admittance spectrum}, 
Discrete Appl. Math. {\bf 16}, 31 (1987).

\bibitem{brouwer2012spectra}
A. Brouwer, {\it Spectra of Graphs} (Springer, New York, 2012).

\end{thebibliography}
\end{document}